\documentclass[11pt]{article}

\usepackage[utf8]{inputenc}

\usepackage{amsmath}
\usepackage{amsthm}
\usepackage{amssymb}
\usepackage{algorithm}
\usepackage{color}
\usepackage{xcolor}
\usepackage[english]{babel}
\usepackage{bbm}

\usepackage{graphicx}
\usepackage{caption}
\usepackage{subcaption}
\usepackage{mathtools}
\usepackage{enumitem}

\usepackage[square,sort,comma,numbers]{natbib}
\usepackage{wrapfig,epsfig}
\usepackage{psfrag}
\usepackage{epstopdf}
\usepackage{url}
\usepackage{graphicx}
\usepackage{color}
\usepackage{epstopdf}
\usepackage{algpseudocode}
\usepackage[hidelinks,pdfencoding=auto,psdextra]{hyperref}
\hypersetup{colorlinks=false}
\hypersetup{%draft, % TODO: remove "draft" option before final version after making sure no hyperlinks are broken due to pagebreaks.
	colorlinks,
	linkcolor={red!40!gray},
	citecolor={blue!40!gray},
	urlcolor={blue!70!gray}
}
\usepackage[margin=1in]{geometry}
\usepackage{float}
\usepackage{multirow}

\newtheorem{theorem}{Theorem}[section]
\newtheorem{lemma}[theorem]{Lemma}

\newtheorem{definition}[theorem]{Definition}

\newtheorem{corollary}[theorem]{Corollary}

\newtheorem{remark}[theorem]{Remark}
\newtheorem{claim}[theorem]{Claim}

\newtheorem{condition}[theorem]{Condition}

\newcommand{\by}{\boldsymbol{y}}
\newcommand{\bz}{\boldsymbol{z}}
\newcommand{\br}{\boldsymbol{r}}

\newcommand{\bg}{\boldsymbol{g}}

\newcommand{\bS}{\mathbf{S}}
\newcommand{\bA}{\mathbf{A}}
\newcommand{\bx}{\mathbf{x}}
\newcommand{\bG}{\mathbf{G}}
\newcommand{\bv}{\mathbf{v}}

\newcommand{\supp}{\mathsf{supp}}

\DeclareMathOperator*{\E}{\mathbb{E}}

\newcommand{\R}{\mathbb{R}}

\newcommand\numberthis{\addtocounter{equation}{1}\tag{\theequation}}

\makeatletter
\newenvironment{breakablealgorithm}
  {% \begin{breakablealgorithm}
    \begin{center}
      \refstepcounter{algorithm}% New algorithm
      \hrule height.8pt depth0pt \kern2pt% \@fs@pre for \@fs@ruled
      \parskip 0pt
      \renewcommand{\caption}[2][\relax]{% Make a new \caption
        {\raggedright\textbf{\fname@algorithm~\thealgorithm} ##2\par}%
        \ifx\relax##1\relax % #1 is \relax
          \addcontentsline{loa}{algorithm}{\protect\numberline{\thealgorithm}##2}%
        \else % #1 is not \relax
          \addcontentsline{loa}{algorithm}{\protect\numberline{\thealgorithm}##1}%
        \fi
        \kern2pt\hrule\kern2pt
     }
  }
  {% \end{breakablealgorithm}
     \kern2pt\hrule\relax% \@fs@post for \@fs@ruled
   \end{center}
  }
\makeatother

\newenvironment{proofof}[1]{\bigskip \noindent {\it Proof of #1.}\quad }
{\qed\par\vskip 4mm\par}

\begin{document}

\title{Revisit the Partial Coloring Method: Prefix Spencer and Sampling}
\author{
Dongrun Cai\footnote{cdr@mail.ustc.edu.cn
} \quad Xue Chen\footnote{\tt{xuechen1989@ustc.edu.cn}, College of Computer Science, University of Science and Technology of China, Hefei 230026 \& Hefei National Laboratory, University of Science and Technology of China, Hefei 230088, China. Supported by Innovation Program for Quantum Science and Technology 2021ZD0302901, NSFC 62372424, and CCF-HuaweiLK2023006.} \quad Wenxuan Shu\footnote{wxshu@mail.ustc.edu.cn
} \quad Haoyu Wang\footnote{why666@mail.ustc.edu.cn} \quad Guangyi Zou\footnote{zouguangyi2001@gmail.com
}\\ University of Science and Technology of China
}

\date{}

\maketitle

\begin{abstract}
%\textcolor{red}{To Do List: (0) Put prefix Spencer at Section 4 and defer the current sections; (1) Generalize the sampling algorithm in Spencer to $n \ge m$. (2) Check sampling algorithms in Integer Programming.}

As the most powerful tool in discrepancy theory, the partial coloring method has wide applications in many problems including  the Beck-Fiala problem \cite{Beck_Fiala} and  Spencer's celebrated result \cite{Spencer_six,Bansal10,LovettM12,Rothvoss17}.
Currently, there are two major algorithmic approaches for the partial coloring method: the first approach uses linear algebraic tools to update the partial coloring for many rounds \cite{Bansal10,LovettM12,BansalG17,LevyRR17,BDGL19}; and the second one, called Gaussian measure algorithm \cite{Rothvoss17,ReisR23}, projects a random Gaussian vector to the feasible region that satisfies all discrepancy constraints in $[-1,1]^n$. 

In this work, we explore the advantages of these two approaches and show the following results for them separately.

\begin{enumerate}
    \item Spencer \cite{spencer1986balancing} conjectured that the prefix discrepancy of any $\bA \in \{0,1\}^{m \times n}$ is $O(\sqrt{m})$, i.e., $\exists \bx \in \{\pm 1\}^n$ such that $\max_{t \le n} \|\sum_{i \le t} \bA(\cdot,i) \cdot \bx(i) \|_{\infty} = O(\sqrt{m})$ where $\bA(\cdot,i)$ denotes column $i$ of $\bA$. Combining small deviations bounds of the Gaussian processes and the Gaussian measure algorithm \cite{Rothvoss17}, we show how to find a partial coloring with prefix discrepancy $O(\sqrt{m})$ and $\Omega(n)$ entries in $\{ \pm 1\}$ efficiently. To the best of our knowledge, this provides the first partial coloring  whose prefix discrepancy is almost optimal \cite{Spencer_six,spencer1986balancing} (up to constants).
    
    However, unlike the classical discrepancy problem \cite{Beck81,Spencer_six}, there is no reduction on the number of variables $n$ for the prefix problem. By recursively applying partial coloring, we obtain a full coloring with prefix discrepancy $O(\sqrt{m} \cdot \log \frac{O(n)}{m})$. Prior to this work, the best bounds of the prefix Spencer conjecture for arbitrarily large $n$ were $2m$ \cite{BARANY19811} and $O(\sqrt{m \log n})$ \cite{Banaszczyk12,BansalG17}. 

    \item Our second result extends the first linear algebraic approach to a sampling algorithm in Spencer's classical setting. On the first hand, besides the six deviation bound \cite{Spencer_six}, Spencer also proved that there are $1.99^m$ good colorings with discrepancy $O(\sqrt{m})$ for any $\bA \in \{0,1\}^{m \times m}$. Hence a natural question is to design efficient random sampling algorithms in Spencer's setting. On the other hand, some applications of discrepancy theory, such as experimental design, prefer a random solution instead of a fixed one \cite{Fisher25, Student38, HSSZ_experiments19}. Our second result is an efficient sampling algorithm whose random output has min-entropy $\Omega(n)$ and discrepancy $O(\sqrt{m})$. Moreover, our technique extends the linear algebraic framework by incorporating leverage scores of randomized matrix algorithms.
\end{enumerate}

\end{abstract}

%\setcounter{page}{0}
%\pagebreak

\newpage

\section{Introduction}\label{intro}
The partial coloring method, as one of the most powerful techniques in discrepancy theory, 
has been applied to obtain the best known bounds for various problems. This method was first introduced by Beck \cite{Beck81}. Later on, Spencer \cite{Spencer_six} successfully applied it to prove the discrepancy of any set system of $m$ subsets is at most $6 \sqrt{m}$. Formally, given $\bA \in \mathbb{R}^{m \times n}$, we call $\min_{\bx \in \{\pm 1\}^n} \|\bA \bx\|_{\infty}$ the discrepancy of $\bA$ where $\bx \in \{\pm 1\}^n$ is a bi-coloring on the columns of $A$. In this work, we use $\bA(i,\cdot)$ to denote row $i$ of $\bA$ and $\bA(\cdot,j)$ to denote its column $j$. For convenience, we consider each row $\bA(i,\cdot)$ as a constraint on $\bx$ and call $\langle \bA(i,\cdot), \bx \rangle$ the discrepancy of this row. A partial coloring is a relaxation of $\bx$ from the discrete Boolean cube to $[-1,1]^n$ with $\Omega(n)$ entries in $\{\pm 1\}$. 

In particular, Spencer's classical result \cite{Spencer_six} shows that for any $\bA \in \{0,1\}^{m \times m}$, $\exists$ a partial coloring $\bx \in \{-1,0,1\}^m$ with $99\%$ entries in $\{\pm 1\}$ such that $\|\bA \bx\|_{\infty} = O(\sqrt{m})$. Via a reduction from $n$ variables to $m$ variables \cite{Beck81,Beck_Fiala}, the discrepancy of any $\bA \in \{0,1\}^{m \times n}$ is $O(\sqrt{m})$ by recursively applying Spencer's partial coloring. A celebrated line of research provides efficient algorithms to find such a partial coloring, to name a few \cite{Bansal10,LovettM12,Rothvoss17,LevyRR17,EldanS18,BansalLV22,PesentiV23}. Recently, the partial coloring method has been successfully applied to many other problems, such as the \emph{prefix} Spencer conjecture \cite{spencer1986balancing,BansalG17}, the vector balancing problem \cite{BansalDG19,BDGL19,BRR_zonotopes23}, matrix sparsification \cite{ReisRoth20}, and matrix discrepancy theory \cite{HRS22,DadushJR22,BJM23}. 

There are two major algorithmic approaches for the partial coloring method. The first one keeps updating the partial coloring 
via linear-algebraic tools, called linear-algebraic framework in this work. In \cite{LovettM12,LevyRR17,PesentiV23}, the algorithms find a safe subspace $H_t$ and pick a (random) vector $\bv_t \in H_t$ to update the coloring. Other works \cite{Bansal10,BansalDG19,BDGL19,BansalG17} have used semi-definite programs to find the subspace and provide the update vector. The second approach measures the feasible region by the (standard) Gaussian distribution, called Gaussian measure algorithm \cite{Rothvoss17,ReisR23}. Here the feasible region is the convex body $K \subset \mathbb{R}^n$ of all $x$ satisfying $\|\bA \bx\|_{\infty} \le L$ for some discrepancy bound $L$. Then a partial coloring is a point $\bx \in K \cap [-1,1]^n$ with $\Omega(n)$ entries in $\{\pm 1\}$. The key idea in the 2nd approach is to use the (standard) Gaussian measure to lower bound the size of $K$ \cite{Gluskin}. Then Rothvoss \cite{Rothvoss17} proved that the projection of a random Gaussian vector to $K \cap [-1,1]^n$ would be a good partial coloring (w.h.p.).
% reformulates the problem as finding a point with a constant fraction of entries in $\{ \pm 1 \}$ in a convex body . While these two approaches have 

While both approaches obtain the same algorithmic results for Spencer's bound, they are quite different. On the first hand, the linear-algebraic framework is more flexible that admits the algorithm to impose various linear constraints and pick any vector in the safe subspace $H_t$. On the second hand, the Gaussian measure approach could comply many more discrepancy constraints than the linear-algebraic framework, since it only needs a lower bound on the Gaussian measure of the feasible region $K$. In this work, we continue the study of the partial coloring method and explore the advantages of these two methods separately.

%%%%%%%%%%%%%%%%%%%%%%%%%%%%%%%%%%%%%%%%%%%%%%%%%%%%%%%%%%%%%%%%%%%%%%%%%%%%%%%%%%%%%%%%%%%%%%%%
%%%%%%%%%%%%%%%%%%%%%%%%%%%%%%%%%%%%%%%%%%%%%%%%%%%%%%%%%%%%%%%%%%%%%%%%%%%%%%%%%%%%%%%%%%%%%%%%
%%%%%%%%%%%%%%%%%%%%%%%%%%%%%%%%%%%%%%%%%%%%%%%%%%%%%%%%%%%%%%%%%%%%%%%%%%%%%%%%%%%%%%%%%%%%%%%%

\subsection{Our Results}
For ease of exposition, we focus on Spencer's setting in this section, whose input matrix $\bA \in \{0,1\}^{m \times n}$ satisfies $n \ge m$. Our first result is an application of the Gaussian measure algorithm \cite{Rothvoss17} to the prefix Spencer conjecture. Instead of bounding $\|\bA \bx\|_{\infty}$, the prefix discrepancy bounds $\max_{t \le n} \|\sum_{i=1}^t \bA(\cdot,i) \cdot \bx(i)\|_{\infty}$ where $\bA(\cdot,i)$ denotes column $i$ of $\bA$ such that $\sum_{i=1}^t \bA(\cdot,i) \cdot \bx(i)$ is a prefix vector summation of $\bA \cdot \bx$. While Spencer \cite{spencer1986balancing} conjectured that the prefix discrepancy is $O(\sqrt{m})$ in the same order of his six deviation bound \cite{Spencer_six}, the best known bounds were $2m$ \cite{BARANY19811} and $O(\sqrt{m \log n})$ \cite{Banaszczyk12,BansalG17} for \emph{arbitrarily large} $n$.

For the prefix Spencer conjecture, we show how to find a partial coloring $\bx$ such that (1) it has $\Omega(n)$ entries in $\{\pm 1\}$ and (2) its prefix discrepancy is $O(\sqrt{m})$. We remark that this upper bound on the prefix discrepancy is \emph{optimal} up to constants \cite{Spencer_six,spencer1986balancing}. Previously, such a optimal partial coloring were known only for the special case $n=m$ \cite{spencer1986balancing}.

However, different from the classical Spencer problem, there is no reduction from $n$ variables to $m$ variables in prefix discrepancy theory. Thus repeating this partial coloring leads to a \emph{full} coloring of prefix discrepancy $O(\sqrt{m} \cdot \log \frac{O(n)}{m})$. For $n=m^{1+o(1)}$, this also improves the previous bounds $2m$ and $O(\sqrt{m \log n})$ mentioned earlier. 
\begin{theorem}\label{thm:inform_pref_disc}
    Given any $\bA \in \{0,1\}^{m \times n}$ with $n \ge m$, the Gaussian measure algorithm finds a partial coloring $\bx \in [-1,1]^n$ such that
    \begin{enumerate}
        \item its prefix discrepancy is $O(\sqrt{m})$;
        \item $\bx$ has $\Omega(n)$ entries in $\{ \pm 1 \}$. 
    \end{enumerate}

    Moreover, there exists an efficient algorithm to find a full coloring $\bx \in\{\pm 1\}^n$ whose prefix discrepancy is $O(\sqrt{m} \cdot \log \frac{O(n)}{m})$.
\end{theorem}
Theorem~\ref{thm:inform_pref_disc} uses the crucial fact that to apply the Gaussian measure algorithm, we only need to lower bound the Gaussian measure of the convex body of feasible $\bx \in \mathbb{R}^n$. While the number of prefix constraints is $nm$ larger than the number of variables $n$, these constraints are highly correlated. For example, $|\sum_{i=1}^t \bA(1,i) \cdot \bx(i)| \le L$ implies that the same bound would hold for the next prefix summation $\sum_{i=1}^{t+1} \bA(1,i) \cdot \bx(i)$ very likely. The key step in Theorem~\ref{thm:inform_pref_disc} is to turn this intuition into a small-deviation bound for (sub-)Gaussian processes:
\begin{equation}\label{eq:small_dev_Gaussian}
\Pr_{\bx \sim N(0,1)^n} \left[ \forall t \le n, \big|\sum_{i=1}^{t} \bx(i) \big| = O(\sqrt{m}) \right] \ge 2^{-c \cdot n/m} \text{ for some small constant $c$}.    
\end{equation}
While inequality \eqref{eq:small_dev_Gaussian} holds for $x \sim \{\pm 1\}^n$ also, it is unclear how to satisfy $m$ constraints simultaneously for $x \sim \{\pm 1\}^n$. When $x \sim N(0,1)^n$, we apply the classical S\u id\' ak-Khatri inequality \cite{vsidak1967rectangular,khatri1967certain,royen2014simple} to lower bound the probability of satisfying all prefix constraints in $\bA$. This allows us to apply the classical Gaussian measure algorithm by Rothvoss \cite{Rothvoss17}.

Our second result studies the sampling question in Spencer's setting via the linear-algebraic framework. In his original paper \cite{Spencer_six}, Spencer had shown that there are exponentially many colorings satisfying $\|\bA \bx\|_{\infty} \le 10\sqrt{m}$ for any $\bA \in \{0,1\}^{m \times m}$. Hence a natural question is to sample from these good colorings $\bx$ (satisfying $\|\bA \bx\|_{\infty} \le 10 \sqrt{m}$). Moreover, some applications of discrepancy theory, such as experimental design \cite{Fisher25, Fisher26}, seek a random balanced coloring rather than a fixed one. Specifically, experimental design requires the coloring to be both balanced and robust. While there are several ways to define robustness, randomized solutions are considered as the most reliable way \cite{Student38}. For example, Harshaw et al. \cite{HSSZ_experiments19} considered how to generate a random coloring whose covariance matrix is balanced. In this work, we use the \emph{min-entropy} of the coloring as an alternative way to measure its randomness and robustness. Equivalently, we bound the probability that the sampling algorithm outputs a single fixed string.

\begin{theorem}\label{thm:infor_sampling}
    There exists a sampling algorithm such that given any $\bA \in \{0,1\}^{m \times n}$ with $n \ge m$, 
    \begin{enumerate}
        \item its output $\bx \in \{\pm 1\}^n$ always satisfies $\|\bA \bx\|_{\infty} \le O(\sqrt{m})$;
        
        \item for any good coloring $\boldsymbol{\epsilon} \in \{\pm 1\}^n$, the probability that the output equals $\boldsymbol{\epsilon}$ is $O(1.9^{-0.9n})$.
    \end{enumerate}
\end{theorem}
Our algorithm extends the linear-algebraic framework by incorporating leverage scores of randomized matrix algorithms \cite{Woodruff_Survey}. Although previous algorithms \cite{Bansal10,LovettM12,Rothvoss17} for Spencer's problem use a large amount of randomness to generate a partial coloring, it is unclear how to measure the \emph{min-entropy} of the full coloring obtained by these partial coloring methods. Our main idea is to use leverage scores of the safe subspace to find an entry $k_t$ such that every update $t$ could either fix $\bx(k_t)=1$ or $\bx(k_t)=-1$ arbitrarily. This guarantees that different partial colorings would lead to different colorings. In fact, in every update, our algorithm could fix $99\%$ uncolored entries in this way. Compared to previous algorithms in this framework \cite{Bansal10,LovettM12,BansalG17,LevyRR17,BansalLV22,PesentiV23}, this provides an alternate way to control the partial coloring; and our algorithm can be viewed as an application of this idea to sampling.

Specifically, we generalize the definition of the leverage scores on rows of a matrix \cite{Woodruff_Survey} to coordinates  of a subspace $H$. In particular, the leverage score $\tau_i(H)$ of coordinate $i$ indicates the closeness of the indicator vector $e_i$ to $H$. While indicator vectors $e_1,\ldots,e_n$ are not in the safe subspace for many settings, we use leverage scores of $H$ to find an approximation $\bv \in H$ for indicator vectors such that updating along $\bv$ or $-\bv$ could set an entry to be $1$ or $-1$ arbitrarily.

Finally, we extend the sampling algorithm to the Beck-Fiala problem \cite{Beck_Fiala} and   the vector-balancing problem \cite{Banaszczyk98}. Analogue to Theorem~\ref{thm:infor_sampling} in the Spencer setting, our algorithm outputs a good coloring $\bx$ satisfying the Beck-Fiala bound with a sufficiently large min-entropy. For the vector-balancing problem, we extend the online algorithm by Liu et~al.~\cite{LiuSS22} which uses a random walk with a Gaussian stationary distribution. Our contribution here is several new ingredients about the properties of this random walk. We describe the details of these two algorithms in Section~\ref{sec:sampling}.

%At last, we view this question as a way to extend the classical partial coloring method \cite{Beck_Fiala,Spencer_six}. The partial coloring method succeeds in proving a tight bound in Spencer's setting, which produces a small-discrepancy coloring in $[-1,1]^n$ with a constant fraction of fixed entries in $\{\pm 1\}$. For dense matrices in Spencer's setting, it uses the fixed entries to reduce the size of every set by a constant fraction. However, its current form is not strong enough to prove the Beck-Fiala conjecture. One particular reason is that it can not guarantee those fixed entries would reduce the size of every subset by a constant fraction in the Beck-Fiala setting. Therefore, we think that one direction to extend the partial coloring method is to design algorithms with a stronger guarantee on those fixed entries.

\subsection{Related Works}
\paragraph{Algorithms in the linear-algebraic framework.} For general $\{0,1\}$-matrices, Spencer's discrepancy bound $O(\sqrt{m})$ is tight for the Hadamard matrix. In a breakthrough, Bansal \cite{Bansal10} used a SDP-based random walk to find a partial coloring efficiently and provided the first algorithmic method of the Spencer problem. Lovett and Meka \cite{LovettM12} simplified Bansal's algorithm via walking on the edge of convex bodies. Very recently, Pesenti and Vladu \cite{PesentiV23} used a regularized random walk to improve the big-O constant in Spencer's bound. 

Several works have extended the linear algebraic framework to obtain efficient algorithms for Banaszczyk's bound  \cite{Banaszczyk98} about the vector balancing problem. The vector balancing problem assumes each column $\bA(\cdot,j)$ of $\bA$ is a unit vector such that the discrepancy minimization becomes balancing $\|\sum_{j=1}^n \bx(j) \cdot \bA(\cdot,j)\|_{\infty}$. The seminal works by Bansal et al. \cite{BansalDG19,BDGL19} showed efficient algorithms to generate $\bx$ with discrepancy $O(\sqrt{\log n})$ via guaranteeing that $\bA \bx$ is a random sub-Gaussian vector. Subsequent works \cite{LevyRR17,BansalLV22} provided unified approaches to obtain both Banszczyk's result and Spencer's result. In particular, Levy et al. \cite{LevyRR17} gave the first deterministic algorithms by derandomizing these random walk algorithms via a carefully selected multiplicative weight update method. 

Recent work by Harrow et~al. \cite{HSSZ_experiments19} improved the sub-Guassian constant of \cite{BDGL19} to $1$. Moreover, they consider balancing the covariance matrix of $\bx$ as a way to measure its randomness since randomization is the major method to guarantee the robustness in experimental design \cite{Student38}. 

\paragraph{Prefix Discrepancy Theory.}  Bárány and Grinberg \cite{BARANY19811} proved that the prefix discrepancy of any $\bA \in \{0,1\}^{m \times n}$ is at most $2m$ independent with $n$. Later on, Spencer \cite{spencer1986balancing} onjectured the prefix discrepancy is $O(\sqrt{m})$ and proved this for the special case $n=m$. In fact, prior to our algorithm, there is no algorithm (to the best of our knowledge) to find such a coloring even for the special case $n=m$. Banaszczyk \cite{Banaszczyk12} used deep techniques from convex geometry to prove the existence of $\bx$ with prefix discrepancy $O(\sqrt{m \log n})$ for Spencer's prefix conjecture.

In \cite{BansalG17}, Bansal and Garg provided an elegant prefix discprenacy algorithm in the linear-algebraic framework that matches Banaszczyk's bound. Technically, they showed how to find an update vector in a high dimensional subspace such that  $\sum_{j \le t} \bA(\cdot,j) \bx(j)$ satisfies certain sub-Gaussian properties. While our work is \emph{inspired} by this elegant algorithm \cite{BansalG17}, our results and techniques are incomparable. First of all, while Bansal and Garg's algorithm produces a full coloring with discrepancy $O(\sqrt{m \log n})$, it leaves a gap $O(\sqrt{\log n})$ to Spencer's conjecture. On the other hand, our prefix algorithm finds a \emph{partial} coloring whose prefix discrepancy matches the bound in Spencer's conjecture. Secondly, the technique of our sampling algorithm is to find a vector in subspace $H_t$ that is very close to indicator vectors $e_1,\ldots,e_n$. But Bansal and Garg's algorithm looks for a ``well-spread" random vector with certain sub-Gaussian properties.

 Recent work by Bansal et~al. \cite{bansal2022flow} conjectured that the prefix discrepancy of the Beck-Fiala problem (where each column of $\bA$ has $O(1)$ ones) is $O(1)$ and provided an beautiful application to approximation algorithms in scheduling theory.

%Very recently, Harshaw et al. \cite{HSSZ_experiments19} considers the randomness of the coloring output by discrepancy algorithms. Because randomization is an important way to guarantee robustness in experimental design, Harshwaw et al. consider balancing the covariance matrix of the coloring as a way to measure its randomness. In this work, we use the min-entropy and consider the random sampling question of the discrepancy problems.

\paragraph{The Gaussian measure algorithm.} Gluskin \cite{Gluskin} used the Gaussian measure to show the existence of a good partial coloring. Rothvoss \cite{Rothvoss17} designed an elegant algorithm to find a partial coloring directly based on the Gaussian measure of convex bodies. While this framework is less flexible than the first linear-algebraic framework, it could comply more discrepancy constraints. For example, this partial coloring method has been used in matrix sparsification \cite{ReisRoth20} and matrix Spencer's conjecture \cite{BJM23}. Furthermore, Reis and Rothvoss \cite{ReisR23} has shown the Gaussian measure could be exponentially small for any constant. Our technical contribution here is a connection between Gaussian processes and the Gaussian measure algorithm. We believe it would be interesting to see more Gaussian measure algorithms based on Gaussian processes.

\paragraph{Online Discrepancy Theory.} Online problems require the algorithm to output $\bx(i)$ after reading column $\bA(\cdot,i)$ whose goal is still to bound every prefix summation. 
Alweiss et al.~\cite{AlweissLS21} provided an elegant algorithm with online discrepancy $O(\log m)$ for the vector balancing problem where each column is of length $\le 1$. Later on, Liu et al.~\cite{LiuSS22} used random walks with Gaussian stationary distributions to obtain a \emph{partial} coloring with online discrepancy $O(\sqrt{\log n})$. 
Very recent work by Kulkarni et al.~\cite{kulkarni2023optimal} have shown optimal online algorithms whose online discrepancy is $O(\sqrt{\log m})$, despite their running time is not a polynomial. One remark is that \cite{kulkarni2023optimal} only implies a coloring with prefix discrepancy $O(\sqrt{m \cdot \log n})$ for the prefix Spencer conjecture. Other online settings have been studied in \cite{spencer1977balancing,bansal2021online}.

Finally, there are many other settings and applications of discrepancy theory including hereditary discrepancy \cite{MNT_hereditary}, discrepancy of permutations \cite{NewmanNN12}, differential privacy \cite{NikolovTZ13}. We refer to textbooks \cite{Chazelle,Matousek} for an overview.

%%%%%%%%%%%%%%%%%%%%%%%%%%%%%%%%%%%%%%%%%%%%%%%%%%%%%%%%%%%%%%%%%%%%%%%%%%%%%%%%%%%%%%%%%%%%%%%%
%%%%%%%%%%%%%%%%%%%%%%%%%%%%%%%%%%%%%%%%%%%%%%%%%%%%%%%%%%%%%%%%%%%%%%%%%%%%%%%%%%%%%%%%%%%%%%%%
%%%%%%%%%%%%%%%%%%%%%%%%%%%%%%%%%%%%%%%%%%%%%%%%%%%%%%%%%%%%%%%%%%%%%%%%%%%%%%%%%%%%%%%%%%%%%%%%

\subsection{Discussion}
In this work, we explore the advantages of the two algorithmic frameworks for the partial coloring method and apply them to two problems in Spencer's setting separately.

For the prefix Spencer conjecture, our algorithm finds a partial coloring whose prefix discrepancy is almost-optimal. This leads to a full coloring of discrepancy $O(\sqrt{m} \cdot \log \frac{O(n)}{m})$. While it is well known how to reduce $n$ in classical discrepancy problems, much less is known about prefix problems. It would be interesting to investigate the reduction of variables for the prefix discrepancy problem. 

Technically, we apply the small deviation bounds of Gaussian processes to the Gaussian measure algorithm. Inspired by the chaining argument in Gaussian processes, we use the correlation of those prefix constraints to reduce the loss of applying the S\u id\' ak-Khatri inequality directly. Similar ideas have been studied in other works about matrix sparsification \cite{ReisRoth20,DadushJR22}. It is interesting to investigate the application of this idea to other problems in discrepancy theory including the Kadison-Singer problem \cite{MSS_II} and $\ell_1$ subspace embedding (a.k.a. sparsification of zonotopes) \cite{Talagrand90,BRR_zonotopes23}.

There are many more open questions in prefix discrepancy theory. While the prefix discrepancy of the Beck-Fiala problem is conjectured to be $O(1)$ \cite{bansal2022flow}, the best known upper bound is $O(\sqrt{\log n})$ from Banaszczyk's existence proof \cite{Banaszczyk12}. In fact, there is no \emph{efficient} algorithm to output a coloring with prefix discrepancy $O(\sqrt{\log n})$ in this setting yet \cite{kulkarni2023optimal}. This is in contrast with the prefix Spencer conjecture where Bansal and Garg \cite{BansalG17} provided an efficient algorithm matching Banaszczyk's bound $O(\sqrt{m \log n})$. Furthermore, an intriguing direction is to design \emph{efficient} online algorithm for the vector balancing problem with a prefix discrepancy $O(\sqrt{\log n})$.

For the linear-algebraic framework, our algorithm admits the partial coloring to fix an entry $\bx(k_t)$ to $1$ or $-1$ arbitrarily in every update $t$. While we state it as a sampling algorithm, one could apply it to control those $\{\pm 1\}$-entries in the partial coloring method. For example, in the Beck-Fiala setting, given the partial coloring $\bx_t$ at step $t$, let $\mathcal{F}_t$ be the set of fixed $\{\pm 1\}$-entries in $\bx_t$. For unfixed entries in $\bx_t\big( [n] \setminus \mathcal{F}_t \big)$, our algorithm guarantees that $0.99 \cdot (n - |\mathcal{F}_t|)$ entries could be chosen to get fixed (to $\pm 1$) in this update. In another word, at each step, it could set \emph{almost} any unfixed entry to either $1$ or $-1$. This provides a way to strengthen the standard partial coloring method. One open question would be to use stronger partial coloring methods to obtain a bound like $O(\sqrt{d \log d})$ for the Beck-Fiala problem of degree-$d$. 

\paragraph{Organization.}
The rest of this paper is organized as follows.
In Section~\ref{sec:Prel}, we provide basic notations and definitions. 
In Section~\ref{sec:overview}, we provide an overview of our algorithms and techniques.
 In Section~\ref{sec:prefix_spencer}, we prove Theorem~\ref{thm:inform_pref_disc} about the prefix discrepancy conjecture.
 In Section~\ref{sec:sampling_Spencer}, we provide a random sampling algorithm in Spencer's setting and prove Theorem~\ref{thm:infor_sampling}.
 In Section~\ref{sec:sampling}, we discuss random sampling algorithms in Beck-Fiala's and Banaszczyk's settings.

\section{Preliminaries}\label{sec:Prel}
Let $\|\cdot\|_p$ denote the $\ell_p$ norm of a vector $\bv$, i.e., $\|\bv\|_p=\big( \sum_i |\bv(i)|^p \big)^{1/p}$. For convenience, we use $\exp(x)$ to denote the function $e^{x}$ of a real variable (number) $x$ and $[n]$ to denote $\{1,\ldots,n\}$. 

We use bold letters like $\mathbf{A}$ and $\mathbf{v}$ to denote matrices and vectors. For a vector $\mathbf{v} \in \mathbb{R}^n$,
let $\supp(\bv)$ denote the set of non-zero entries in $\bv$ as the support of $\bv$ and $\mathbf{v}(S)$ for a subset $S \subseteq [n]$ denote the punctured-vector in $\mathbb{R}^S$ such that $\mathbf{v}(i)$ denotes its entry $i$. For two vectors $\mathbf{u} \in \mathbb{R}^S$ and $\bv \in \mathbb{R}^T$ with disjoint supports $S \cap T= \emptyset$, let $\mathbf{u} \circ \bv$ denote their concatenation on $S \cup T$.

For convenience, let $\vec{0}$ and $\vec{1}$ denote the all-0 and all-1 vectors separately. Moreover, we use $\mathbf{e}_i$ to denote the indicator vector for coordinate $i$. 

For a symmetric matrix $\bG \in \mathbb{R}^{n \times n}$, we use $\bG=\sum_{j=1}^n \mu_j \cdot \phi_j \phi_j^{\top}$ to denote its eigen-decomposition where $\mu_1 \ge \mu_2 \ge \cdots \ge \mu_n$. For any matrix $\bA$ whose dimension is $m \times n$, we use $\bA(S,T)$ to denote the submatrix on $S \times T$ for any $S \subseteq [n]$ and $T \subseteq [m]$. In this work, $\bA(\cdot,i)$ denotes column $i$ of $\bA$ and $\bA(j,\cdot)$ denotes its row $j$. Moreover, $\bA(\cdot,[t])$ will denote the submatrix of the first $t$ columns.

\paragraph{Discrepancy theory.} Through this work, we always use $\mathbf{A} \in \mathbf{R}^{m \times n}$ to denote the input matrix for various settings such that discrepancy theory looks for a $\mathbf{x} \in \{\pm 1\}^n$ minimizing $\|\mathbf{A} \mathbf{x}\|_{\infty}$. Since $\|\bA \bx\|_{\infty}=\max_{i \in [m]} \big| \langle \bA(i,\cdot), \bx \rangle \big|$, we consider each row $\bA(i,\cdot)$ as a constraint in the discrepancy problem. Moreover, we call $\bx \in [-1,1]^n$ a partial coloring and $\bx \in \{\pm1\}^n$ a coloring.

We generalize Spencer's setting to matrices whose entries are in $[-1,1]$. For convenience, we focus on the setting where $n \ge m$ in this work. 

We summarize the classical bounds by Spencer \cite{Spencer_six} and its algorithmic results by \cite{Bansal10,LovettM12,Rothvoss17,LevyRR17}. We state the following version whose starting point $\bx_0$ could be arbitrary.
\begin{theorem}\label{thm:spencer}
    Given any $\bA \in [-1,1]^{m \times n}$ with $n \ge m$, there exist exponentially many $\bx \in \{\pm 1\}^n$ such that $\|\bA \bx\|_{\infty} =O(\sqrt{m})$. Moreover, there exist efficient algorithms that given any starting point $\bx_0$, they find $\bx \in \{\pm 1\}^n$ with $\bx(i)=\bx_0(i)$ for each $\bx_0(i) \in \{\pm 1\}$ satisfying $\|\bA (\bx-\bx_0) \|_{\infty}=O(\sqrt{m})$.
\end{theorem}

\paragraph{Prefix discrepancy.} In the prefix problem, the goal is to bound the prefix summation of any $\ell$ terms
\begin{equation}\label{eq:prefix_disc}
\min_{x \in \{\pm 1\}^n} \max_{\ell \le n} \|\sum_{i \le \ell} \bA(\cdot,i) \cdot \bx(i)\|_{\infty},    
\end{equation}
instead of the total summation $\bA \cdot \bx=\sum_{i \le n} \bA(\cdot,i) \cdot \bx(i)$. Since $\bA(\cdot,[\ell])$ denotes the submatrix of the first $\ell$ columns and $\bx([\ell])$ denotes the subvector of the first $\ell$ entries, we simplify the prefix summation as
\[
\sum_{i \le \ell} \bA(\cdot,i) \cdot \bx(i)=\bA(\cdot,[\ell]) \cdot \bx([\ell]).
\]
One remark is that the prefix discrepancy in \eqref{eq:prefix_disc} provides a upper bound on the discrepancy of any consecutive subset of indices (up to factor 2) \cite{bansal2022flow}.

Spencer \cite{spencer1986balancing} conjectured that when $\bA \in [-1,1]^{m \times n}$, no matter how large is $n$, there always exists $\bx \in \{-1,1\}^n$ such that the prefix discrepancy defined in \eqref{eq:prefix_disc} is $O(\sqrt{m})$.

%\paragraph{The Beck-Fiala setting.} Recall that $\bv_1,\ldots,\bv_n$ are the columns of $A$. We use $d:=\underset{j \in [n]}{\max} \big| \supp(\bv_j) \big|$ to denote the degree of a matrix $\bA$. We state the classical result by Beck and Fiala \cite{Beck_Fiala} for bounded-degree matrices in $[-1,1]^{m \times n}$. For convenience, we use the following version whose starting point $\bx_0$ could be arbitrary.
%\begin{theorem}\label{thm:beck-fiala} Given any $\bA \in [-1,1]^{m \times n}$ of degree at most $d$ and any starting point $\bx_0 \in [-1,1]^n$, there exists an efficient algorithm to find $\bx \in \{\pm 1\}^n$ such that $\bx(i)=\bx_0(i)$ for each $\bx_0(i) \in \{\pm 1\}$ and $\|\bA \bx\|_{\infty} = O(d) + \|\bA \bx_0\|_{\infty}$. \end{theorem}

%\begin{itemize}
%    \item $\mathbf{A}$: the input matrix of dimension $m \times n$
%    \item $\mathbf{v}_{i} \in \mathbb{R}^m$ for $i \in [n]$: column $i$ of $\mathbf{A}$
%    \item $\epsilon_{i} \in \{\pm 1\}$: the assignment of $v_i$
%    \item $a\vee b$: $(a\vee b)(i)=a(i)\vee b(i)$ where $a,b\in\R^{n}$.
%    \item A set system $(V,\mathbb{S})$, $|V|=n$,$|\mathbb{S}|=m$ and its adjacency matrix $A\in \{0,1\}^{m\times n}$: $A_{ij}$=1 iff $v_{j}\in S_{i}$ and $A_{ij}$=0 iff $v_{j}\notin S_{i}$.
%    \item For matrix $A\in \R^{m\times n}$, $A_{i}\in \R^{n}$ is the i'th row and $A^{i}\in \R^{m}$ is the i'th column.\end{itemize}

%Namely that algorithm can find a full coloring $\epsilon \in\{\pm 1\}^{m}$ efficiently such that $\|\sum_{i=1}^{m}\epsilon_{i}\cdot v_{i}\|_{\infty}=O(\sqrt{\log(mn)})$.

\section{Overview}\label{sec:overview}
In this section, we provide an overview of our algorithms in Theorem~\ref{thm:inform_pref_disc} and Theorem~\ref{thm:infor_sampling}.

\paragraph{Prefix Discrepancy Theory.} The starting point of the prefix discrepancy algorithm is L\'evy's inequality~\cite{levy1938theorie} which bounds the deviation of the prefix summations of symmetric random variables like $X_i \sim \{\pm 1\}$.
\begin{lemma}
    [L\'evy's inequality ] \label{lem:levy ineq}
    Let $X_{1},X_{2},\ldots, X_{n}$ be independent symmetric random variables and $\bS_{k}=\sum_{i=1}^{k}X_{i}$ be their partial summations of $k\in[n]$. For any $\theta>0$ we have \begin{enumerate}
        \item $\Pr[\underset{k\in[n]}{\max}\{\bS_{k}\}\ge \theta]\le 2\cdot\Pr[\bS_{n}\ge \theta],$
        \item \label{levy_ineq_2}
        $\Pr[\underset{k\in[n]}{\max}\{|\bS_{k}|\}\ge \theta]\le 2\cdot\Pr[|\bS_{n}|\ge \theta].$        
    \end{enumerate}
\end{lemma}
While L\'evy's inequality has been used for the special case $n=m$ \cite{spencer1986balancing}, this is not stronger enough to bound the prefix discrepancy by $O(\sqrt{m})$ for arbitrarily large $n$. Even for a single constraint $\bA(j,\cdot)$, $\sum_{i=1}^n \bA(j,i) \bx(i)$ could have variance $\Theta(n)$; and $\Pr\big[ |\sum_{i=1}^n \bA(j,i) \bx(i)| = \omega(\sqrt{m}) \big]$ is close to 1 when $n=\omega(m)$.

The main technical tool is to prove a small-deviation bound directly: for $\bx \sim \{\pm 1\}^n$,
\begin{equation}\label{eq:small_dev}
\Pr \left[ \forall t \le n, \big|\sum_{i=1}^{t} \bA(j,i) \cdot \bx(i) \big| = O(\sqrt{m}) \right] \ge 2^{-n/m}.    
\end{equation}
The intuition behind \eqref{eq:small_dev} is that Lemma~\ref{lem:levy ineq} implies that the prefix discrepancy of $2m$ variables is $O(\sqrt{m})$ with a constant probability. Hence we could partition $n$ entries into $n/2m$ intervals and apply Lemma~\ref{lem:levy ineq} to each interval. Conditioned on that each interval satisfies the prefix constraint, then we need to consider the summation of every interval. While these $n/2m$ interval summations  may no longer be sub-Gaussian conditioned on the prefix constraint on each interval, they are still symmetric. This turns out to be enough for a lower bound like $2^{-n/m}$.

Back to the prefix discrepancy algorithm, although \eqref{eq:small_dev} implies a probability $2^{-n/m}$ for any single constraint, it is unclear how to turn this into a probability $2^{-n}$ for all $m$ constraints. Hence we consider the Gaussian measure instead of the discrete Boolean cube and apply the  S\u id\' ak-Khatri inequality for $m$ contraints. However, by the Gaussian measure algorithm \cite{Rothvoss17,ReisR23}, we could only obtain a \emph{partial} coloring with a prefix discrepancy $O(\sqrt{m})$. The last point here is that to obtain a full coloring of prefix discrepancy $O(\sqrt{m} \cdot \log \frac{n}{m})$ instead of $O(\sqrt{m} \cdot \log n)$, we will adjust the deviation from $O(\sqrt{m})$ to $O(\sqrt{n \cdot \log \frac{O(m)}{n}})$ when $n=O(m)$ and apply Levy's inequality for the new deviation.

\paragraph{Sampling.} Next we consider the sampling algorithm in Spencer's setting. For ease of exposition, we consider the case $n=\omega(m)$ and use $\bx_t$ to denote the partial coloring at step $t$. The sampling algorithm for $n=O(m)$ is more involved whose discussion will be deferred to the end of this section.

The classical reduction \cite{Beck81} considers the dual space $H \subset \mathbb{R}^n$ of all rows in $\bA$ such that any update vector $\bv_t \in H$ has $\|\bA \bv_t\|_{\infty}=0$. Back to the sampling question, even though there are exponentially many choices of $\bv_t$ in $H$, one needs to guarantee that these updates will generate exponentially many colorings eventually. Our first idea circumvents this by requiring that each update sets an fixed entry $\bx(k_t)$ to be either $1$ or $-1$. In another word, our algorithm finds a coordinate $k_t$ instead of a direction $\bv_t$.
Even though this reduces the number of choices, it guarantees that all partial colorings generated in this process are distinct. Then our second idea of finding the coordinate $k_t$ relies on the properties of leverage scores of $H$. We recall a few properties of the leverage score here.

\paragraph{Leverage Scores.} The standard leverage score is defined for every row of a full-rank matrix $\bG \in \mathbb{R}^{n \times d}$ with $n \ge d$:
\begin{equation}\label{eq:leverage_score}
    \tau_i=\bG(i,\cdot)^{\top}  (\bG^{\top} \bG)^{-1} \bG(i,\cdot) \textit{ for row $i$ in $\bG$.} 
\end{equation}  
One could generalize it for every coordinate of a subspace by basic properties of leverage scores \cite{Woodruff_Survey,CP19}. We summarize two properties of $\tau_i$ here and provide a self-contained proof.
\begin{claim}\label{clm:leverage_socres_subspace}
    Given a linear subspace $H \subseteq \mathbb{R}^n$, for each coordinate $i\in n$, let $\tau_i(H):=\underset{\mathbf{u} \in H \setminus \{\vec{0}\}}{\max} \frac{|\mathbf{u}(i)|^2}{\|\mathbf{u}\|_2^2}$. Then $\sum_{i \in [n]} \tau_i(H)=\text{dim}(H)$.

    Moreover, given $H$ and $i$, Algorithm~\ref{alg:leverage_score_alg} computes the corresponding vector $\mathbf{u}$ such that $\frac{\mathbf{u}(i)^2}{\|\mathbf{u}\|_2^2}=\tau_i(H)$ (and $\mathbf{u}(i)=1$).
\end{claim}

 %(since an invertible linear transform on $G$ will not change its leverage scores in \eqref{eq:leverage_score}, one could consider its orthonormal basis to prove these properties like Claim~\ref{clm:leverage_socres_subspace} below).
%\begin{equation*}
%    \sum_{i=1}^m \tau_i= \text{dim}(\bG) \textit{ and } \tau_i = \underset{x \in \mathbb{R}^n : \bG x \neq 0}{\max} \frac{\langle \bG_i, x \rangle^2}{\|\bG x\|_2^2}.
%\end{equation*}

%In fact, the leverage scores of $\bG$ only depend on the linear subspace $H=\{\bG \cdot x| x \in \mathbb{R}^n\}$ from Claim~\ref{clm:leverage_socres_subspace}. Therefore Our algorithm computes the leverage score of every coordinate in $H_t$. Moreover, there are efficient algorithms to compute the leverages scores of $H$. 
In fact, any matrix $\bG$ generating $H$ will have the same sequence of leverage scores. Thus one could put a basis of $H$ into $\bG$ and compute $\tau_i$ via \eqref{eq:leverage_score}. Then we describe Algorithm~\ref{alg:leverage_score_alg} here. For ease of exposition, we require that the output vector has $\mathbf{u}(i)=1$ instead of $\|\mathbf{u}\|_2=1$.
 \begin{algorithm}[H]
 \caption{Compute a vector with a large entry}\label{alg:leverage_score_alg}
     \hspace*{\algorithmicindent} \textbf{Input}: a linear subspace $H \subset \mathbb{R}^{n}$ and an index $i \in [n]$. \\ 
     \hspace*{\algorithmicindent}
     \textbf{Output}: $\mathbf{u} \in H$.
    
     \begin{algorithmic}[1]
     \Procedure{FindVector}{$H,i$}
        \State Find an orthonormal basis $\psi_1,\ldots,\psi_d$ of $H$.
        \State Let $\mathbf{u} \gets \psi_1(i) \cdot \psi_1 + \cdots + \psi_d(i) \cdot \psi_d$.
        \State \Return $\mathbf{u} := \mathbf{u}/\mathbf{u}(i)$.
     \EndProcedure
     \end{algorithmic}
 \end{algorithm}

\begin{proofof}{Claim~\ref{clm:leverage_socres_subspace}}
    Let $\psi_1,\ldots,\psi_d$ be an orthonormal basis of $H$. Then for any $\mathbf{u} \in H$, we rewrite $\mathbf{u} = \alpha_1 \psi_1 + \cdots + \alpha_d \psi_d$ and obtain
    \[
     \frac{|\mathbf{u}(i)|^2}{\|\mathbf{u}\|_2^2} = \frac{\big( \sum_{j=1}^d \alpha_j \psi_j(i) \big)^2}{\sum_{j=1}^d \alpha_j^2} \le \frac{\big( \sum_j \alpha_j^2 \big) \cdot \big(\sum_j \psi_j(i)^2 \big)}{\sum_j \alpha_j^2} = \sum_j \psi_j(i)^2
    \]
    where we use the Cauchy-Schwartz inequality in the 2nd step. Since the inequality is tight when $\alpha_1=\psi_1(i),\ldots, \alpha_d=\psi_d(i)$, we know $\tau_i(H)$ equals $\sum_j \psi_j(i)^2$. Thus $\sum_i \tau_i(H) = d$. Also, the correctness of Algorithm~\ref{alg:leverage_score_alg} follows the above analysis.
\end{proofof}

In our sampling algorithm, we choose $k_t$ as the coordinate with the largest leverages score in the safe subspace $H_t$. The reason is as follows. Since updating $\bx_t$ by an indicator vector $\mathbf{e}_{k}$ (like $\bx_{t+1}=\bx_t+ \alpha_t \cdot \mathbf{e}_k$ for a scalar $\alpha_t$) could set $\bx_{t+1}(k)$ to be either $1$ or $-1$, the intuition is to find a vector $\mathbf{v}_t$ in $H$ as close to indicator vectors as possible. By Claim~\ref{clm:leverage_socres_subspace}, the closest vector to $\mathbf{e}_k$ in $H$ is determined by the leverage score of coordinate $k$ in $H$  \big(because $\underset{\mathbf{u} \in H:\|\mathbf{u}\|_2=1}{\max} \langle \mathbf{e}_k, \mathbf{u} \rangle^2 = \underset{\mathbf{u} \in H:\|\mathbf{u}\|_2=1}{\max} |\mathbf{u}(k)|^2=\tau_k(H)$\big). Thus we pick the coordinate $k_t$ with the largest leverage score.

However, $\mathbf{v}_t \neq \mathbf{e}_{k_t}$. To guarantee the update $\bx_{t+1} \gets \bx_t + \alpha_t \cdot  \mathbf{v}_t$ is in $[-1,1]^n$ after setting $\bx_{t+1}(k_t) \in \{\pm 1\}$, we require $\bx_t$ to be almost $\vec{0}$ except those fixed entries. One way to prove this is to show $\|\bx_t\|_2^2 \approx t$, which is achieved by a strong martingale concentration (and further requiring $H$ and $\mathbf{u}$ to be orthogonal to $\bx_t$). 

%Finally, to guarantee the sampling probability is sufficiently small, we give a high concentration bound on the difference between $\|\mathbf{x}_t\|_2^2$ and $t$ such that this update works for $0.9n$ steps.

For the case $n=O(m)$, many partial-coloring algorithms set a polynomial (or logarithmic) small scalar $\alpha_t$ in each update. However, we need the scalar $\alpha_t$ to be at least 1 such that each update could fix an entry to either $+1$ or $-1$. We notice that the deterministic discrepancy algorithm of Levy, Ramadas, and Rothvoss \cite{LevyRR17} could work for  constant scalars in each update. Thus we apply our techniques to the deterministic discrepancy algorithm in \cite{LevyRR17}, which provides the sampling algorithm for $n=O(m)$ with discrepancy $O(\sqrt{m})$.

\section{Prefix Spencer Problem}
\label{sec:prefix_spencer}
We show how to find an optimal partial coloring for the prefix Spencer conjecture in this section. In this section, $\epsilon$ is a small constant to be fixed later (in Lemma~\ref{lem:Roth17_main}). Since we will run the partial coloring algorithm recursively to obtain a full coloring, we state the general form with an arbitrary starting point $\bx \in (-1,1)^n$.

\begin{theorem}
\label{thm:prefix_main}
    There exists $\epsilon>0$ such that for any $\bA \in [-1,1]^{m\times n}$ with $n\ge m$, function \textsc{PrefixPartialColoring} in Algorithm~\ref{alg:prefix_partial_coloring} takes any $\bx \in (-1,1)^n$ and $\bA$ as inputs to return a partial coloring $\bx'$ such that with high probability,
    \begin{enumerate}
        \item $\epsilon/4$ fraction of entries in $\bx'$ are $\pm 1$.
        
        \item The prefix discrepancy of $\bx-\bx'$ is $O(\sqrt{m})$: 
        \[\sup_{\ell \le n} \| \bA(\cdot,[\ell]) \cdot \big(\bx'([\ell])-\bx([\ell]) \big)\|_{\infty}=O(\sqrt{m}).\]
    \end{enumerate}
\end{theorem}
We have the following corollary by applying Theorem~\ref{thm:prefix_main} recursively to obtaining a full coloring whose proof is deferred to Section~\ref{sec:proof_cor}.

\begin{corollary}\label{cor:prefix_spencer_v2}
    For any $\bA \in [-1,1]^{m\times n}$ with $n\ge m$, there is a polynomial time  algorithm that generates a coloring $\bx\in \{-1,1\}^{n}$ satisfying   \[\underset{1\le t\le n}{\max} \ \|\bA(\cdot, [t]) \cdot \bx([t])\|_{\infty}\le O\bigg(\sqrt{m}\cdot \ln\frac{O(n)}{m}\bigg)\]   with high probability.
\end{corollary}

We describe function~\textsc{PrefixPartialColoring} of Theorem~\ref{thm:prefix_main} in Algorithm~\ref{alg:prefix_partial_coloring}. Basically, function~\textsc{PrefixPartialColoring} function rounds $\Omega(n)$ entries of non-negative entries in $\bx$ (defined as $I$) to $\pm 1$. For the first time reader, it might be more convenient to assume $\bx=\vec{0}$ such that $I=[n]$ and $T_{\bx(I)}$ is the identity map of $\mathbb{R}^n$. For general $\bx$, $T_{\bx(I)}(Q)$ denotes the convex body of $Q$ after applying the linear transformation $T_{\bx(I)}$. In the rest of this section, we finish the proof of Theorem~\ref{thm:prefix_main}. %We use the following Procedures.

\begin{algorithm}[H]
 \caption{partial coloring with optimal prefix discrepancy}\label{alg:prefix_partial_coloring}
     % \hspace*{\algorithmicindent} \textbf{Input}: a symmetric convex body $Q\subset \R^{n}$ \\ 
     % \hspace*{\algorithmicindent}
     % \textbf{Output}: a partial coloring $ \bx^{\star}\in C\cdot Q\cap  [-1,1]^{n}$ with $\Big|\Big\{i: \bx(i)\in \{\pm 1\}\Big\}\Big|\ge \epsilon n$.
     \begin{algorithmic}[1]
     \Function{GaussianSampling}{$Q$ and $S$ such that $Q$ is a symmetric convex body in $\mathbb{R}^S$}
    \State Sample a Gaussian vector $\by^{*}\sim N(0,1)^S$
    \State Compute $\bx^{*} = \text{arg} \min_{\bz\in Q\cap [-1,1]^{S}} \{\|\by^{*}-\bz\|_{2}\}$
    
    % where $Q\triangleq \{\bx\in \R^{n}|\underset{1\le k\le n}{\max} \ \|\ba_{[k]} \cdot \bx_{[k]}\|_{\infty}\le 100\cdot\sqrt{n\cdot\ln (\frac{2m}{n})}\}$
    
    \State Return $\bx^{*}$
    \EndFunction
    \Function{PrefixPartialColoring}{$\bx\in (-1,1)^n,\bA\in \R^{m\times n}$}
%    \State $I^{\prime}\gets\{i\in [n]: \bx(i) \notin \{-1,1\}\}$
 
    \State W.l.o.g. we assume half entries in $\bx$ are non-negative --- otherwise consider its flip $-\bx$
    
    \State Let $ I$ be the sequence of indices $i\in [n]$ with $\bx(i) \ge 0$
% \State $U_{t}^{\pm}=\{\bx\in \R^{I_{t}^{\pm }}\}$

\State  $Q\gets  \bigg\{\bg\in \R^{I} : \underset{j\in [m]}{\max} \underset{ 1\le k\le |I| }{\max}  \Big|\sum_{i = I(1)}^{I(k)} \bg(i)\cdot \bA(j,i)\Big| \le 112\sqrt{m}\bigg\}$

\State Define a linear operator $T_{\bx(I)}: \mathbb{R}^{I} \rightarrow \mathbb{R}^{I}$ such that $T_{\bx(I)}(\by)\gets \frac{1}{1-\bx(i)}\cdot \by(i)$ for every $i \in I$

%\State $\hat{\bx}_{t}^{+}= \bx_{t}(I_{t}^{+})$ and $\hat{\bx}_{t}^{-}= \bx_{t}(I_{t}^{-})$ 

\State $\bx^*\gets \textsc{GaussianSampling}\big( T_{\bx(I)}(Q) ,I \big)$ 
% \State Run Algorithm \ref{alg:gaussian sampling} on $T_{\bx_{t}(I_{t})}(Q_{t})$ to get the partial coloring  $\bx^{*}_{t}\in \R^{I_{t}}$. 

    \State Flip $\bx^{*}$ if the number of $1$-entries in $\bx^*$ are less than $\frac{\epsilon}{2} \cdot |I|$
%    \wxnote{Maybe we can't determine $\epsilon$}
    
    \State For every $i \in I$, $\bx'(i)\gets\bx(i)+(1-\bx(i)) \cdot \bx^*(i)$; otherwise 
     $\bx'(i)\gets\bx(i)$ for $i \notin I$
    
    \State \Return $\bx'$
    \EndFunction
     \end{algorithmic}
 \end{algorithm}

The key of Function~\textsc{PrefixPartialColoring} is the Gaussian measure algorithm \cite{Rothvoss17,ReisR23}, which is reformulated as Function~\textsc{GaussianSampling}. %Function~\textsc{GaussianSampling} takes a symmetric convex body $Q\subset \R^n$ as input and outputs a partial coloring $x^{\star}\in C\cdot Q\cap  [-1,1]^{n}$.  
The main property of Function~\textsc{GaussianSampling} is that when $Q$ is large enough, its output $\bx^*$ has $\Big|\Big\{i: \bx^*(i)\in \{\pm 1\}\Big\}\Big| = \Omega(n)$. In particular, we use the Gaussian measure $\gamma(\cdot)$ %, i.e., \[
to measure the size of $Q$. For convenience, we provide a general definition:
\[
\gamma_S(K):=\Pr_{\bg \sim N(0,1)^S}[\bg \in K] \text{ for any measurable set $K \subset \mathbb{R}^S$}\]
where $N(0,1)^S$ denotes a random vector in $\mathbb{R}^S$ whose entries are sampled from $N(0,1)$ independently.
The correctness of function \textsc{GaussianSampling} follows from Theorem~7 of \cite{Rothvoss17}.
\begin{lemma}\label{lem:Roth17_main}
  Let $\epsilon:=10^{-4}$ and $\delta \coloneqq \frac{9}{5000}$. Suppose that $K\subset \R^S$ is a symmetric convex body with $\gamma_S(K)\geq e^{-\delta n}$. 
  %Choose a random Gaussian $x^*\sim N^n(0,1)$ and let $y^*$ be the point in $K\cap [-1,1]^n$ that minimizes $\|x^*-y^*\|_2$. 
  Then the output $x^*$ of function~\textsc{GaussianSampling}$(K,S)$ has at least $\epsilon \cdot |S|$ many coordinates in $\{-1,1\}$ with probability $1-e^{-\Omega(|S|)}$.
\end{lemma}

%\begin{lemma}[Theorem 21 of \cite{ReisR23}] \label{lem:RR_main}   For any constant $\alpha,\beta>0$, there exist $\epsilon\coloneqq \epsilon(\alpha,\beta)>0$ and $\delta(\alpha,\beta)>0$ so that the following holds: Let $K\subset \R^n$ be a symmetry convex body with $K\subset H$ for a subspace $H\subset \R^n$ with $dim(H)\ge \beta n$     and $\gamma_{H}(K)\ge e^{-\alpha\cdot n}$. Assuming a weak    separation oracle for $K$, function \textsc{GaussianSampling} could find a partial coloring $\bx\in K\cap [-\delta,\delta]^{n}$ so that $\bx$ has at least $\epsilon \cdot n$ many entries in $\{\pm \delta\}$ with probability at least $1-e^{-\Theta_{\epsilon,\delta}(n)}.$\end{lemma}

In this section, we fix $\epsilon$ and $\delta$ as defined in Lemma~\ref{lem:Roth17_main}. The proof of Theorem~\ref{thm:prefix_main} relies on several properties of the Gaussian measure. 
%We shall use $\epsilon,\delta$ as fixed constants in the rest of this section.

%\subsection{Proofs of Theorem~\ref{thm:prefix_main}}\label{sec:pf_prefix_main}

%        \item $\Pr[\underset{k\in[n]}{\max}\{\bs_{k}-\mu(\bs_{k}-\bs_{n})\}\ge t]\le 2\cdot\Pr[\bs_{n}\ge t],$
%        \item \label{levy_ineq_2}        $\Pr[\underset{k\in[n]}{\max}\{|\bs_{k}-\mu(\bs_{k}-\bs_{n})|\}\ge t]\le 2\cdot\Pr[|\bs_{n}|\ge t],$
\paragraph{Gaussian random variables.} 
We recall several useful facts about the Gaussian random variables. The first one is L\'evy's inequality stated as Lemma~\ref{lem:levy ineq}.
The second one is the \v{S}id{\'a}k-Khatri inequality about the correlation of two bodies in Gaussian measures \cite{vsidak1967rectangular,khatri1967certain,royen2014simple}.
%\begin{claim}\label{cla:Gaussian_tail}
%    For the Gaussian random variable $N(0,\sigma^2)$ and deviation $\delta \cdot \sigma$, we have 
%\begin{equation}   \Pr[|N(0,\sigma^2)| \le \delta \cdot \sigma] \ge \left\{ \begin{aligned} & \delta/5 & \textit{ if } \delta \le 1,\\ & 1 - \frac{e^{-\delta^2/2}}{\sqrt{2 \pi} \cdot \delta} & \textit{ if } \delta>1. \end{aligned} \right.      \end{equation}
%\end{claim}
% {\v{S}}id{\'a}k's inequality

% \hynote{the name of  Sidak is not proper here. It's Gaussian correlation inequality.}

\begin{lemma}
\label{lem:SidakKhatri}
Let $T_1 \subseteq \R^n$ and $T_2 \subseteq \R^n$ be any two symmetric convex bodies. 
Then $\gamma_n(T_1 \cap T_2) \geq \gamma_n(T_1) \cdot \gamma_n(T_2)$.
\end{lemma}

The new ingredient is a lower bound on the probability of a sequence of Gaussian random variables (a.k.a. a Gaussian process) with small deviations. While we only state the lower bound, this estimation is tight up to constants in the exponent. \cite{chung1948maximum, Shao93}. 

\begin{lemma}
\label{lem:gaussian_lower_bound}
    Let $\{X_i\}_{i\in[n]}$ be $i.i.d$ Gaussian random variables, with $\E[X_i]=0,\E[X_i^2]=1$.
    Let $S_n=\sum_{i=1}^nX_i$ and $S_n^*=\underset{1\leq i\leq n}{\max}|S_i|$.
    Then for any $\ell$, we have
    $$\Pr[S^*_n<4\sqrt{\ell}]\geq 4^{-\frac{n}{\ell}}. $$
\end{lemma}
In fact, this lower bound also holds for i.i.d symmetric sub-Gaussian random variables such as $X_i \sim \{\pm 1\}$. We use the deviation bound of Lemma~\ref{lem:gaussian_lower_bound} for $\ell<\sqrt{n}$.
For completeness, we provide an elementary proof in Appendix~\ref{sec:pf_gaussian_lower_bound}. 

While Lemma~\ref{lem:SidakKhatri} and Lemma~\ref{lem:gaussian_lower_bound} together provide a lower bound on $\gamma_{I}(Q)$,  Function~\textsc{GaussianSampling} takes $T_{\bx(I)}(Q)$ as the convex body given the starting point $\bx$. Hence we need the following lemma to lower bound $\gamma_I\big( T_{\bx(I)}(Q) \big)$ by $\gamma_I(Q)$.

\begin{lemma}[Corollary 14 of \cite{Rothvoss17}] \label{lem:sketch_gau_measure}
Let $K\in \R^{n}$ be a symmetry convex body and $\lambda\in \R^{n}$. Then \[\Pr_{\bg\in N(0,1)^n}\left[ \big(\lambda(1)\bg(1),\ldots, \lambda(n)\bg(n) \big)\in K \right]  \ge \frac{1}{\prod_{i=1}^{n}\max\{1,|\lambda(i)|\}} \Pr_{\bg\in N(0,1)^n}[\bg\in K].  \] 
\end{lemma}

% \begin{lemma}[\cite{chung1948maximum},\cite{PRP59}]
% \label{lem:small_ball_probability} Suppose $\bg(1)\ldots \bg(n)$ are a series of independent Gaussian variables with mean $0$ and variance $1$. Let $\sigma=\sqrt{\E[(\sum_{i=1}^{n}\bg(i))^{2}]}$. For any $0<z\le 1$, 
% \begin{align*}
%     \Pr\Big[ \underset{1\le t\le n}{\max} \Big|\sum_{i=1}^{t} \bg(i)\Big| \le  z\cdot \sigma \Big]  \ge   \frac{2}{\pi}\cdot \text{exp}\Big(\frac{-\pi^{2}}{8\cdot z^{2}}\Big) \ge e^{-2/z^{2}}
% \end{align*}
% \end{lemma}

% \begin{lemma}[\cite{Rothvoss17}] \label{lem:sketch_gau_measure}
% Let $K\in \R^{n}$ be a symmetry convex body and $\lambda\in \R^{n}$. Then \[\Pr_{\bg\in N(0,I)}\Big[(\lambda(1)\bg(1),\ldots, \lambda(n)\bg(n))\in K\Big]  \ge \frac{1}{\prod_{i=1}^{n}\max\{1,|\lambda(i)|\}} \Pr_{\bg\in N(0,I)}[\bg\in K].  \] 
% \end{lemma}

\begin{proofof}{Theorem \ref{thm:prefix_main}} Define the convex body satisfying the prefix constraint of row $j$ as
\begin{align*}
Q_j& := \bigg\{\bg\in \R^{I} : \underset{1\le k\le |I|}{\max} \big| \sum_{i=I(1)}^{I(k)}\bg(i)\cdot \bA(j,i)\big| \le 112\sqrt{m}\bigg\}.
\end{align*}
Then $Q =\cap_{j=1}^m Q_j$ and $T_{\bx(I)}(Q)=\cap_{j=1}^mT_{\bx(I)}( Q_j)$. It is obvious that $Q$ and $Q_j$ are symmetric convex bodies.
Then,
\begin{align}
  \gamma_{I}(T_{\bx(I)}(Q_j))&= \Pr_{\bg\in N(0,1)^I}\bigg[\max_{1\le k\le |I|}\Big|\sum_{i=I(1)}^{I(k)}(1-\bx(i))\cdot\bg(i)\cdot \bA(j,i)\Big|\le 112\sqrt{m}\bigg]\notag\\  
  &\ge \Pr_{\bg\in N(0,1)^I}\bigg[\max_{1\le k\le |I|}\Big|\sum_{i=I(1)}^{I(k)}\bg(i) \Big|\le 112\sqrt{m}\bigg] \tag{By Lemma \ref{lem:sketch_gau_measure}, $\bA(j,i)\in[-1,1]$ and $1-\bx(i)\in[-1,1]$}\\
  &\ge e^{-\frac{|I|}{28^2m}}\tag{By Lemma \ref{lem:gaussian_lower_bound}}.
\end{align}
Thus
\begin{equation}\label{eq:Qj_lowerbound}
     \gamma_{I}\left(T_{\bx(I)}(Q_j)\right)\ge e^{-\frac{9}{5000}\cdot |I|/m}.
\end{equation}
% By Lemma \ref{lem:SidakKhatri}, we have
Then we apply Lemma~\ref{lem:SidakKhatri}: 
\begin{align*}
    \gamma_{I}\left(T_{\bx(I)}(Q)\right) &\ge \prod_{j=1}^m\gamma_{I}\left(T_{\bx(I)}( Q_j)\right)\tag{By Lemma~\ref{lem:SidakKhatri}}\\
    & \ge (e^{-\frac{9}{5000}|I|/m})^m\tag{By Equation\eqref{eq:Qj_lowerbound}}\\
    &=e^{-\frac{9}{5000}|I|}.
\end{align*}
%\textcolor{red}{Apply Lemma~\ref{lem:sketch_gau_measure} to bound $\gamma(T_{\bx(I)}(Q))$!!!}
Then by Lemma~\ref{lem:Roth17_main}, \textsc{GaussianSampling} can find a partial coloring $\bx^*\in T_{x(I)}(Q)\cap [-1,1]^{|I|}$ so that $\bx^*$ has at least $\epsilon\cdot |I|$ many entries in $\{\pm 1\}$ with probability $1-e^{-\Omega(|I|)}$.

%W.l.o.g we can assume $I=\{i\in [n]\mid 0\le \bx(i)<1\} $. Otherwise, we consider $-\bx$.

Line 12 and Line 13 of Algorithm~\ref{alg:prefix_partial_coloring} guarantees that the number of $i\in I$ with $\bx^*(i)=1$ is at least $\frac{\epsilon}{2}|I|$ and $\bx'(i)=1$ for all the  $i\in I$ with $\bx^{*}(i)=1$, $\bx'(i)\in [-1,1]$  for $ i\in I$ .  Thus we have $\Big|\big\{i\in [n]\mid \bx'(i) \in \{\pm 1\}\big\}\Big| \ge \frac{\epsilon}{4}\cdot n$, i.e., $\epsilon/4$ fraction of entries in $\bx'$ are $\pm 1$.

%\textcolor{red}{no factor $C$ in front of $Q$!! 
Observe that the update $\bx'(i)-\bx(i)=(1-\bx(i))\cdot \bx^*(i)$ for $i \in I$ guarantees that $\bx'(I)-\bx(I)\in  Q\cap [-1,1]^{I}$ and the update $\bx'(i)-\bx(i)=0$ for $i\in [n] \setminus I$ shows $\bx'([n]/I)-\bx([n]/I)=0$. 
% Therefore $\bx'-\bx\in  Q([n])\cap [-1,1]^{n}$.
Thus, the prefix discrepancy of $\bx-\bx'$ is $O(\sqrt{m}).$

\end{proofof}

\subsection{Proof of Corollary~\ref{cor:prefix_spencer_v2}}\label{sec:proof_cor}
We have the following standard bound for the Gaussian random variable \cite{vershynin2020high}. %Notice for the sub-Gaussian random variable we have a similar tail bound  .
\begin{claim}\label{clm:Gaussian_tail}
    Given the Gaussian random variable $N(0,\sigma^2)$, for any deviation $\delta \cdot \sigma$ with $\delta>1$,
    $$ \Pr[|N(0,\sigma^2)| \ge \delta \cdot \sigma] \le   \frac{e^{-\delta^2/2}}{\sqrt{2 \pi} \cdot \delta }.$$
\end{claim}
% \begin{claim}\label{cla:Gaussian_tail}
%    For the Gaussian random variable $N(0,\sigma^2)$ and deviation $\delta \cdot \sigma$, we have 
% \begin{equation}   \Pr[|N(0,\sigma^2)| \le \delta \cdot \sigma] \ge \left\{ \begin{aligned} & \delta/5 & \textit{ if } \delta \le 1,\\ & 1 - \frac{e^{-\delta^2/2}}{\sqrt{2 \pi} \cdot \delta} & \textit{ if } \delta>1. \end{aligned} \right.      \end{equation}
% \end{claim}

We use Algorithm~\ref{alg:prefix_spencer_discrepancy} to prove Corollary~\ref{cor:prefix_spencer_v2}. It contains two phases. In the first phase, we iteratively apply function \textsc{PrefixPartialColoring} to entries in $(-1,1)$  until there are $O(m)$ entries left in $(-1,1)$. Then in the second phase, we apply function \textsc{PrefixPartialColoring} with a different definition of $Q$ whose upper bound becomes $\sqrt{2|I|\cdot \ln\frac{2000m}{|I|}}$ instead of $112\sqrt{m}$ in every iteration.
\begin{breakablealgorithm}
 \caption{prefix spencer constructive upper bound}\label{alg:prefix_spencer_discrepancy}
     % \hspace*{\algorithmicindent} \textbf{Input}: a symmetric convex body $Q\subset \R^{n}$ \\ 
     % \hspace*{\algorithmicindent}
     % \textbf{Output}: a partial coloring $ \bx^{\star}\in C\cdot Q\cap  [-1,1]^{n}$ with $\Big|\Big\{i: \bx(i)\in \{\pm 1\}\Big\}\Big|\ge \epsilon n$.
     \begin{algorithmic}[1]
     \Function{PrefixSpencer}{$\bA\in \R^{m\times n}$}
     \State $\bx\gets \Vec{0},I'\gets [n]$
    \Repeat \State $\bx\gets$\textsc{PrefixPartialColoring}$(\bx(I'),\bA(\cdot,I'))$
    \State Let $ I'$ be the sequence of indices $i\in [n]$ with $\bx(i) \not\in\{\pm 1\}$ 
    \Until{$|I'|= 50m$}
    \State
    Modify function \textsc{PrefixPartialColoring} by changing line 9 into
    $$Q\gets  \bigg\{\bg\in \R^{I} : \underset{j\in [m]}{\max} \underset{ 1\le k\le |I| }{\max}  \Big|\sum_{i = I(1)}^{I(k)} \bg(i)\cdot \bA(j,i)\Big| \le \sqrt{2|I|\cdot \ln\frac{2000m}{|I|}}\bigg\}$$ and we call the new function as \textsc{PrefixPartialColoringModified}
    \State $t\gets 0$ and $\bx_0 \gets \bx$  ,$\Tilde{I}_0=I'$ 
    \Repeat 
    \State $\bx_{t+1}\gets$\textsc{PrefixPartialColoringModefied}$(\bx_t(\Tilde{I}_t),\bA(\cdot,\Tilde{I}_t))$
    \State Let $ \Tilde{I}_{t+1}$ be the sequence of indices $i\in [n]$ with $\bx_{t+1}(i) \not\in\{\pm 1\}$ 
    \Until{$|\Tilde{I}_{t+1}|=O(\sqrt{m})$}
    \State Color the entries of $\bx_{t+1}$ not in $\{\pm 1\}$ arbitrarily and denote it as $\bx$
    \State \Return $\bx$
     \EndFunction
    
     \end{algorithmic}
 \end{breakablealgorithm}

In the first phase, by Theorem~\ref{thm:prefix_main}, we color  at least $\frac{\epsilon}{4}$  fraction of uncolored entries in each iteration. Thus we need $O(\frac{\ln \frac{n}{m}}{\ln\frac{1}{1-\epsilon/4}})$ iterations to reduce $|I'|$ from $n$ into $O(m)$. The success probability is at least $(1-e^{-\Omega(m)}).$ After the first phase, we get a partial coloring $\bx_0$ with prefix discrepancy $O(\sqrt{m}\cdot \ln\frac{n}{m})$ and at most $ O(m)$ entries in $(-1,1)$.

Then we consider the second phase. 
Denote $I_t \subseteq \tilde{I}_t$ as the sequence of indices $j\in[n]$ with $0\leq \bx_t(j)< 1$ at the $t$-th iteration of the second phase. 
 Similarly we define the convex body $Q_j^t$  and $Q^t$ as follows: 
\begin{align*}
Q_j^t& := \bigg\{\bg\in \R^{I_t} : \underset{1\le k\le |I_t|}{\max} \big| \sum_{i=I_t(1)}^{I_t(k)}\bg(i)\cdot \bA(j,i)\big| \le \sqrt{2|I_t|\cdot \ln\frac{2000m}{|I_t|}}\bigg\}.\\
Q^t& :=\bigg\{\bg\in \R^{I_t} : \underset{j\in [m]}{\max}\underset{1\le k\le |I_t|}{\max} \big| \sum_{i=I_t(1)}^{I_t(k)}\bg(i)\cdot \bA(j,i)\big| \le \sqrt{2|I_t|\cdot \ln\frac{2000m}{|I_t|}}\bigg\}.
\end{align*}

Then we compute,
\begin{align*}
   &\Pr_{\bg\in N(0,1)^{I_t}}\bigg[\max_{1\le k\le |I_t|}\Big|\sum_{i=I_t(1)}^{I_t(k)}\left(1-\bx(i)\right)\cdot\bg(i)\cdot \bA(j,i)\Big|\ge \sqrt{2|I_t|\cdot \ln\frac{2000m}{|I_t|}}\bigg]\\
   &\leq 2\Pr_{\bg\in N(0,1)^{I_t}}\bigg[\Big|\sum_{i=I_t(1)}^{I_t(k)}(1-\bx(i))\cdot\bg(i)\cdot \bA(j,i)\Big|\ge \sqrt{2|I_t|\cdot \ln\frac{2000m}{|I_t|}}\bigg]\tag{By Lemma~\ref{lem:levy ineq}}\\
   &\leq \frac{2}{\sqrt{2\pi}} e^{-\ln\frac{2000m}{|I_t|}} \le \frac{|I_t|}{2000m}. \tag{By Claim~\ref{clm:Gaussian_tail} ,$\bA(j,i)\in [-1,1]$ and $1-\bx(i)\in[-1,1]$}
\end{align*}
%\textcolor{red}{Confused. Where is the factor 2 gone? \hynote{Becaue in the gaussian tail bound, there is an extra $\frac{1}{\sqrt{2\pi}}$. See Claim 4.7.} Why $\ln$ not $\ln$?\hynote{Done} Why is it $2|I_t|$ nor $|I_t|$? \hynote{Maybe to cancel the $\frac{1}{2}$ in the exponential factor in gaussian tail bound.} Why is the constant $112$ in the next line?}\hynote{typo,done}

Thus
\begin{equation}\label{eq:gamma_qj_m}
   \gamma_{I_t}\left(T_{\bx(I_t)}(Q_j^t)\right)= \Pr_{\bg\in N(0,1)^{I_t}}\bigg[\max_{1\le k\le |I_t|}\Big|\sum_{i=I_t(1)}^{I_t(k)}(1-\bx(i))\cdot\bg(i)\cdot \bA(j,i)\Big|\le \sqrt{2|I_t|\cdot \ln\frac{2000m}{|I_t|}}\bigg]\ge 1-\frac{|I_t|}{2000m}. 
\end{equation}

Then
\begin{align*}
    \gamma_{I_t}(T_{\bx(I_t)}(Q^t)) &\ge \prod_{j=1}^m\gamma_{I_t}(T_{\bx(I_t)}(Q^t_j))\tag{By Lemma~\ref{lem:SidakKhatri}}\\
    & \ge (1-\frac{|I_t|}{2000m})^m\tag{By Equation~\ref{eq:gamma_qj_m}}\\
    &\ge e^{-2m\cdot \frac{|I_t|}{2000m}}\tag{For $0<x<\frac{1}{2}$, we have $1-x\ge e^{-2x}$}\\
    &\ge e^{-\frac{9}{5000}\cdot |I_t|}.
\end{align*}

Similar to the proof of  Theorem~\ref{thm:prefix_main},
we color at least $\frac{\epsilon}{4}$  fraction of uncolored entries for each iteration. Thus we need another $O(\frac{\ln m}{\ln \frac{4}{4-\epsilon}})$ iterations to make $|\Tilde{I}_t| =O(\sqrt{m})$. Besides, $\bx_{t}(I_t)-\bx_{t+1}(I_t)\in Q^t\cap [-1,1]$. 
Thus, the prefix discrepancy of $\bx_t-\bx_{t+1}$ is $\sqrt{2|I_t|\cdot \ln\frac{2000m}{|I_t|}}$ for each iteration.

Since in the $t$-th iteration, we turn at least $\frac{\epsilon}{4}|\Tilde{I}_t|$ entries of $\Tilde{I}_t$ into $1$,
we obtain $|\Tilde{I}_{t+1}|\leq(1-\frac{\epsilon}{4})|\Tilde{I}_{t}|$. Thus
$|\Tilde{I}_{t}|\leq (1-\frac{\epsilon}{4})^{t}|\Tilde{I}_0|$. Besides, $I_t\subset \Tilde{I}_t$.

% Then by Lemma~\ref{lem:Roth17_main}, \textsc{GaussianSampling} can find a partial coloring $\bx^*\in T_{x(I)}(Q)\cap [-1,1]^{|I|}$ so that $\bx^*$ has at least $\epsilon\cdot |I|$ many entries in $\{\pm 1\}$ with probability $1-e^{\Omega(|I|)}$.
Therefore,
\begin{align*}
    \sqrt{2|I_{t}|\cdot \ln\frac{2000m}{|I_{t}|}}&\leq \sqrt{2|\Tilde{I}_{t}|\cdot \ln\frac{2000m}{|\Tilde{I}_{t}|}}\tag{$x\cdot\ln\frac{2000m}{x} $ increases on $x$  for $0<x<100m$}\\
    &=\sqrt{2|\Tilde{I}_{0}|\cdot \ln\frac{2000m}{|\Tilde{I}_{0}|}}\cdot \sqrt{\frac{|\Tilde{I}_{t}|}{|\Tilde{I}_{0}|}}\cdot\sqrt{\frac{\ln\frac{2000m}{|\Tilde{I}_{t}|}}{\ln\frac{2000m}{|\Tilde{I}_{0}|}}}.
\end{align*}
Notice for $0< |\Tilde{I}_t|\leq |\Tilde{I}_0|<100m$, we have
$$\sqrt{\frac{\ln\frac{2000m}{|\Tilde{I}_{t}|}}{\ln\frac{2000m}{|\Tilde{I}_{0}|} }}\leq \left(\frac{|\Tilde{I}_{0}|}{|\Tilde{I}_{t}|}\right)^{\frac{1}{6}}.$$
Thus we have 
\begin{equation}\label{eq:geometry_constant}
    \sqrt{2|I_{t}|\cdot \ln\frac{2000m}{|I_{t}|}}\leq \sqrt{2|\Tilde{I}_{0}|\cdot \ln\frac{2000m}{|\Tilde{I}_{0}|}}\cdot (1-\frac{\epsilon}{4})^{\frac{t}{3}}.
\end{equation}

Suppose we stop at time $T=O(\ln \frac{n}{m}+\ln m)=O(\ln n )$,
 the prefix discrepancy of $\bx_0-\bx_T$ should be 
\begin{align*}
\sup_{\ell \le n} \|\bA(\cdot,[\ell]) \cdot \big(\bx_0-\bx_T \big)\|_{\infty} &\leq \sum_{i=0}^{T-1}\sup_{\ell \le n} \|\bA(\cdot,[\ell]) \cdot \big(\bx_i-\bx_{i+1} \big)\|_{\infty}\\
    &\leq \sum_{i=0}^{T-1}\sqrt{2|I_i|\cdot \ln\frac{2000m}{|I_i|}}\\
    &\leq \sqrt{2|\tilde{I}_0|\cdot \ln\frac{2000m}{|\tilde{I}_0|}}\sum_{i=0}^\infty(1-\frac{\epsilon}{4})^{\frac{i}{3}}=O(\sqrt{m}).\tag{Notice $|\Tilde{I}_0|=O(m)$ and Equation~\eqref{eq:geometry_constant}.}
\end{align*}
Then we color the left $\sqrt{m}$ entries to get a full coloring. The prefix discrepancy will only increase by $O(\sqrt{m})$. 

Combine the whole analysis, we shall obtain a coloring $\bx\in\{\pm 1\}^n$ with prefix discrepancy $O(\sqrt{m}\cdot \ln\frac{n}{m})+O(\sqrt{m})+O(\sqrt{m})=O(\sqrt{m}\cdot \ln\frac{n}{m}).$

\section{Random Sampling in Spencer's Setting}\label{sec:sampling_Spencer}
%\xcnote{July 10 --- Check my comments in Section 6.1.}

In this section, we provide a random sampling algorithm in Spencer's setting. This extends the classical linear algebraic framework by incorporating leverage scores such that the min-entropy of the output is $\Omega(n)$. %such that the algorithm could safely fix an entry to $+1$ or $-1$ arbitrarily.

\begin{theorem}\label{thm:spencer_sampling}
    For any $n \ge m$, there exist a constant $C$ and an efficient randomized algorithm such that for any input $\bA \in [-1,1]^{m \times n}$,
    \begin{enumerate}
        \item its output $\bx$ always satisfies $\|\mathbf{A} \bx\|_{\infty} \le   C \sqrt{m}$.
        \item for any $\boldsymbol{\epsilon} \in \{\pm 1\}^n$, the probability of   $\bx=\boldsymbol{\epsilon} $ is $ O(1.9^{-0.9n})$.
    \end{enumerate} 
\end{theorem}

\begin{remark}
Spencer \cite{Spencer_six} has shown that $\forall K\ge 6$, $\exists \beta_K<1$ such that there are $(2-\beta_K)^n$ solutions satisfying $\|\bA \bx\|_{\infty}\le K \sqrt{n}$ for any $\bA \in \{0,1\}^{n \times n}$. We could generalize Theorem~\ref{thm:spencer_sampling} to guarantee that the probability of outputting any $\boldsymbol{\epsilon} \in \{\pm 1\}^n$ is $\le (2-\beta)^{-n}$ for any $\beta>0$. This is done by relaxing the discrepancy to $\|\bA \bx\|_{\infty} \le C_{\beta} \cdot \sqrt{n}$ for some constant $C_\beta$ and resetting parameters $\delta = \beta / 4, T = n \cdot (1 - \beta / 4),  \gamma = O(\beta^3)$, and $\theta = O(\beta^4)$. %For ease of exposition, we did not attempt to optimize the sampling probability of $O(1.9^{-0.9n})$. 

%$\delta < \beta / 2, T = n \cdot (1 + \delta - \beta / 2), \eta = 0.1, \gamma = \delta^2(\beta - 2 \delta) / 10, \theta = \delta^2(\beta - 2 \delta)^2 / 50$.

On the other side, Spencer \cite{Spencer_six} also showed that for any $C=O(1)$, $\exists \bA \in \{0,1\}^{n \times n}$ such that the number of $\bx \in \{\pm 1\}^n$ with $\|\bA \bx\|_{\infty} \le C \sqrt{m}$ is exponentially small compared to $2^n$.

Finally, there are weaker sampling algorithms whose output has the Shannon entropy $0.5n$ instead of the min-entropy\footnote{We provide a description in Appendix~\ref{sec:counting} and thank the anonymous referee to suggest it.}. However, for the Shannon entropy, it is possible that this sampling algorithm will output a fixed coloring with probability $0.9$; but its output still has the Shannon-entropy $\ge 0.5n$.
\end{remark}

Our algorithm has two parts. When $n \ge C^* \cdot m$ for a fixed constant $C^* \ge 10^6$ in this section, we apply Algorithm~\ref{alg:sample_spencer_big_n} directly to sample a random coloring $\bx\in\{\pm 1\}^{n}$. Otherwise we apply Algorithm~\ref{alg:sample_spencer_small_n} (with $\binom{\bA}{-\bA}$ instead of $\bA$) to sample $\bx$. We state the correctness of Algorithm~\ref{alg:sample_spencer_big_n} for $n \ge C^* \cdot m$ as follows.

\begin{lemma}\label{lem:spencer_big_n}
    Given $\bA \in [-1,1]^{m \times n}$ where $n \ge C^* \cdot m$, with probability $0.99$  over $\br$, Function~\textsc{ColumnReductionSampling} in  Algorithm~\ref{alg:sample_spencer_big_n} returns $\bx_T \in [-1,1]^n$ with $|\{ i \in [n] : |\bx(i)| = 1 \}| = 0.9n$ and $\| \bA\bx_T\|_\infty = 0$.

    Moreover, for different seeds $\br$ and $\br^{\prime}$, $\bx_T \gets \textsc{ColumnReductionSampling}(\br)$ and $\bx'_T \gets \textsc{ColumnReductionSampling}(\br')$ have a different $\{\pm 1\}$-entry.    
\end{lemma}
While Algorithm~\ref{alg:sample_spencer_big_n} is still in the classical linear-algebraic framework of the partial coloring method, the key difference compared to previous works \cite{Bansal10,LovettM12,BansalG17} is to find an safe entry $k_t$ such that it could fix $\bx_{t+1}(x_t)$ to $+1$ or $-1$ arbitrarily. We obtain this entry by considering its leverage scores in the safe subspace $H_t$. Recall that the key properties of the leverage scores are stated in Claim~\ref{clm:leverage_socres_subspace}.

Then we state the correctness of Algorithm~\ref{alg:sample_spencer_small_n} for $n < C^* \cdot m$. This algorithm combines the deterministic discrepancy algorithm in \cite{LevyRR17} with our leverage score sampling technique. One particular reason of choosing this algorithm in \cite{LevyRR17} is because its update vector $\alpha_t \mathbf{u}_t$ could have length $\ge 1$ in \cite{LevyRR17} such that $\bx_{t+1}(x_t)$ could be either $1$ or $-1$ after the update.

\begin{algorithm}[h]
\caption{Reduction Sampling in Spencer's Setting when $n\ge C^*\cdot m$} \label{alg:sample_spencer_big_n}
\hspace*{\algorithmicindent} \textbf{Input}: $\mathbf{A}\in [-1,1]^{m \times n}$, $\eta =0.1$, $T=0.9n$,     $\gamma=0.0001$, $\delta=0.05$. \\ 
\hspace*{\algorithmicindent}
\textbf{Random seed} $\mathbf{r} \in \{\pm 1\}^{T}$ will be generated on the fly.\\
\hspace*{\algorithmicindent}
\textbf{Output}: $\bx_T \in [-1,1]^{n}$.

\begin{algorithmic}[1]
\Function{ColumnReductionSampling}{$\br$}
\State $\mathbf{x}_0=\vec{0}$.
\For{$0\le t\le T-1$}
 \State $\mathcal{F}_{t}\leftarrow \{i\in [n]:|\bx_{t}(i)|=1\}$.
 \State $\mathcal{D}_{t}\leftarrow \{i\in [n]:|\bx_{t}(i)|\ge 1-\eta\}$.
 % \State $G_{t}\leftarrow \{\bA(j,[n]\setminus \mathcal{D}_t):j\in [m]\} $.
 \State $H_{t}:=$ the subspace in $\mathbb{R}^{[n]\setminus \mathcal{D}_t}$ orthogonal to $\mathbf{x}_t([n] \setminus \mathcal{D}_t)$ and $\{\bA(j,[n]\setminus \mathcal{D}_t):j\in [m]\}$.

 \State Compute the leverage scores of $\tau_i(H_t)$ for each $i \in [n] \setminus \mathcal{D}_t$.
 
 \State  Find the first coordinate $k_{t} \notin \mathcal{D}_t$  such that $\tau_{k_t}(H_t) \ge 1-\gamma$ and $|\bx_t(k_t)| \le \delta$ --- return $\bot$ if such a coordinate does not exist.  %\hynote{"quit” seems not defined, quit from all the process or just this procedure/round?}
 
 \State Let $\mathbf{u}_{t}$ be the vector $\mathbf{u}_t(k_t)=1$ and $\|\mathbf{u}_t\|_2^2=1/\tau_i$ (see Algorithm~\ref{alg:leverage_score_alg} for details). 

 \State For $\delta_t := \bx_t(k_t)$, set $\br(t+1)=1$ with probability $\frac{1+\delta_t}{2}$; otherwise $\br(t+1)=-1$.
 
 \State Choose $\alpha_{t} \in \mathbb{R}$ such that $\mathbf{x}_{t+1} \gets \mathbf{x}_{t} + \alpha_{t} \cdot \mathbf{u}_{t}$ satisfies $\mathbf{x}_{t+1}(k_{t})=\mathbf{r}(t+1)$. %\xcnote{Change $C_t$ to $\alpha_t$ and modify the corresponding proof}
\EndFor
%\State Get $\bx^{\star }\in \{\pm 1 \}^{n}$ by coloring each entry $\bx_T(i) \in (-1,1)$ to $\{\pm 1\}$ arbitrarily.
%\State Get $\bx^{\star }\in \{\pm 1 \}^{n}$ by applying Theorem~\ref{thm:spencer} to color all entries of $\mathbf{x}_T$ not in $\{\pm 1\}$. \hynote{should be moved out of procedure}
%\State Return $\bx^{\star }$.
\State \Return $\bx_T$
\EndFunction
%\Procedure{Main}{}
    %\hynote{Is $1-ne^{-n^{0.4}}$ enough? maybe we don't need this loop.}
%    \Repeat
%        \State Randomly sample $r \sim \{0,1\}^n$
%        \State $\bx^{\star} \gets \textsc{OnePadSampling}$.
%        \State If $\bx^{\star} \neq \bot$,  return $\bx^{\star}$.
%    \Until $n$ times
%     \State Apply Theorem~\ref{thm:spencer} to generate an $\bx \in \{\pm 1 \}^{n}$,  return $\bx$.
%\EndProcedure 
\end{algorithmic}
\end{algorithm}

\begin{lemma} \label{lem:spencer_small_n}
    Given $\bA \in [-1,1]^{m \times n}$ where $n \le C \cdot m$ for any constant $C =O(1)$, with probability $0.99$  over $\br$, Function~\textsc{SamplingSpencer} in Algorithm~\ref{alg:sample_spencer_small_n} returns $\bx_T \in [-1,1]^n$ with $|\{ i \in [n] : |\bx(i)| = 1 \}| = 0.9n$ and $\langle \bA(j,\cdot), \bx_T \rangle = O(\sqrt{m})$ for any $j \in [m]$. 

    Moreover, for different seeds $\br$ and $\br'$, $\bx_T \gets \textsc{SamplingSpencer}(\bA,\br)$ and $\bx'_T \gets \textsc{SamplingSpencer}(\bA,\br')$ have a different $\{\pm 1\}$-entry.
\end{lemma}

\begin{algorithm}[h]
    \caption{Sampling in Spencer's Setting when $ n\le C\cdot m$}
    \label{alg:sample_spencer_small_n}
    \hspace*{\algorithmicindent} \textbf{Input}: $\mathbf{A}\in [-1,1]^{m \times n}$, $\rho = 1/\sqrt{2n}, \theta = 5 \cdot 10^{-7}, \eta =  0.1,  \gamma=0.0001, \delta=0.05$ and $T = 0.9 n$. \\ 
    \hspace*{\algorithmicindent}
    \textbf{Random seed} $\mathbf{r} \in \{\pm 1\}^{T}$ will be generated on the fly. \\
    \hspace*{\algorithmicindent}
    \textbf{Output}: $\bx_T \in [-1,1]^{n}$.
    
    \begin{algorithmic}[1]
    \Function{SamplingSpencer}{$\br$}

        \State $\mathbf{x}_0=\vec{0}$.
        \State $w_i^{(0)} \leftarrow n / m$ for $i \in [m]$. 
        \For{$t = 0$ \textbf{to} $T-1$}
            \State $\mathcal{F}_{t} \leftarrow \{i\in [n]:|\bx_{t}(i)|=1\}$.
            \State $\mathcal{D}_{t} \leftarrow \{i\in [n]:|\bx_{t}(i)|>1-\eta\}$.
            \State Let $I^{(t)} \subset [m]$ be the set of heavy constraints whose weights $w_i^{(t)}$ are the largest $\theta \cdot n$ ones. 
            \State  $U_1^{(t)} \leftarrow \left\{ \mathbf{u} \in \mathbb{R}^{[n]\setminus \mathcal{D}_t} : \big \langle \mathbf{u}, \mathbf{A}(i,[n]\setminus \mathcal{D}_t) \big \rangle = 0 \text{ for } i \in I^{(t)} \right \}$ % where $I^{(t)} \subseteq [m]$ contains the $|I^{(t)}| = \theta n$ indices with maximum weight $w_i^{(t)}$.
            \State $U_2^{(t)} \subseteq \mathbb{R}^{[n]\setminus \mathcal{D}_t}$ be the linear subspace of $\textrm{span}\{ \phi_j : \theta n \leq j \leq n\}$ where $\phi_1,\ldots,\phi_n$ are eigenvectors of $\mathbf{M}^{(t)} \leftarrow \sum_{i=1}^m w_i^{(t)} \cdot \mathbf{A}\big(i, [n]\setminus \mathcal{D}_t \big) \cdot \mathbf{A}\big(i, [n]\setminus \mathcal{D}_t \big)^{\top}$ whose eigen-decomposition is $\sum_j \mu_j \cdot \phi_j \cdot \phi_j^{\top}$ with $\mu_1 \ge \mu_2 \ge \cdots \ge \mu_n$.
            \State $U_3^{(t)} \leftarrow \left\{ \mathbf{u} \in \mathbb{R}^{[n]\setminus \mathcal{D}_t} \big| \big \langle \mathbf{u}, \sum_{i=1}^m w_i^{(t)}  \cdot \mathbf{A}(i,[n]\setminus \mathcal{D}_t) \big \rangle = 0 \right\}$.
            \State $H_t \gets U_1^{(t)} \cap U_2^{(t)} \cap U_3^{(t)} \cap$ (dual space of $\bx_t$ in $R^{[n] \setminus \mathcal{D}_t}$). %\drnote{miss a constraint that $H_t \bot \bx_{t}$.}
            \State Compute the leverage scores $\tau_i(H_t)$ for each $i \in [n] \setminus \mathcal{D}_t$.
            \State  Find the first coordinate $k_{t} \notin \mathcal{D}_t$  such that $\tau_{k_t}(H_t) \ge 1-\gamma$ and $|\bx_t(k_t)| \le \delta$ --- return $\bot$ if such a coordinate does not exist
            
\State Let $\mathbf{u}_{t}$ be the vector $\mathbf{u}_t(k_t)=1$ and $\|\mathbf{u}_t\|_2^2=1/\tau_i$ (see Algorithm~\ref{alg:leverage_score_alg} for details).     \State For $\delta_t:=\bx_t(k_t)$, randomly set $\br(t+1)=1$ with probability $\frac{1+\delta_t}{2}$ and $\br(t+1)=-1$ with probability $\frac{1-\delta_t}{2}$.
\State Choose $\alpha_t \in \mathbb{R}$ such that $\mathbf{x}_{t+1} \gets \mathbf{x}_{t}+ \alpha_t \mathbf{u}_{t}$ satisfies $\mathbf{x}_{t+1}(k_t)=\mathbf{r}(t+1)$. 
\State $w_i^{(t+1)} \leftarrow w_i^{(t)} \cdot \exp \left[ \rho \cdot {\left<\alpha_t \mathbf{u}_{t}, \mathbf{A}(i,\cdot) \right>} \right] \cdot \exp\left[-2/(\theta n)\right]$.
        \EndFor
\State \Return $\bx_T$.
%\State Apply Theorem~\ref{thm:spencer} to color all entries of $\bx_T$ not in $\{\pm 1\}$ to get $\bx \in \{\pm 1 \}^{n}$. \hynote{This line should be moved out of procedure.}
    \EndFunction
%    \Procedure{Main}{}
%        \Repeat
%        \State $\bx \gets \textsc{SamplingSpencer}(\mathbf{A})$.
%        \State If $\bx \neq \bot$,  return $\bx$.    \Until $2n$ times
%        \State Apply Theorem~\ref{thm:spencer} to color all entries of $\mathbf{x}$ not in $\{\pm 1\}$ to get $\bx^{\star }\in \{\pm 1 \}^{n}$.
%        \State Return $\bx^{\star}$.
%    \EndProcedure 
    \end{algorithmic}
\end{algorithm}

Both proofs are in the same framework --- the key step is to  show the existence of $k_t$ for $T=0.9n$ steps, which is proved by a strong martingale concentration bound on $\|\bx_t\|_2$. 
For ease of exposition, we present the proof of Lemma~\ref{lem:spencer_big_n} in Section~\ref{sec:proof_sample_spencer_big_n} before the proof of Lemma~\ref{lem:spencer_small_n} in Section~\ref{sec:proof_sample_spencer_small_n}. Finally  we  finish the proof of Theorem~\ref{thm:spencer_sampling} in as follows.

\begin{proofof}{Theorem~\ref{thm:spencer_sampling}} 
While both functions could fail and output $\bot$, one could repeat Function~\textsc{ColumnReductionSampling} (when $n\ge C^*\cdot m$) or  Function~\textsc{SamplingSpencer} (when $n\le C^*\cdot m$) $n$ times. Thus with probability at least  $1-0.01^n$ we could get   a partial coloring $\bx_T\in [-1,1]^{n}$.  Otherwise, if all the $n$ repeats fail and return $\bot$, we apply Theorem~\ref{thm:spencer} directly to output $\bx\in\{\pm 1\}^{n}$ with $\|\bA\bx\|_{\infty}\le O(\sqrt{m})$.

Then we bound the discrepancy for the two cases $n<C^* \cdot m$ and $n \ge C^* \cdot m$ separately.
When $n\le C^*\cdot m$,  note that we run Function~\textsc{ColumnReductionSampling} on $\bA^{\prime}:= \binom{\bA}{-\bA}$ instead of  $\bA$. By Lemma~\ref{lem:spencer_small_n}, $\underset{j\in[2m]}{\max} \langle \bA^{\prime}(j,\cdot), \bx_{T}\rangle = O(\sqrt{m}).$ Then we  apply Theorem~\ref{thm:spencer} on $\bx_T$ and $\bA'$ to get the full coloring $\bx\in\{\pm 1\}^{n}$ with $\|\bA'(\bx-\bx_T)\|_{\infty}\le O(\sqrt{m})$.  So the discrepancy is $$\|\bA \bx\|_{\infty} = \max_{j\in [2m]}\langle \bA^{\prime}(j,\cdot), \bx \rangle \le \max_{j\in [2m]}\langle \bA^{\prime}(j,\cdot), \bx_{T} \rangle+  \|\bA^{\prime}(\bx-\bx_T)\|_{\infty}\le O(\sqrt{m}).$$

When $n\ge C^*\cdot m$, we apply Theorem~\ref{thm:spencer} on $\bx_T$ and $\bA$ to get full coloring $\bx$ as well. By Lemma~\ref{lem:spencer_big_n}, it's clear that $\|\bA \bx\|_{\infty}\le \|\bA \bx_T\|_{\infty}+\|\bA (\bx-\bx_T)\|_{\infty}\le O(\sqrt{m})$.

Next, we discuss the probability $\Pr[\bx=\boldsymbol{\epsilon}]$. Let us consider the case $n \ge C^* \cdot m$.  Lemma~\ref{lem:spencer_big_n} guarantees that 
for any  two different random seeds $\br$ and $\br^{\prime}$, $\bx_T\neq \bx_T^{\prime}$.  Let $\bx$ and $\bx^{\prime}$ be the full colorings obtained by applying Theorem~\ref{thm:spencer} on $\bx_T$ and $\bx_T^{\prime}$ respectively. By the guarantee of Theorem~\ref{thm:spencer}, $\bx \neq \bx'$ for different seeds $\br$ and $\br'$. Hence we bound the probability of generating any $\br \in \{\pm 1\}^T$. Recall that for any  $0\le t\le T-1$,   $\delta_t:=\bx_t(k_t)$ is $\le \delta:=0.05$. From Line 10 of Algorithm~\ref{alg:sample_spencer_big_n},  
$$\br(t+1)=
\begin{cases}
1& \text{with probability\;} \frac{1+\delta_t}{2},\\
-1& \text{with probability\;} \frac{1-\delta_t}{2}.
\end{cases}$$
So for any  $\boldsymbol{\epsilon} \in \{\pm 1\}^n$, 
$\Pr[\bx=\boldsymbol{\epsilon}] \le (\frac{1+\delta}{2})^{T}\le 0.525^{0.9n}$. When $n\ge C^*\cdot m$,  
\begin{align*}\Pr[\bx =\boldsymbol{\epsilon}] &= \Pr\big[\bx=\boldsymbol{\epsilon}\wedge \exists \;  \text{one call of }\textsc{ColumnReductionSampling}\;  \text{succeeds}  \big] \\
& + \Pr[\bx=\boldsymbol{\epsilon} \wedge \text{All calls of }\textsc{ColumnReductionSampling} =\bot]\\& \le  0.525^{0.9n}/(1-0.01)+0.01^{n}
\\ &\le  O(1.9^{-0.9n}).
\end{align*} 
When   $n\le C^*\cdot m$, we run Function~\textsc{SamplingSpencer}. By the same argument above,  $\Pr[\bx=\boldsymbol{\epsilon}] \le O(1.9^{-0.9n})$ holds as well.

\end{proofof}

\subsection{Proof of Lemma~\ref{lem:spencer_big_n}}\label{sec:proof_sample_spencer_big_n}

Before proving Lemma~\ref{lem:spencer_big_n}, we state some  basic properties of \textsc{ColumnReductionSampling}. The proof of Lemma \ref{lem:spencer_big_n_iteration} is defered to the end of this section.

\begin{lemma} \label{lem:spencer_big_n_iteration}
The $t$-loop in Function~\textsc{ColumnReductionSampling} satisfies that with probability $1-e^{-\Omega(n^{0.4})}$ over $\br$, for each $t \in [0,T-1]$,
\begin{enumerate}
    \item  $ \exists k_t\in [n]\setminus \mathcal{D}_t$ satisfying $|\bx_{t}(k_t)|\le \delta$ and $\tau_{k_t}(H_t)\ge 1-\gamma$. \label{property:lem_spencer_big_n_iteration_1}
    \item After updating $\bx_{t+1}$, $\bx_{t+1} \in [-1,1]^{n}$ and $|\mathcal{F}_{t+1}|=t+1$.\label{property:lem_spencer_big_n_iteration_2}
    \item After updating $\bx_{t+1}$, $\sum_{i \notin \mathcal{F}_t} |\bx_t(i)|^2 \le \frac{\gamma}{1-\gamma} \cdot t + 8 \cdot \delta \cdot n^{0.7}$.
    \label{property:lem_spencer_big_n_iteration_3}
    \item After updating $\bx_{t+1}$, $|\mathcal{D}_{t+1}|\le 1.000124(t+1)+O(n^{0.7})$. \label{property:lem_spencer_big_n_iteration_4} 
\end{enumerate}
\end{lemma}
In fact, one could show that $99\%$ coordinates in $[n]\setminus \mathcal{F}_t$ could be chosen as the coordinate $k_t$ after rescaling those parameters. To show this, we reset $C^*$ to be large enough and $\gamma$ to be small enough such that $1-10^{-3}$ fraction of coordinates has $\tau_{j}(H_t) \ge 1 - \gamma$. Lemma~\ref{lem:spencer_big_n_iteration} implies that the main term of $\sum_{i \notin \mathcal{F}_t} |\bx_t(i)|^2 \le \frac{\gamma}{1-\gamma} \cdot t + 8 \cdot \delta \cdot n^{0.7}$ is $\frac{\gamma}{1-\gamma} \cdot t$ independent with $\delta$. It means that $1-10^{-3}$ fraction of coordinates also satisfy $|\bx_t(i)|\le \delta$. A union bound shows $99\%$ coordinates in $[n]\setminus\mathcal{F}_t$ could be chosen as the entry $k_t$.

Then we can finish the proof of Lemma~\ref{lem:spencer_big_n}.

\begin{proofof}{Lemma~\ref{lem:spencer_big_n}}  
By a union bound over Lemma~\ref{lem:spencer_big_n_iteration}, Function~\textsc{ColumnReductionSampling} outputs $\bot$ with probability less than $T \cdot e^{-\Omega(n^{0.4})}< 0.01$. Assuming Function~\textsc{ColumnReductionSampling} doesn't output $\bot$ in the following argument.  Lemma~\ref{lem:spencer_big_n_iteration}  shows that the partial coloring $\bx_T \in [-1,1]^n$. %has $T$ entries in $\{ \pm 1 \}$. In fact, those $T$ coordinates are $k_1,\ldots,k_T$ chosen during the for loop. 
Denote  $\{\bA(j,[n]\setminus \mathcal{D}_t):j\in [m]\}$ by $G_t$. Because in \textsc{ColumnReductionSampling} the update $\mathbf{u }_t$ is othogonal to   $G_{t}$ for any $t\in [T-1]$, we have $\|\bA\cdot  \bx_T\|_{\infty}=0$

For two different random seeds $\br$ and $\br^{\prime}$, we only need to show that there is an entry $i \in [n]$ such that $\bx_T:=\textsc{ColumnReductionSampling}(\br)$ and $\bx_T^{\prime}:=\textsc{ColumnReductionSampling}(\br^{\prime})$ are in $\{\pm 1\}$  different on entry $i$. We consider the first entry $j \in [T]$ such that $\br(j) \neq \br^{\prime}(j)$. Since $\br([j-1])=\br^{\prime}([j-1])$, the first $j$ indices chosen in $\textsc{ColumnReductionSampling}(\br)$ and $\textsc{ColumnReductionSampling}(\br^{\prime})$ are the same, say $k_0,\ldots,k_{j-1}$. From Algorithm~\ref{alg:sample_spencer_big_n}, $\bx_T(k_{j-1})=\br(j)$ and $\bx_T^{\prime}(k_{j-1})=\br^{\prime}(j)$. Moreover, those entries are fixed after moment $j$. Thus $\bx_T \neq \bx_T^{\prime}$.

% At last, by the definition of $\br $ on the fly and $\delta = 0.05$, $\Pr[\textsc{SamplingBeckFiala}(\br)=\boldsymbol{\epsilon}] \le (\frac{1+\delta}{2})^{T}\le 0.525^{T}$, for any $\boldsymbol{\epsilon} \in \{\pm 1\}^n$.    
\end{proofof}

In the rest of this section, we finish the proof of Lemma~\ref{lem:spencer_big_n_iteration}. The plan is to apply induction on $t$. The most technical part is a martingale concentration for Property~\ref{property:lem_spencer_big_n_iteration_3} of Lemma~\ref{lem:spencer_big_n_iteration}, whose proof is deferred to the end of this section. %$\sum_{i \notin \mathcal{F}_t} |\bx_t(i)|^2 \le \frac{\gamma}{1-\gamma} \cdot t + 8 \cdot \delta \cdot n^{0.7}$ in the following, 

\begin{lemma}[Martingale]\label{lem:spencer_big_n_martingale} If Lemma~\ref{lem:spencer_big_n_iteration}  holds up to $(t-1)$-loop, then 
$\|\bx_{t}\|_{2}^{2}\le \frac{1}{1-\gamma}\cdot t + 8\cdot \delta \cdot n^{0.7}$ holds with probability $1-e^{-\Omega(n^{0.4})}$. Since $|\mathcal{F}_t|=t$, this is equivalent to \begin{equation} \label{eq:spencer_big_n_iteration_martingale}
    \sum_{i\notin \mathcal{F}_{t}} |\bx_{t}(i)|^{2}\le \frac{\gamma}{1-\gamma} \cdot t+ 8\cdot \delta \cdot n^{0.7}. 
\end{equation} 
\end{lemma}
    
Then we apply induction to prove the rest 3 properties.    
\paragraph{Base case:} For $t=0$, $\mathcal{F}_0=\emptyset$, $\bx_0=\vec{0}$, and $\mathcal{D}_0=\emptyset$. The existence of such $k_0$ and $\mathbf{u}_0$ is direct.

\paragraph{Hypothesis:} Lemma \ref{lem:spencer_big_n_iteration} holds up to $(t-1)$, where $0
\le t-1\le T-1$.

\paragraph{Induction step:} Consider the $t$-loop. 
We prove Property \ref{property:lem_spencer_big_n_iteration_1} by showing the following fact: let $\mathcal{E} := \{ j\in [n]\setminus \mathcal{D}_{t}: \tau_{j}(H_{t})\ge 1-\gamma\} $ and $\mathcal{Q} := \{ \ell \in [n]\setminus \mathcal{F}_t: |\bx_{t}(\ell )|\le \delta \}$, and  we have $|\mathcal{E}|/(n-|\mathcal{D}_{t}|) \ge 0.51$ and  $|\mathcal{Q}|/(n-|\mathcal{F}_{t}|)\ge 0.51$ respectively. 

We show $|\mathcal{E}|/(n-|\mathcal{D}_{t}|) \ge 0.51$ by contradiction. Notice that \begin{equation} \label{ep:spencer_big_n_iteration_1}
\frac{\sum_{j\in [n]\setminus \mathcal{D}_t}\tau_j}{n-|\mathcal{D}_{t}|} =\frac{dim(H_t)}{n-|\mathcal{D}_{t}|} = \frac{n-|\mathcal{D}_{t}|-1-|G_{t}|}{n-|\mathcal{D}_{t}|}\ge 1-\frac{m+1}{n-|\mathcal{D}_t|}\ge1-  0.49\cdot \gamma. 
\end{equation} 
Assume   $|\mathcal{E}| < 0.51\cdot (n-|\mathcal{D}_{t}|)$, we have  
$$\sum_{j \in [n]\setminus \mathcal{D}_{t}} \tau_j< 0.51\cdot(n-|\mathcal{D}_{t}|)+0.49\cdot(n-|\mathcal{D}_{t}|)(1-\gamma)< (1-0.49\cdot \gamma) \cdot (n-|\mathcal{D}_t|).$$ 
Since $\tau_{j}(H_{t})$ is always $\in [0,1]$.  It's contradicted with inequality (\ref{ep:spencer_big_n_iteration_1}).

Then we show $|\mathcal{Q}|/(n-|\mathcal{F}_{t}|)\ge 0.51$ by contradiction as well. We could  verify the following inequality by chosen parameters.

\begin{align}
    &\frac{\gamma}{1-\gamma} \cdot t+ 8\cdot \delta \cdot n^{0.7}\le 0.49\cdot \delta^{2}\cdot (n-t)\label{eq:spencer_big_n_iteration_2}\\
    \iff & \left(\frac{\gamma}{1-\gamma}+ 0.49\cdot \delta^{2}\right)\cdot t \le 0.49\cdot \delta^{2 }\cdot n-8\cdot \delta\cdot n^{0.7},\notag \\
    \Longleftarrow\;\; &0.00133\cdot t\le 0.00125\cdot n-0.4\cdot n^{0.7},\tag{plug $\delta=0.05$ and $\gamma=0.0001$.} \\
    \Longleftarrow\;\; &0\le  2.5\cdot10^{-5}\cdot n-0.4\cdot n^{0.7}. \tag{$n$ is large enough and $t< 0.9n$.}
    \end{align}  
    Combining inequality \eqref{eq:spencer_big_n_iteration_martingale} and \eqref{eq:spencer_big_n_iteration_2}, we have 

\begin{equation}
    \sum_{i\notin \mathcal{F}_{t}} |\bx_{t}(i)|^{2}\le 0.49\cdot \delta^{2} \cdot (n-t).\label{eq:spencer_big_n_iteration_3}
\end{equation} 
Assume $|\mathcal{Q}|<0.51\cdot (n-|\mathcal{F}_{t}|)$. For  $|\mathcal{F}_{t}|=t$ and $|\bx_{t}(i)|\ge 0$, we have 
$$   \sum_{i\notin \mathcal{F}_{t}} |\bx_{t}(i)|^{2} > 0\cdot 0.51\cdot (n-|\mathcal{F}_t|)+ \delta^{2}\cdot  0.49\cdot (n-|\mathcal{F}_t|)=\delta^{2}\cdot  0.49\cdot (n-t),
$$ 
which is contradicted with inequality  \eqref{eq:spencer_big_n_iteration_3}.  For   $\mathcal{F}_{t}\subseteq \mathcal{D}_{t}$, we have   $|\mathcal{Q}|/(n-|\mathcal{D}_{t}|)\ge |\mathcal{Q}|/(n-|\mathcal{F}_{t}|) \ge 0.51$. By the definition of $\mathcal{Q}$ and $\mathcal{E}$, both of them belong to $[n]\setminus \mathcal{D}_t$. Since $|\mathcal{Q}|/(n-|\mathcal{D}_{t}|)\ge 0.51$ and $|\mathcal{E}|/(n-|\mathcal{D}_{t}|)\ge 0.51$, there must exist some coordinate $k_{t}\in [n]\setminus \mathcal{D}_{t}$ satisfying  $k_{t}\in \mathcal{Q}\cap \mathcal{E}$.

Next, we show that Property \ref{property:lem_spencer_big_n_iteration_2} holds for the induction hypothesis on $t$.
For any $j\in [n]\setminus \mathcal{D}_t$ and $j\neq k_t$,  $\frac{|\alpha_{t}\cdot\mathbf{u}_{t}(j)|}{|\alpha_{t}\cdot \mathbf{u}_{t}(k_t)|}\le \sqrt{\frac{\gamma}{1-\gamma}}$   which implies 
$$    |\alpha_{t}\cdot \mathbf{u}_{t}(j)|\le \sqrt{\frac{\gamma}{1-\gamma}}\cdot (1+\delta) \le  0.01< \eta.$$ Because $|\bx_t(j)|<1-\eta$, $|\bx_{t+1}(j)|< 1$. By the definition of $\alpha_{t}$,  $|\bx_{t+1}(k_t)|=1$. For any  $h\in \mathcal{D}_t$, $\bx_{t}(h)$ will not be changed after time $t$. So  $\bx_{t+1}\in [-1,1]^{n}$ and $|\mathcal{F}_{t+1}|=t+1$.

Finally we show Property \ref{property:lem_spencer_big_n_iteration_4} for the induction hypothesis of $t$. By lemma \ref{lem:spencer_big_n_martingale}, with probability  $1-e^{-\Omega(n^{0.4})}$,   $$\sum_{i\notin \mathcal{F}_{t+1}} |\bx_{t+1}(i)|^{2}\le \frac{\gamma}{1-\gamma} \cdot (t+1)+ 8\cdot \delta \cdot n^{0.7}.$$ 
Therefore we have  $$  |\mathcal{D}_{t+1}|\le t+1+ \frac{\sum_{i\notin \mathcal{F}_{t+1}} |\bx_{t+1}(i)|^{2}}{(1-\eta)^{2}} \le 1.000124\cdot(t+1)+O(n^{0.7}).$$

\begin{proofof}{Lemma~\ref{lem:spencer_big_n_martingale}}
      For any $0\le j \le t$, let $A_j=\sum_{i\notin \mathcal{F}_j} |\bx_{j}(i)|^{2}$ and $\mathcal{Y}_{j}$ denote the filtration that consists of $\{\alpha_1,\cdots ,\alpha_j\} $ and $\{\mathbf{u}_1,\cdots ,\mathbf{u}_j\}$. Notice that 
      \begin{align}
        \E\left[A_{j+1}\middle\vert \mathcal{Y}_j\right]&= A_{j}+\E\left[A_{j+1}-A_{j} \middle\vert \mathcal{Y}_j\right]\\
        &= A_{j}+\E\left[\sum_{i\notin \mathcal{F}_{j+1}} |\bx_{j+1}(i)|^{2}-\sum_{i\notin \mathcal{F}_j} |\bx_{j}(i)|^{2} \middle\vert \mathcal{Y}_j\right]\\
        &= A_{j}+\E\left[\|\bx_{j+1}\|_{2}^{2}-\|\bx_{j}\|_{2}^{2}-1\middle\vert \mathcal{Y}_j\right]\\
        &= A_{j}+\E\left[\alpha_{j+1}^{2}\cdot \|\mathbf{u}_{j+1}\|_{2}^{2}-1 
        \middle\vert \mathcal{Y}_j\right]\\
        &= A_{j}+\sum_{s=1}^{j+1}\E\left[\alpha_{s}^{2}\cdot \|\mathbf{u}_{s}\|_{2}^{2}-1 
        \middle\vert \mathcal{Y}_{s-1}\right]-\sum_{s=1}^{j}\E\left[\alpha_{s}^{2}\cdot \|\mathbf{u}_{s}\|_{2}^{2}-1 
        \middle\vert \mathcal{Y}_{s-1}\right].\label{eq:spencer_big_n_martingale_4}
    \end{align}
Let $M_j=A_j-\sum_{s=1}^{j}\E\left[\alpha_{s}^{2}\cdot \|\mathbf{u}_{s}\|_{2}^{2}-1 
\middle\vert \mathcal{Y}_{s-1}\right]$, by equation \eqref{eq:spencer_big_n_martingale_4} we have  

$$\E\left[M_{j+1}\middle\vert \mathcal{Y}_{j}\right]=M_j,$$
which means $\{M_{j}\}_{j\ge 0}$ is a martingale. By the definition of $\mathbf{r}(j+1)$ and $\alpha_{j+1}$, $\alpha_{j+1}\le 1+|\delta_{j+1}|$ and   $$
\E\left[\alpha_{j+1}^{2} \middle\vert \mathcal{Y}_{j}\right] = (1-\delta_{j+1})^{2}\cdot\frac{1+\delta_{j+1}}{2}+ (1+\delta_{j+1})^{2}\cdot \frac{1-\delta_{j+1}}{2} =1-\delta_{j+1}^{2}.
$$ Besides,    $\|\mathbf{u}_{j+1}\|_{2}^{2}\le 1+\frac{\gamma}{1-\gamma}$ and  $\mathbf{u}_{j+1}$ is determined by $\mathcal{Y}_{j}$. Hence 
\begin{align*}
        \left|M_{j+1}-M_{j}\right|&=\bigg|\alpha_{j+1}^{2}\cdot\|\mathbf{u}_{j+1}\|_{2}^{2}-1-\E\left[\alpha_{j+1}^{2}\cdot\|\mathbf{u}_{j+1}\|_{2}^{2}-1\middle\vert \mathcal{Y}_{j}\right] \bigg|,\\
      &\le (1+\frac{\gamma}{1-\gamma})( (1+|\delta_{j+1}|)^{2}-(1-\delta_{j+1}^{2}) ),\\ &\le 8\delta.
\end{align*}
The last inequality holds for $\gamma<0.1$ and $|\delta_{t+1}|\le \delta<1$. Therefore by the Azuma-Hoeffding inequality \cite{mitzenmacher2017probability}, we have $$  \Pr[|M_{t}|\ge \mu]\le 2\cdot \exp\left(\frac{-\mu^{2}}{(8\delta)^{2}\cdot t}\right).$$ 
Setting $\mu = 8\cdot \delta\cdot n^{0.7}$, we can conclude that  for any $0\le t\le 0.9n-1$, with probability at least  $1-2\cdot e^{-n^{0.4}}$,
\begin{equation}
    |M_{t}|\le 8\cdot \delta\cdot n^{0.7} \label{eq:spencer_big_n_martingale_5}.
\end{equation}
Besides, for any $1\le s\le 0.9n-1$,   $\E\left[\alpha_{s}^{2}\middle\vert \mathcal{Y}_{s-1}\right]=(1-\delta_{s}^{2})\le 1$ and $\|\mathbf{u}_s\|_{2}^{2}\le 1+ \frac{\gamma}{1-\gamma}$,

\begin{equation}
\sum_{s=1}^{t}\E\left[\alpha_{s}^{2}\cdot \|\mathbf{u}_{s}\|_{2}^{2}-1 
\middle\vert \mathcal{Y}_{s-1}\right]\le (1+\frac{\gamma}{1-\gamma} )\cdot t- t= \frac{\gamma}{1-\gamma} t. \label{eq:spencer_big_n_martingale_6}
\end{equation}

Because $\sum_{i\notin \mathcal{F}_t} |\bx_{j}(i)|^{2}\ge 0$, we only care about the upper bound. 
Combining inequality\eqref{eq:spencer_big_n_martingale_5} and \eqref{eq:spencer_big_n_martingale_6}, we can conclude that  with probability at least  $1-2\cdot e^{-n^{0.4}}$,  \begin{equation}
\sum_{i\notin \mathcal{F}_t} |\bx_{j}(i)|^{2}= A_t= M_t+ \sum_{s=1}^{t}\E\left[\alpha_{s}^{2}\cdot \|\mathbf{u}_{s}\|_{2}^{2}-1 
\middle\vert \mathcal{Y}_{s-1}\right] \le \frac{\gamma}{1-\gamma} t+ 8\cdot \delta\cdot n^{0.7}.
\end{equation} Because $\sum_{i\in \mathcal{F}_{t}} |\bx_{t}(i)|^{2}=t$,  $\|\bx_{t}\|_{2}^{2}\le \frac{1}{1-\gamma}\cdot t + 8\cdot \delta \cdot n^{0.7}.$

\end{proofof}

\subsection{Proof of Lemma~\ref{lem:spencer_small_n}}\label{sec:proof_sample_spencer_small_n}

The proof of Lemma~\ref{lem:spencer_small_n} contains two parts. First we show that Algorithm~\ref{alg:sample_spencer_small_n} could find $k_t$ with high probability for all the $0\le t\le T-1$. Thus we could get the partial coloring $\bx_T$ with $|\{ i \in [n] : |\bx_T(i)| = 1 \}| = 0.9n$. By the same argument in Lemma~\ref{lem:spencer_big_n}, for  different random seeds $\br$ and $\br'$, there is an entry $i \in [n]$ such that $\bx_T:=\textsc{SamplingSpencer}(\br)$ and $\bx_T^{\prime}:=\textsc{SamplingSpencer}(\br^{\prime})$ are in $\{\pm 1\}$  different on entry $i$. Secondly, we bound $\langle \bA(i,\cdot), \bx_T\rangle$ for all the $i\in [m]$.

To prove the existence of $k_t$, analogue to the role of Lemma~\ref{lem:spencer_big_n_iteration} in the  proof of Lemma~\ref{lem:spencer_big_n}, %and  we could show those properties  also hold in Algorithm~\ref{alg:sample_spencer_small_n}. We state the variant of 
we have Lemma~\ref{lem:spencer_small_n_iteration} for  the proof of Lemma~\ref{lem:spencer_small_n}.

\begin{lemma}\label{lem:spencer_small_n_iteration}
 The $t$-loop in Function~\textsc{SamplingSpencer}  satisfies that with probability $1-e^{-\Omega(n^{0.4})}$ over $\br$, for each $t \in [0,T-1]$,
\begin{enumerate}
    \item  $ \exists k_t\in [n]\setminus \mathcal{D}_t$ satisfying $|\bx_{t}(k_t)|\le \delta$ and $\tau_{k_t}(H_t)\ge 1-\gamma$. \label{property:sampling_spencer_small_n_iteration_1}
    \item After updating $\bx_{t+1}$, $\bx_{t+1} \in [-1,1]^{n}$. $|\mathcal{F}_{t+1}|=t+1$ \label{property:sampling_spencer_small_n_iteration_2}
    \item After updating $\bx_{t+1}$,  $\sum_{i \notin \mathcal{F}_t} |\bx_t(i)|^2 \le \frac{\gamma}{1-\gamma} \cdot t + 8 \cdot \delta \cdot n^{0.7}$.
    \label{property:sampling_spencer_small_n_iteration_3}
    \item After updating $\bx_{t+1}$, $|\mathcal{D}_{t+1}|\le 1.000124(t+1)+O(n^{0.7})$. \label{property:sampling_spencer_small_n_iteration_4} 
\end{enumerate}
\end{lemma}

\begin{proof} The main difference between \textsc{SamplingSpencer} and \textsc{ColumnReductionSampling} is the definition of $H_t$. The proof is partially omitted here since it is the same as that in    Lemma~\ref{lem:spencer_big_n_iteration}. We  apply induction on $t$. 

\paragraph{Base case: }When $t=0$, $\mathcal{F}_0=\emptyset$, $\bx_0=\vec{0}$, and $\mathcal{D}_0=\emptyset$, 
the existence of $k_0$ and $\mathbf{u}_0$ is straightforward.

\paragraph{Hypothesis: }Lemma \ref{lem:spencer_small_n_iteration} holds up to $t-1$ for $0
\le t\le T-1$.

\paragraph{Induction step: } For the $t$-th loop, we prove Property~\ref{property:sampling_spencer_small_n_iteration_1} first. Let $\mathcal{Q} = \{ \ell \in [n]\setminus \mathcal{D}_{t}: |\bx_{t}(\ell )|\le \delta \}$ and $\mathcal{E} = \{ j\in [n]\setminus \mathcal{D}_{t}: \tau_{j}(H_{t})\ge 1-\gamma\}$. By the same argument in  Lemma~\ref{lem:spencer_big_n_iteration},  $|\mathcal{Q}|\ge 0.51 \cdot (n-|\mathcal{D}_{t}|)$. To prove  $|\mathcal{E}|\ge 0.51 \cdot (n-|\mathcal{D}_{t}|)$, we only need to show the following  inequality:  \begin{equation}\label{eq:spencer_small_n_iteraion_1}
    \frac{dim(H_{t})}{n-|\mathcal{D}_{t}|} \ge 1 - 0.49 \gamma.
\end{equation}
By the definition of $H_{t}$,
\begin{align*}
    \dim(H_{t}) & \geq n - |\mathcal{D}_{t}|- |I^{(t)}| -|\textrm{span}\{ \phi_j : 1 \leq j <\theta n \}| - 2 \\
    & \geq n - |\mathcal{D}_{t}| - 2\theta \cdot n - 2.
\end{align*}
Hence,
\begin{equation*}
    \frac{dim(H_{t})}{n-|\mathcal{D}_{t}|} = \frac{n - |\mathcal{D}_{t}| - 2\theta \cdot n - 2}{n-|\mathcal{D}_{t}|}\ge 1- \frac{2\theta \cdot n + 2}{n-|\mathcal{D}_{t}|}.
\end{equation*}
By the induction  hypothesis on $t-1$, $$|\mathcal{D}_{t}|\le 1.000124\cdot t+O(n^{0.7}) \le 0.91\cdot n.$$
By the chosen parameters  $T=0.9n, \theta = 5 \cdot 10^{-7} $ and  $ \gamma=0.0001$,  we have 
$$1- \frac{2\theta \cdot n + 2}{n-|\mathcal{D}_{t}|} \ge 1 - 0.49 \gamma$$ so that Inequality~(\ref{eq:spencer_small_n_iteraion_1}) holds. Consequently,  $|\mathcal{E}| \ge 0.51 (n-|\mathcal{D}_{t}|)$. Therefore $\mathcal{Q} \cap \mathcal{E} \ne \emptyset$, which infers the existence of $k_{t+1}$.

The proofs of other properties are the same as that of  Lemma~\ref{lem:spencer_big_n_iteration} and we omit them.
\end{proof}

By Lemma~\ref{lem:spencer_small_n_iteration} and a union bound on $T$ loops, we have proved that (1)  Function~\textsc{SamplingSpencer} can find the key coordinate $k_{t} \in [n] \setminus \mathcal{D}_t$ for all the $T$ loops with probability $1-T e^{-\Omega(n^{0.4})}>0.99$. (2) When the $T$-th loop is over, we get $\bx_T \in [-1,1]^n$ with $|\{ i \in [n] : |\bx_T(i)| = 1 \}| = 0.9n$.  

%\textcolor{red}{Xue: I do not understanding the role of this claim. Why not finish the proof of LEMMA 5.4?}\hynote{This claim is only a part of lemma 5.4. We need to prove the discrepancy bound in the following arugment.}

%\begin{claim}   Function~\textsc{SamplingSpencer} can find the key coordinate $k_{t} \in [n] \setminus \mathcal{D}_t$ for all the $T$ loops with probability $1-T e^{-\Omega(n^{0.4})}>0.99$. When the $T$-th loop is over, we get $\bx_T \in [-1,1]^n$ with $|\{ i \in [n] : |\bx_T(i)| = 1 \}| = 0.9n$. \end{claim}

To finish the proof of Lemma~\ref{lem:spencer_small_n}, we bound $\langle \bA(i,\cdot), \bx_T \rangle$ similar to the analysis in \cite{LevyRR17}. For completeness, we provide the full proof here. Define $$\Phi^{(t)}:=\sum_{i=1}^m w_i^{(t)}$$ and let $\by_t := \alpha_t \mathbf{u}_t$ in the rest of this section. 
We show that $\Phi^{t}$ is non-increasing in the following claim.

\begin{claim}\label{PotentialFunctionBounded} 
For any $t\in [0,T-1]$, $\Phi^{(t+1)} \le \Phi^{(t)}$. 
\end{claim}
\begin{proof}
      By the line $17$ of Algorithm~\ref{alg:sample_spencer_small_n}, 
    \[ 
        \Phi^{(t+1)} 
        = \sum_{i=1}^m w_i^{(t+1)} = \sum_{i=1}^m w_i^{(t)} \cdot \exp \left[ \rho \cdot {\left<\by_t, \bA(i,\cdot) \right>} \right] \cdot \exp\left[-2/(\theta n)\right]. 
    \]

    To control the factor $\exp \left[ \rho \cdot {\left<\by_t, \bA(i,\cdot)\right>} \right]$, we  show that   $ \left< \by_{t}, \bA(i,\cdot) \right>$ is bounded in the following. 
    In our setting, $|\by_t(k_t)| \le 1 + \delta$ and ${|\by_t(k_t)|^2}/{\| \by_t \|^2_2} \ge 1 - \gamma$.
    Hence, $\| \by_t \|_2 \le \sqrt{(1 + \delta)^2/(1 - \gamma)} \le \sqrt{2}$. 
    Given $\rho =  1/\sqrt{2n}$ and $\|\bA(i,\cdot)\|_{2}\le \sqrt{n}$, $$\rho\cdot {\left< \by_{t}, \bA(i,\cdot) \right>} \le \rho\cdot \| \by_t \|_2\cdot \| \bA(i,\cdot) \|_2 \le 1.$$ 
    Because $e^x\leq 1+x+x^2$ for any  $x\in [-1,1]$, we have 
    
    \[ 
        \Phi^{(t+1)} 
        \le \sum_{i=1}^m w_i^{(t)} 
            \cdot \left[ 1 + \rho \cdot \left<\bA(i,\cdot),\by_t \right> + \rho^2 \cdot  \left<\bA(i,\cdot),\by_t \right>^2 \right] 
            \cdot \exp\left[-2/(\theta n)\right]. 
    \]
    Since $\by_t \in U_3^{(t)}$ guarantees $\left<\by_t, \sum_{i=1}^m w_i^{(t)}\cdot \bA(i,\cdot) \right> = 0$,
    \begin{equation}
    \label{eq:spencer_small_n_non_increasing_1}
         \Phi^{(t+1)} 
        \le \sum_{i=1}^m w_i^{(t)} \cdot \exp\left[-2/(\theta n)\right]  
            + \rho^2 \cdot\sum_{i=1}^m w_i^{(t)} \cdot \left<\bA(i,\cdot),\by_t \right>^2.
    \end{equation}
    In the  Inequality~\eqref{eq:spencer_small_n_non_increasing_1}, 
    \begin{equation}
        \label{eq:spencer_small_n_non_increasing_2}
        \sum_{i=1}^m w_i^{(t)} \cdot \left<\bA(i,\cdot),\by_t \right>^2 = \by_t^T \left( \sum_{i=1}^m w_i^{(t)}\cdot  \bA(i,\cdot)\bA(i,\cdot)^{\top} \right) \by_t = \by_t^T \mathbf{M}^{(t)} \by_t,
    \end{equation}
    where $\mathbf{M}^{(t)}:=\sum_{i=1}^m w_i^{(t)} \bA(i,\cdot) \bA(i,\cdot)^{\top}$. Notice that the sum of eigenvalues of $\mathbf{M}^{(t)}$ is at most $\sum_{i=1}^{m} w_i^{(t)} \|\bA(i,\cdot) \|_{2}^{2}$.  For  $\by_t \in U_2^{(t)}$, the largest $\theta n$ eigenvalues of $\mathbf{M}^{(t)}$ are excluded so that  none of the remaining eigenvalues is greater than $\frac{1}{\theta n} \cdot \sum_{i=1}^{m} w_i^{(t)} \| \bA(i,\cdot) \|^2_2$.
    Consequently, \begin{equation}
        \label{eq:spencer_small_n_non_increasing_3}
        \by_t^{\top}\mathbf{M}^{(t)} \by_t \leq \frac{1}{\theta n} \sum_{i=1}^m w_i^{(t)}\cdot  \| \bA(i,\cdot) \|^2_2  \cdot \| \by_t \|^2_2 \le \frac{2}{\theta} \cdot \sum_{i=1}^m w_i^{(t)}
    \end{equation} for $\|\by_t \|_2 \le \sqrt{2}$ and $\|\bA(i,\cdot)\|_{2}^{2}\le n$. 
    And for sufficiently large $n$, $\exp\left[-2/(\theta n)\right] \le 1 - \frac{1}{\theta n}$. Combining Inequality \eqref{eq:spencer_small_n_non_increasing_1}, \eqref{eq:spencer_small_n_non_increasing_2} and \eqref{eq:spencer_small_n_non_increasing_3}, we have 
    
    \[ \Phi^{(t+1)} 
    \le \sum_{i=1}^m w_i^{(t)} \left( 1 - \frac{1}{\theta n} + \frac{1}{\theta n}\right)  = \Phi^{(t)}. \]
\end{proof}

\begin{claim}\label{clm:bound_disc} For any $i\in[m]$ and $ n\le C\cdot m$, 
$\left<\mathbf{A}(i,\cdot),\bx_{T}\right> \le \rho^{-1} \left[ \ln\left({e m}/{\theta n }\right) + 2T/(\theta n) \right] = O(\sqrt{m})$. 
\end{claim}

\begin{proof}
    By the definition of $w_i^{(t)}$, \[ w_i^{(t)} = \frac{n}{m} \cdot \exp \left[ \rho \cdot {\left<\mathbf{A}(i,\cdot), \bx_t \right>} \right] \cdot \exp\left[-2t/(\theta n) \right]. \]
    After rearrangement,  
    \begin{equation}
    \label{eq:innerproduct_by_potential}
        \left<\mathbf{A}(i,\cdot),\bx_T\right> 
        = \rho^{-1} \left[ \ln w_i^{(T)} 
            + \ln\left(\frac{m}{n}\right)  
            + 2T/(\theta n) \right]. 
    \end{equation}
So  we aim to  bound  $\ln w_i^{(T)}$ next. Recall that $I^{(t)}$ denotes the set of heavy constraints whose weights are the largest $\theta \cdot n$ ones at iteration $t$. It's sufficient to bound $w_i^{(T)}$ for all the  $i \in I^{(T)}$. Let $t^*$ be the last iteration when $i \notin I^{(t^*)}$. (Since initially all weights are equal, we can assume that $i \notin I^{(0)}$.) At the beginning of  Algorithm~\ref{alg:sample_spencer_small_n}, $\Phi^{(0)} = m \cdot \frac{n}{m} = n$. Markov’s inequality and Claim~\ref{PotentialFunctionBounded} infer that $w_i^{(t^*)} \le \frac{1}{\theta n} \Phi^{(t^*)} \le \frac{1}{\theta n} \Phi^{(0)} = \frac{1}{\theta}$.

By the definition of $U_1^{(t)}$, for any $t \ge t^* + 1$,  $\by_t$ is orthogonal to $\mathbf{A}(i,\cdot)$ and  $\exp \left[ \rho \cdot {\left<\by_t, \mathbf{A}(i,\cdot) \right>} \right] = 1$. So after the $(t^*+1)$-th iteration, $w_i^{(t)}$ keeps decreasing. 
     Because $ \exp \left( \rho \cdot {\left< \by_{t^*}, \mathbf{A}(i,\cdot) \right>} \right) \le \exp(1/\sqrt{2n} \cdot \sqrt{2} \cdot \sqrt{n}) \le e$ and $\exp\left[-2/(\theta n) \right] \le 1$,  we have $$ w_i^{(t^* + 1)} = w_i^{(t^*)} \cdot \exp \left[ \rho \cdot {\left<\by_{t^*}, \mathbf{A}(i,\cdot) \right>}\right] \cdot \exp\left[-2/(\theta n) \right]  \le e \cdot w_i^{(t^*)}.$$
    Hence \begin{equation} \label{eq:spencer_small_n_discreaancy_1}
        w_i^{(T)} \le w_i^{(t^* + 1)} \le e \cdot  w_i^{(t^*)} \le e / \theta.
    \end{equation} 

    By plugging Inequality~\eqref{eq:spencer_small_n_discreaancy_1} into  Equation~\eqref{eq:innerproduct_by_potential}, we get the desired upper bound on $\langle \bA(i,\cdot), \bx_T\rangle$.
\end{proof}

\section{Random sampling in Beck-Fiala's and Banaszczyk's settings}
\label{sec:sampling}

%\textcolor{red}{(1) Define Beck-Fiala and state 2 bounds; (2) Subsampling 1 of Beck-Fiala; (3) subsampling 2 of Beck-Fiala from Appendix; (4) define Banaszczyk; (5) sampling of Banaszczyk}

In this section, we discuss  random sampling algorithms in   Beck-Fiala's and Banaszczyk's settings. It is folklore that there are exponentially many good colorings satisfying these the Beck-Fiala bound and Banaszczky's bound (see Appendix~\ref{sec:counting}). Hence it is natural to ask the same question of sampling as Spencer's setting.

\paragraph{Beck-Fiala's   setting.} In this setting, given the sparsity $d$, the input matrix  $\bA\in[-1,1]^{m\times n}$ will have each column of at most $d$ non-zero entries. For simplicity, we  refer to such  a matrix $\bA$ as having degree $d$. %Our goal is to determine a good coloring $\bx\in\{\pm 1\}^{n}$ such that $\|\bA\bx\|_{\infty}$ is as small as possible. 
Beck and Fiala~\cite{Beck_Fiala} proved $\min_{\bx \in\{\pm 1\}^{n}} \|\bA\bx\|_{\infty}=O(d)$, which is independent of $n$ and $m$. Moreover, they conjectured the best bound is $O(\sqrt{d})$. However, the best known bounds for the Beck-Fiala problem are still $O(d)$ \cite{Beck_Fiala} and $O(\sqrt{d \log n})$ \cite{Banaszczyk98}. We defer the discussion of $O(\sqrt{d \log n})$ to the second half of this section. We summarize Beck and Fiala's bound as follows.
%Additionally, an alternative, uncomparable bound of $O(\sqrt{d\log n})$ can be derived from Banaszczyk's result~\cite{Banaszczyk98}, which we discuss in Banaszczyk's setting. Here, we restate Beck-Fiala's  result.

\begin{theorem}[Bekc-Fiala~\cite{Beck_Fiala}] \label{thm:Beck-Fiala}
 Given any $\bA\in [-1,1]^{m\times n}$ of degree $d$ and any starting point $\bx_0\in [-1,1]^{n}$, there exists an efficient algorithm to find $\bx\in \{\pm 1\}^{n}$ such that $\bx(i)=\bx_0(i)$ for each $\bx_0(i)\in \{\pm 1\}$ and $\|\bA(\bx-\bx_0)\|_{\infty}=O(d)$.
\end{theorem}

Our main result here is a sampling algorithm whose output has discpreancy $O(d)$ and min-entropy $\Omega(n)$. We present our random sampling algorithm in Algorithm~\ref{alg:sampling_Beck-Fiala}.

\begin{theorem}\label{thm:beck_fiala_full}
    For some constant $C=O(1)$,  there exists an efficient  sampling algorithm that satisfies the following two properties: for any input $\mathbf{A}\in [-1,1]^{m \times n}$ of degree $d$, 
    \begin{enumerate}
        \item Its output $\bx \in \{\pm 1\}^n$ always satisfies $\|\mathbf{A} \bx \|_{\infty} \le C \cdot d$. \label{pro:1_thm_bf_complex}
        \item For any $\boldsymbol{\epsilon} \in \{\pm 1\}^n$, $\Pr[\textsc{SamplingBeckFiala} =\boldsymbol{\epsilon}] = O(1.9^{-0.9n})$. %\hynote{July 10. Just $\le 1.9^{-0.9n}$.Besides, this presentation in $\Pr[\cdot] should be re-organised.$}
        \label{pro:2_thm_bf_complex}\end{enumerate} 
\end{theorem}

%\begin{remark}
%    When $d=O(1)$, there are set systems of degree $d$ such that the number of colorings with $\|\bA \bx\|_{\infty} \le C d$ is less than $(2-\beta)^n$ for some $\beta>0$. For example, consider a set system with $\frac{n}{10C \cdot d}$ disjoint subsets where each subset is of size $10C \cdot d$. 
    
%    Moreover, we could generalize this algorithm to output each coloring with probability at most $(2-\beta)^{-n}$ for any $\beta>0$ %\hynote{There should be an upper bound here.} 
%    by re-setting $T=(1-\beta/4)n, \delta = \beta/4, \gamma=\Theta(\beta^3)$, and $C=\Theta(1/\beta^3)$ in Algorithm~\ref{alg:sampling_Beck-Fiala}.  But for ease of exposition, we did not attempt to optimize those constants and show a sampling probability of at most $O(1.9^{-0.9n})$.
%\end{remark}

\begin{proof}

To control the probability of failing, we repeat 
Algorithm~\ref{alg:sampling_Beck-Fiala} $n$ times. If all the $n$ repeats fail and output $\bot$, we directly apply Theorem~\ref{thm:Beck-Fiala} to get a coloring $\bx\in \{\pm 1\}^{n}$.

The primary distinction between  Algorithm~\ref{alg:sampling_Beck-Fiala} and  Algorithm~\ref{alg:sample_spencer_big_n} is the definition of $H_t$.  We only need to ensure  that \begin{equation}
    \label{eq:banaszczyk_1}
    \frac{\sum_{j\in [n]\setminus \mathcal{D}_t}\tau_j}{n-|\mathcal{D}_{t}|} =\frac{dim(H_t)}{n-|\mathcal{D}_{t}|} = 1-\frac{1+|G_t|}{n-|\mathcal{D}_t|} \ge 1-  0.49\cdot \gamma.
\end{equation}      
By the definition of $G_t$, we have  $|G_t|\le \frac{d\cdot (n-t)}{C\cdot d/2}$. By the chosen parameters of $C$ and $\gamma$, we could verify Inequality \eqref{eq:banaszczyk_1} directly, like  Inequality~\eqref{ep:spencer_big_n_iteration_1}. The rest proof is the same as that of Lemma~\ref{lem:spencer_big_n} and we omit it.

    In fact, in every update $t$, we could pick $99\%$ entries in $[n]\setminus \mathcal{F}_t$ as the entry $k_t$.   To show this, we reset $C$ to be large enough and $\gamma$ to be small enough such that $1-10^{-3}$ fraction of coordinates has $\tau_{j}(H_t) \ge 1 - \gamma$. Lemma~\ref{lem:spencer_big_n_iteration} implies that the main term of $\sum_{i \notin \mathcal{F}_t} |\bx_t(i)|^2 \le \frac{\gamma}{1-\gamma} \cdot t + 8 \cdot \delta \cdot n^{0.7}$ is $\frac{\gamma}{1-\gamma} \cdot t$ independent with $\delta$. It means that $1-10^{-3}$ fraction of coordinates also satisfy $|\bx_t(i)|\le \delta$. A union bound shows $99\%$ coordinates in $[n]\setminus\mathcal{F}_t$ could be chosen as the entry $k_t$.
\end{proof}

\begin{algorithm}[h]
\caption{Sampling Procedure for the Beck-Fiala setting}\label{alg:sampling_Beck-Fiala}
\hspace*{\algorithmicindent} \textbf{Input}: $\mathbf{A}\in [-1,1]^{m \times n}$ of degree $d$, $\eta =0.1$, $T=0.9n$, $C=41000$, $\gamma=0.0001$, $\delta=0.05$. \\ 
\hspace*{\algorithmicindent}
\textbf{Random seed} $\mathbf{r} \in \{\pm 1\}^{T}$ will be generated on the fly.\\
\hspace*{\algorithmicindent}
\textbf{Output}: $\bx\in \{\pm 1\}^{n}$.

\begin{algorithmic}[1]
\Function{SamplingBeckFiala}{}
\State $\mathbf{x}_0=\vec{0}$.
\For{$0\le t\le T-1$}
 \State $\mathcal{F}_{t}\leftarrow \{i\in [n]:|\bx_{t}(i)|=1\}$.
 \State $\mathcal{D}_{t}\leftarrow \{i\in [n]:|\bx_{t}(i)|\ge 1-\eta\}$.
 \State Let $G_{t} \subseteq \mathbb{R}^{[n]\setminus \mathcal{D}_t}$ denote $\bigg\{\bA(i,[n] \setminus \mathcal{D}_t):  i\in [m] \text{ with } \big\| \bA(i,[n]\setminus \mathcal{F}_t) \big\|_0 > \frac{C \cdot d}{2}\bigg\}$ as the set of heavy constraints at moment $t$. %\gynote{the same comment in Algorithm \ref{alg:BF_simple}}
 \State Let $H_{t}$ be the subspace in $\mathbb{R}^{[n]\setminus \mathcal{D}_t}$ orthogonal to $\mathbf{x}_t([n]\setminus\mathcal{D}_t)$ and $G_t$.

 \State Compute the leverage score $\tau_i(H_t)$ of each $i \in [n] \setminus \mathcal{D}_t$.
 
 \State  Find the first coordinate $k_{t} \notin \mathcal{D}_t$  such that $\tau_{k_t}(H_t) \ge 1-\gamma$ and $|\bx_t(k_t)| \le \delta$ --- return $\bot$ if such a coordinate does not exist.  %\hynote{"quit” seems not defined, quit from all the process or just this procedure/round?}
 
 \State $\mathbf{u}_{t} \gets \textsc{FindVector}(H_t,k_{t})$ in Algorithm~\ref{alg:leverage_score_alg}.

 \State For $\delta_t := \bx_t(k_t)$, randomly set $\br(t+1)=1$ with probability $\frac{1+\delta_t}{2}$ and $\br(t+1)=-1$ with probability $\frac{1-\delta_t}{2}$.
 
 \State Choose $\alpha_{t} \in \mathbb{R}$ such that $\mathbf{x}_{t+1} \gets \mathbf{x}_{t} + \alpha_{t} \cdot \mathbf{u}_{t}$ satisfies $\mathbf{x}_{t+1}(k_{t})=\mathbf{r}(t+1)$. %\xcnote{Change $C_t$ to $\alpha_t$ and modify the corresponding proof}
\EndFor
%\State Get $\bx^{\star }\in \{\pm 1 \}^{n}$ by coloring each entry $\bx_T(i) \in (-1,1)$ to $\{\pm 1\}$ arbitrarily.
\State Get $\bx\in \{\pm 1 \}^{n}$ by applying Theorem~\ref{thm:Beck-Fiala} to color all entries of $\mathbf{x}_T$ not in $\{\pm 1\}$.
\State \Return $\bx$.
\EndFunction
\end{algorithmic}
\end{algorithm}

\paragraph{Banaszczyk's setting} In Banaszczyk's setting, we investigate the discrepancy of $\bA\in\mathcal{R}^{m\times n}$, each column of which has $\ell_2$-norm at most $1$. Rewrite the $n$ columns of $\bA$ as vectors $\bv_1,\bv_2,\ldots,\bv_n\in \mathcal{R}^{m}$ with $\|\bv_i\|_{2}\le 1$ for all the  $i\in [n]$. Our objective  is to find a coloring $\bx\in \{\pm 1\}^{n}$ to  minimize $\|\bA\bx\|_{\infty} = \|\sum_{i=1}^{n}\bx(i)\cdot\bv_i\|_{\infty}$. The best  bound  to date  is  
$$\underset{\bx\in\{\pm 1\}^{n}}{\min }  \Big\|\sum_{i=1}^{n}\bx(i)\cdot\bv_i\Big\|_{\infty}\le  O\left(\sqrt{\log \min\{n,m\}}\right),$$  established by Banaszczyk \cite{Banaszczyk98}. A celebrated line of research has provided efficient algorithms for this bound including \cite{BansalDG19,BDGL19,LevyRR17,BansalLV22}.

We  introduce a random sampling algorithm for Banaszczyk's setting in Algorithm~\ref{alg:sampling_Banaszczyk}.  It is based on the random walk with Gaussian stationary distribution from \cite{LiuSS22}.  Our technical contribution is several monotone properties about this random walk whose stationary distribution is Gaussian. For completeness, we provide its full proof in Appendix~\ref{sec:sampling_Banaszczyk}.

\begin{theorem}\label{thm:sampling_Banaszczyk}
    Given any $\gamma>0$, for any $n \ge \gamma \cdot m$, vectors $\bv_1,\ldots,\bv_n \in \R^m$ with $\ell_2$-norm at most $1$, there exist $e^{\Omega (n)}$ many colorings $\bx \in \{\pm 1\}^{n}$ such that $\|\sum_{i=1}^n \bx(i) \cdot \bv_i\|_{\infty} = O(\sqrt{\log m})$. Algorithm~\ref{alg:sampling_Banaszczyk} is an efficient algorithm and satisfies the following properties: 
    \begin{enumerate}
        \item Its output $\bx$ always has discrepancy $\|\sum_{i=1}^n \bx(i) \cdot \bv_i\|_{\infty} = O(\sqrt{\log m})$. %\gynote{with probability of 1-$(nm)^{-C}$}.
        \item The probability of outputting any string in $\{\pm 1\}^n$ is at most $\delta^n$ for some $\delta<1$.
    \end{enumerate}
\end{theorem}
\begin{remark}
While we did attempt to optimize those constants, there are a few remarks.
\begin{enumerate}
%     But for $n=m$, we could prove that there are $\ge 1.05^n$ colorings (by setting $\sigma=10$, $\alpha=0.8931$, $\beta=0.99974$, and $B=210$ for  Lemma~\ref{lem:induction_bana}). 
    \item There are sets of vectors whose number of colorings $\bx$ with $\|\sum_{i=1}^n \bx(i) \bv_i\|_{\infty} \le \sqrt{\log m}$ is at most $2^{-\Omega(\frac{n}{\log n})} \cdot 2^n$. For example, consider $\frac{n}{\log n}$  groups of same vectors like $\bv_1=\cdots=\bv_{\log n}=\mathbf{e}_1,\ldots,\bv_{\frac{n}{\log n} - \log n + 1}=\cdots=\bv_{n}=\mathbf{e}_{\frac{n}{\log n}}$.
    \item The argument in  \cite{bansal2022smoothed} provides another random sampling algorithm with a similar guarantee\footnote{Thanks the anonymous referee to suggest it.}. For completeness, we include it in Appendix~\ref{appen:subgaussian_sampling_Bana}.
    \item  Theorem~\ref{thm:sampling_Banaszczyk}    provides an alternative random sampling algorithm with discrepancy  $O(\sqrt{d \log m})$ for the Beck-Fiala setting, if we scale the degree $d$  matrix $\bA$ by a factor $\frac{1}{\sqrt{d}}$.
\end{enumerate}
\end{remark}

\begin{algorithm}[H]
\caption{Randomly Sample $\bx \in \{\pm 1\}^{m}$ for Banaszczyk's bound}\label{alg:sampling_Banaszczyk}
\hspace*{\algorithmicindent} \textbf{Input}: $\gamma$, $\left\{\bv_{i}\right\}_{i=1}^{n}\subset \R^{m}$ where $0<\|\bv_{i}\|_{2}\le 1~\forall i\in[n]$. \\
\hspace*{\algorithmicindent} \textbf{Output}: $\bx$ in $\left\{\pm 1 \right\}^{n}$ and its discrepancy vector $\boldsymbol{\omega}$.
\begin{algorithmic}[1]
\Procedure{SamplingBanaszczyk}{}
\State Resort $\bv_1,\ldots,\bv_n$ such that $\|\bv_1\|_2 \ge \|\bv_2\|_2 \ge \cdots \ge \|\bv_n\|_2$.
\State Set $\sigma \ge 1$ such that $r_{\sigma}(\ell)<\frac{1}{6}e^{-\frac{24}{\gamma}}$ for all $\ell \in [-0.5,0.5]$ --- see Section~\ref{sec:stationary_RW} for the definition of $r_{\sigma}(\ell)$ and its properties.
\State Sample $\boldsymbol{\omega}_0 \sim N(0,\sigma^{2} \cdot I_{m})$.
\For{$1\le i\le n$}
 \State $\sigma^{\prime}$ $\gets$ $\sigma/\|\bv_{i}\|_{2}$.
 \State $\ell_i \gets$ $\langle \boldsymbol{\omega}_{i-1},\bv_{i}\rangle/\|\bv_{i}\|_{2}$ .%\xcnote{$x'$ is a bad notation since our output is $\bx$, any suggestion?} \hynote{Maybe we could replace $x^{\prime}$ with $\ell_{i}$. And this notation should be consistent with the following argument.}
 \State  $\by(i) \gets$ $J^{\sigma^{\prime}}_{\ell_i}$ (see Section~\ref{sec:stationary_RW} for the definition). 
 \State $\boldsymbol{\omega}_{i}$ $\gets$ $\boldsymbol{\omega}_{i-1}+\by(i) \cdot \bv_{i}$.
\EndFor
%\State $\boldsymbol{\omega}^{\prime}$ $\gets$ $\boldsymbol{\omega}_{n}-\boldsymbol{\omega}_{0}$
\State Let $E=\{i|\by(i)=0\}$ denote those empty entries in $\by$.
\State Apply Theorem~\ref{subgaussian_Gsw_HSSZ} to obtain $\mathbf{z} \in \{\pm 1\}^E$ for vectors $\{\bv_i\}_{i \in E}$.
\State Output $\bx=\by\left(\overline{E}\right) \circ \bz(E)$ and $\boldsymbol{\omega}=\sum_{i=1}^{n} \bx({i})\cdot \bv_{i}$.
\EndProcedure
\end{algorithmic}
\end{algorithm}

\section*{Acknowledgement} We thank the 
anonymous referee for the helpful comments on the previous version.

\bibliographystyle{alpha} 
\bibliography{discrepancy}

\appendix

\section{Proof of Lemma~\ref{lem:gaussian_lower_bound}}\label{sec:pf_gaussian_lower_bound}
% We have the following standard bound for gaussian random variable.
% \begin{claim}\label{cla:Gaussian_tail}
%    For the Gaussian random variable $N(0,\sigma^2)$ and deviation $\delta \cdot \sigma$, we have 
% \begin{equation}   \Pr[|N(0,\sigma^2)| \le \delta \cdot \sigma] \ge \left\{ \begin{aligned} & \delta/5 & \textit{ if } \delta \le 1,\\ & 1 - \frac{e^{-\delta^2/2}}{\sqrt{2 \pi} \cdot \delta} & \textit{ if } \delta>1. \end{aligned} \right.      \end{equation}
% \end{claim}
Actually, the proof still holds for the symmetric sub-Gaussian random variables (i.e. Bernoulli random variables supported on $\{\pm 1\}$). But for simplicity, we consider the standard Gaussian random variable here.

For the case that $n<\ell$, we apply Lemma~\ref{lem:levy ineq} directly.
\begin{align*}
\Pr[S^*_n\geq 4\sqrt{\ell}]&\leq 2\Pr[S_n\geq 4\sqrt{\ell}]\\
&\leq 2\frac{e^{-8\frac{\ell}{n}}}
{\sqrt{2\pi}\cdot4\cdot\sqrt{\frac{\ell}{n}}}\\
&\leq e^{-8\frac{\ell}{n}}.
\end{align*}
Then
\begin{align*}
   \Pr[S^*_n<4\sqrt{\ell}]&=1-\Pr[S^*_n\geq 4\sqrt{\ell}]\\
   &\geq 1-e^{-8\frac{\ell}{n}}\\
   &\geq e^{-2\ln 2\cdot \frac{n}{\ell}}\tag{$n<\ell$}.
\end{align*}

For $n\geq \ell$,
without loss of generality, assume $n$ divide $l$.
Then we divide $X_1,...,X_n$ into $\frac{n}{l}$ intervals. Interval $i$ contains $l$ random variables $X_{(i-1)\cdot l+1},...,X_{i\cdot l}.$ 
Let $Z_i^*=\underset{1\leq j\leq l}{\max}|\sum_{k=1}^jX_{(i-1)\cdot l+k}|$ denote the largest deviation of the prefix sum in interval $i$. Let $Z_i=\sum_{j=1}^lX_{(i-1)\cdot l+j}$.
Thus $Z_i$ for $i=1,2,...,\frac{n}{l}$ are i.i.d. Gaussian random variables with distribution $N(0,l)$.
  Recall that $S^*_l=\underset{1\le i\le l
 }{\max}|\sum_{j=1}^iX_i|$ and $S_l=\sum_{i=1}^lX_i.$ Thus $Z_1^*=S^*_l$.

For the $i$-interval, we have
\begin{align*}
    \Pr[Z_i^*\geq 2\sqrt{\ell}]&\leq 2\Pr[|Z_i|\geq 2\sqrt{\ell}] \tag{By Lemma~\ref{lem:levy ineq}}\\
    &\leq \frac{e^{-2}}{\sqrt{2\pi}}.\tag{By Claim~\ref{clm:Gaussian_tail}}
\end{align*}
This means $\Pr[Z_i^*< 2\sqrt{\ell}]\geq (1- \frac{e^{-2}}{2\sqrt{2\pi}})>\frac{1}{2}$. In the conditional distribution of $Z_i^* < 2 \sqrt{\ell}$, one useful property is that $Z_i \in (-2\sqrt{\ell},2\sqrt{\ell})$ is symmetric such that 
\begin{equation}\label{eq:sym_z}
    \Pr\left[0\leq Z_i<2\sqrt{\ell}\big|Z_i^*<2\sqrt{\ell}\right]=\Pr\left[-2\sqrt{\ell}< Z_i\leq 0\big|Z_i^*<2\sqrt{\ell}\right]\ge\frac{1}{2}.
\end{equation}
Similarly, in the conditional distribution of $|S_{(i-1)\cdot\ell}<2\sqrt{\ell}|$, $S_{(i-1)\cdot\ell}\in (-2\sqrt{\ell},2\sqrt{\ell})$ is symmetric such that
\begin{equation}\label{eq:sym_s}
    \Pr\left[0\leq S_{(i-1)\cdot\ell}<2\sqrt{\ell}\big||S_{(i-1)\cdot\ell}|<2\sqrt{\ell}\right]=\Pr\left[-2\sqrt{\ell}< S_{(i-1)\cdot\ell}\leq 0\big||S_{(i-1)\cdot\ell}|<2\sqrt{\ell}\right]\ge\frac{1}{2}.
\end{equation}
Moreover, $Z_i$ and $Z_i^*$ are independent with $S_{(i-1)\cdot\ell}$.

Then we consider
\begin{align*}
    &\Pr\left[|S_{i\cdot \ell}|<2\sqrt{\ell}\bigg||S_{(i-1)\cdot \ell}|<2\sqrt{\ell},Z_i^*<2\sqrt{\ell}\right]\\
    &=\Pr\left[|S_{(i-1)\cdot \ell}+Z_i|<2\sqrt{\ell}\bigg||S_{(i-1)\cdot \ell}|<2\sqrt{\ell},,Z_i^*<2\sqrt{\ell}\right]\\
    &=\Pr\left[-2\sqrt{\ell}<S_{(i-1)\cdot l}+Z_i<2\sqrt{\ell}\bigg||S_{(i-1)\cdot \ell}|<2\sqrt{\ell},,Z_i^*<2\sqrt{\ell}\right]\\
    &=\Pr\left[S_{(i-1)\cdot l}+Z_i\in(-2\sqrt{\ell},2\sqrt{\ell}),S_{(i-1)\cdot \ell}\in(-2\sqrt{\ell},0]\bigg||S_{(i-1)\cdot \ell}|<2\sqrt{\ell},,Z_i^*<2\sqrt{\ell}\right]\\
    &+\Pr\left[S_{(i-1)\cdot l}+Z_i\in(-2\sqrt{\ell},2\sqrt{\ell}),S_{(i-1)\cdot \ell}\in(0,2\sqrt{\ell})\bigg||S_{(i-1)\cdot \ell}|<2\sqrt{\ell},,Z_i^*<2\sqrt{\ell}\right].
    % &\geq \Pr\left[0\le Z_i<2\sqrt{\ell},,S_{(i-1)\cdot \ell}<0\bigg||S_{(i-1)\cdot \ell}|<2\sqrt{\ell},Z_i^*<2\sqrt{\ell}\right]\\
    % &+\Pr\left[-2\sqrt{\ell}< Z_i\le 0,,S_{(i-1)\cdot \ell}> 0\bigg||S_{(i-1)\cdot \ell}|<2\sqrt{\ell},Z_i^*<2\sqrt{\ell}\right]\\
    % &\geq \frac{1}{2}\cdot\frac{1}{2}+\frac{1}{2}\cdot\frac{1}{2}=\frac{1}{2}\tag{Use the property that $Z_i$ and $S_{(i-1)\cdot \ell}$ are symetric and $Z_i,Z_i^*$ are independent with $S_{(i-1)\cdot \ell}$.}
\end{align*}
In the case $-2\sqrt{\ell}<S_{(i-1)\cdot\ell}\leq 0$, $[0,2\sqrt{\ell})\subset (-2\sqrt{\ell}-S_{(i-1)\cdot\ell},2\sqrt{\ell}+S_{(i-1)\cdot\ell})$ .
Thus 
\begin{align*}
    &\Pr\left[S_{(i-1)\cdot l}+Z_i\in(-2\sqrt{\ell},2\sqrt{\ell}),S_{(i-1)\cdot \ell}\in(-2\sqrt{\ell},0]\bigg||S_{(i-1)\cdot \ell}|<2\sqrt{\ell},,Z_i^*<2\sqrt{\ell}\right]\\
&\ge\Pr\left[Z_i\in[0,2\sqrt{\ell}),S_{(i-1)\cdot \ell}\in(-2\sqrt{\ell},0]\bigg||S_{(i-1)\cdot \ell}|<2\sqrt{\ell},,Z_i^*<2\sqrt{\ell}\right]\\
    &=\Pr\left[Z_i\in[0,2\sqrt{\ell})\big|Z_i^*<2\sqrt{\ell}\right]\cdot \Pr\left[-2\sqrt{\ell}\leq S_{(i-1)\cdot\ell}\leq 0\big||S_{(i-1)\cdot\ell}|<2\sqrt{\ell}\right]\tag{$Z_i$ and $Z_i^*$ are independent with $S_{(i-1)\cdot\ell}$.}\\
    &\geq\frac{1}{2}\cdot \frac{1}{2}=\frac{1}{4}.\tag{By Equations\eqref{eq:sym_z} and \eqref{eq:sym_s}}
\end{align*}
Similarly, we have
$$\Pr\left[S_{(i-1)\cdot l}+Z_i\in(-2\sqrt{\ell},2\sqrt{\ell}),S_{(i-1)\cdot \ell}\in(0,2\sqrt{\ell})\bigg||S_{(i-1)\cdot \ell}|<2\sqrt{\ell},,Z_i^*<2\sqrt{\ell}\right]\ge \frac{1}{4}.$$

Thus
\begin{equation}\label{eq:EF}
    \Pr\left[|S_{i\cdot \ell}|<2\sqrt{\ell}\bigg||S_{(i-1)\cdot \ell}|<2\sqrt{\ell},Z_i^*<2\sqrt{\ell}\right]\geq \frac{1}{4}+\frac{1}{4}=\frac{1}{2}.
\end{equation}

Denote the event $E_i=\{Z_i^*<2\sqrt{\ell}\}$, $F_i=\{|S_{i\cdot l}|<2\sqrt{\ell}\}$. Since $Z_1\leq Z_1^*=S_{\ell}^*$, we have $E_1\subset F_1$.
Besides, $E_i$ is independent with $F_j$ for $j<i$ and is independent with $E_k$ for $k\neq i$. Notice that $S_{i\cdot \ell}=S_{(i-1)\cdot\ell}+Z_i$, which means $S_{i\cdot \ell}$ only depends on $S_{(i-1)\cdot\ell},Z_i$.
Thus
\begin{align*}
    \Pr[F_1,F_2,...,F_{\frac{n}{\ell}},E_1,...,E_{\frac{n}{\ell}}]&=(\prod_{i=2}^{\frac{n}{\ell}}\Pr[F_i|F_{i-1},E_i]\cdot\Pr[E_i])\cdot \Pr[F_1,E_1]\\
    &\ge (\frac{1}{4})^{\frac{n}{\ell}-1}\Pr[E_1]\tag{Use Equation\eqref{eq:EF} and $E_1\subset F_1$}\\
    &\ge (\frac{1}{4})^{\frac{n}{\ell}}.
\end{align*}
 
If events $\{E_i\}_{i\in [\frac{n}{\ell}]}$ and $\{F_i\}_{i\in [\frac{n}{\ell}]}$ all happen,
for any $1\le t\le n$, we have
\begin{align*}
    |S_t|&=|S_{\lfloor\frac{t}{\ell}\rfloor\cdot \ell}+\sum_{i=1}^{t-\lfloor\frac{t}{\ell}\rfloor\cdot \ell}X_{\lfloor\frac{t}{\ell}\rfloor\cdot \ell+i}|\\
    &< |S_{\lfloor\frac{t}{\ell}\rfloor\cdot \ell}|+Z_{\lfloor\frac{t}{\ell}\rfloor\cdot \ell}^*\\
    &< 2\sqrt{\ell}+2\sqrt{\ell}=4\sqrt{\ell}.
\end{align*}
Thus $\Pr[S_n^*<4\sqrt{\ell}]\geq \Pr[F_1,F_2,...,F_{\frac{n}{\ell}},E_1,...,E_{\frac{n}{\ell}}]\geq (\frac{1}{4})^{\frac{n}{\ell}}=e^{-2\ln 2\cdot\frac{n}{\ell}}$.
% \begin{align*}
%     \Pr[S_n^*<8\sqrt{m}]&=\Pr[\sum_{i=1}^{\frac{n}{m}}(S_{i\cdot m}^*-S_{(i-1)\cdot m}^*)<8\sqrt{m}]\\
%     &\geq \Pr[\sum_{i=1}^{\frac{n}{m}}Z_i^*<8\sqrt{m}]\\
%     &=\left(\prod_{j=2}^{\frac{n}{m}}\Pr[\sum_{i=1}^{j}Z^*_i<8\sqrt{m}|\sum_{i=1}^{j-1}Z^*_i<8\sqrt{m}]\right)\cdot \Pr[Z^*_1<8\sqrt{m}]
% \end{align*}
% \begin{align*}
%     \Pr[S_n^*<4\sqrt{m}]&=
%     \left(\prod_{j=2}^{\frac{n}{m}}\Pr\left[S^*_{i\cdot m}<4\sqrt{m}|S^*_{(i-1)\cdot m}<4\sqrt{m}\right]\right)\cdot \Pr[S^*_m<4\sqrt{m}]\\
%     &\geq \left(\frac{1}{4}\right)^{\frac{n}{m}-1}\cdot\left(\frac{1}{4}\right)\\
%     &\geq e^{-2\ln 2\cdot \frac{n}{m}}
% \end{align*}

\section{Counting Argument from \cite{kim2005discrepancy}}\label{sec:counting}
We show that there are exponentially many solutions in the Beck-Fiala setting and the Banaszczyk setting based on a Shannon-entropy argument of \cite{kim2005discrepancy}.

\begin{theorem} 
Given any matrix $\bA \in \mathbb{R}^{m \times n}$, if its any submatrix $\bA'$ constituted by a subset of columns in $\bA$ has discrepancy at most $h$, then there are at least $2^{0.5 n}$ colorings $\bx \in \{\pm 1\}^n$ satisfying $\|\bA \bx\|_{\infty} \le 2h$.
%If an algorithm $\mathcal{A}$ outputs coloring $\mathbf{x}\in\{\pm 1\}^{n}$ with discrepancy at most $h$ (in different settings $h$ is different, such as Spencer's, Beck-Fiala's and Banaszczyk's settings), there is an algorithm based on $\mathcal{A}$ which outputs  coloring $\mathbf{x}\in \{-1,1\}^{n}$ with discrepancy at most $2h$ and the distribution of $\mathbf{x}$ has Shannon entropy at least  $0.5n$.   
\end{theorem}
In particular, this implies that there are exponentially many solutions in the Spencer's setting, the Beck-Fiala setting, and the Banszczyk's setting.

\begin{proof}
 Consider a random set $E \subset [n]$ and split the matrix $\bA$ into two submatrices by partition its columns into $\bA_E$ and $\bA_{[n]\setminus E}$. Let $\mathbf{y}_E$ be the coloring of $E$ satisfying $\|\bA_E\cdot  \mathbf{y}_E\|_{\infty} \le h$ and $\mathbf{y}_{[n]\setminus E}$ be the counterpart of $\bA_{[n]\setminus E}$ with discrepancy at most $h$. Now we define two colorings based on the two sub-colorings $\mathbf{y}_E$ and $\mathbf{y}_{[n]\setminus E}$ as 
$$ \mathbf{x^{\prime}}_{E}(i)=\left\{
\begin{array}{rcl}
\mathbf{y}_{E}(i)      &      & {i\in E}\\
\mathbf{y}_{[n]\setminus E}(i)     &      & {i\in [n]\setminus E}
\end{array} \right. \text{ and }
 \mathbf{x^{\prime\prime}}_{E}(i)=\left\{
\begin{array}{rcl}
-\mathbf{y}_{E}(i)      &      & {i\in E}\\
\mathbf{y}_{[n]\setminus E}(i)     &      & {i\in [n]\setminus E}
\end{array} \right.$$
Thus both $\mathbf{x^{\prime}}_{E}$ and $\mathbf{x^{\prime \prime}}_{E}$ have discrepancy at most $2h$. Given a subset $E\in [n]$, we call $(\mathbf{x^{\prime}}_{E},\mathbf{x^{\prime\prime}}_{E})$ a pair of colors encoded by $E$. 

We first prove all the pairs are different. 
For any two different subsets $E\in 2^{[n]}$ and $Q\in 2^{[n]}$, we have $E\bigcap ([n]\setminus Q) \neq \emptyset$. Therefore  
 $\exists i\in E\bigcap ([n]\setminus Q)$ so that   
 $\mathbf{x^{\prime}}_{Q}(i)=\mathbf{x^{''}}_{Q}(i)$ but $\mathbf{x^{\prime}}_{E}(i) \neq \mathbf{x^{\prime \prime }}_{E}(i)$. Because there are $2^{n}$ different choices of $E$, we get $2^{n}$ different pairs. 

 Since these paris are distinct, the Shannon entropy of the pair $(\bx'_E,\bx''_E)$ for a uniform random subset $E$ in $[n]$ is $n$. By the sub-additivity of Shannon entropy either the entropy of $\bx'_E$ or the entropy of $\bx''_E$ is at least $0.5n$. If $H(\bx'_E) \ge 0.5n$, this implies that there are $2^{0.5n}$ colorings with discrepancy $2h$. %\drnote{"$H(x'_E) \ge 0.5n$" and "there are $2^{0.5n}$ colorings"?}
 
 \end{proof}

\begin{remark}
This argument also provides a random sampling algorithm whose output has \emph{the Shannon entropy} at least $0.5n$ by outputting one of $\{\bx'_E,\bx''_E\}$ with probability half for a uniformly random $E \subset [n]$. 

However, our algorithm has a \emph{stronger guarantee} of a min-entropy $0.9n$ such that each coloring will appear with an exponential small probability. On the other hand, for the Shannon entropy, it is possible that this sampling algorithm will output a fixed coloring with probability $0.9$; but its output still has the Shannon-entropy $\ge 0.5n$.
%The conclusion that there are  $2^{0.5n}$ different colorings with discrepancy $2h$ is direct from the entropy at least $0.5n$. 
\end{remark}

\section{Random Sampling in Banaszczyk's setting via sub-Gaussian argument} \label{appen:subgaussian_sampling_Bana}
%We consider $\bA \in \mathbb{R}^{m \times n}$ where each column $\mathbf{v}_j$ has $\|\mathbf{v}_j\|_2 \le 1$ and 
%Recall the algorithmic result of the Banaszczyk's bound \cite{Banaszczyk98} when each column has $\|\bv_j\|_2 \le 1$.

In this section, we show a random sampling algorithm for Banaszczyk's bound based on the discrepancy algorithm of \cite{BDGL19} and the idea of matrix stacking from \cite{bansal2022smoothed}. To simplify the presentation, we always use $\mathbf{v}_j$ for $j \in [n]$ to denote column $j$ in $\mathbf{A}$ such that the goal is equivalent to $\min_{\mathbf{x} \in \{\pm 1\}^n} \|\sum_{j=1}^n \mathbf{x}(j) \cdot \mathbf{v}_j \|_{\infty}$ for the Koml\'{o}s problem. 

For convenience, we call a random vector $\mathbf{x}\in \mathbb{R}^{n}$ $\lambda$-subgaussian, if it satisfies  $\E[\text{exp}(\langle \mathbf{u}, \mathbf{x}  \rangle)]\le \text{exp}\left(\frac{\lambda^{2}\cdot\|\mathbf{u}\|_{2}^{2} }{2}\right)$ for every  $\mathbf{u}\in \mathbb{R}^{n}$. One useful property of $\lambda$-subgaussian vectors is that for any \emph{unit vector} $\mathbf{v}$ and any $t>0$, 
\begin{equation}
    \Pr[\langle \mathbf{v},\mathbf{x}\rangle \ge t] \le 2 \cdot e^{-\frac{t^2}{2 \lambda^2}}.
\end{equation}
This is obtained by setting $\mathbf{u}=\frac{t}{\lambda^2} \cdot \mathbf{v}$ and apply Markov's inequality to $\Pr[\langle \mathbf{u}, \mathbf{x}  \rangle \ge t^2/\lambda^2] \le \E[\text{exp}(\langle \mathbf{u}, \mathbf{x}  \rangle)]/e^{t^2/\lambda^2}$ and vice versa for $\mathbf{u}=-\frac{t}{\lambda^2} \cdot \mathbf{v}$.

We restate the guarantee of the Gram-Schmidt walk algorithm \cite{BDGL19}. We remark that the original subgaussian constant is $\sqrt{40}$ in \cite{BDGL19} but a tighter constant $1$ is obtained in Theorem 6.6 of \cite{HSSZ_experiments19}.

\begin{theorem}\label{subgaussian_Gsw_HSSZ} 
Given a matrix $\mathbf{A}\in \mathbb{R}^{m\times n}$ whose columns have $\ell_2$ norm at most $1$, the Gram-Schmidt walk algorithm outputs a random coloring $\mathbf{x}\in \{\pm 1\}^{n}$ such that $\mathbf{A}\cdot\mathbf{x}$ is 1-subgaussian.

In particular, its output $\mathbf{x} \in \{\pm 1\}^n$ satisfies $\|\bA \bx\|_{\infty}=O(\sqrt{\log \min\{m,n\}})$ with high probability.
\end{theorem}

The sampling algorithm applies the matrix stacking technique \cite{bansal2022smoothed} as follows: Let $\bA'=\tbinom{\mathbf{A}}{\mathbf{I}_{n}}$ be the stack matrix of the input $\bA$ with the identity matrix $I_n$. Then let the output be the output of the Gram-Schmidt walk algorithm with input $\bA'$.

Now we prove its correctness.
\begin{theorem}
    Let $\mathbf{x}\in \{\pm 1\}^{n}$ be  the output of the Gram-Schmidt walk algorithm on the stacked matrix $\tbinom{\mathbf{A}}{\mathbf{I}_{n}}$. For any vector $\mathbf{\boldsymbol{\epsilon}}\in \{\pm 1\}^{n}$, we have \[\Pr_{\mathbf{x}}[\boldsymbol{\epsilon}=\mathbf{x}]\le 2 \cdot e^{-n/8}.\] 
    and with probability 
$1-m^{-\Omega(1)}$,  \[      \|\mathbf{A}\cdot \mathbf{x}\|_{\infty}\le O\left(\sqrt{\log m}\right).\]
\end{theorem}

\begin{proof}
    Because the column of the stacked matrix has $\ell_{2} $ norm at most $\sqrt{2}$, by Lemma \ref{subgaussian_Gsw_HSSZ} the discrepancy vector $\tbinom{\mathbf{A}}{\mathbf{I}_{n}}\cdot \mathbf{x}= \tbinom{\mathbf{A}\cdot \mathbf{x}}{\mathbf{x}}$ is $\sqrt{2}$-subgaussian. By the definition of sub-Gaussian, the two subvectors $\mathbf{x}$ and $\mathbf{A}\cdot\mathbf{x}$ are $\sqrt{2}$-subgaussian. So for any unit vector $\mathbf{u}\in \mathcal{R}^{n}$, \[
    \Pr_{\mathbf{x}}[|\langle \mathbf{x}, \mathbf{u} \rangle|\ge t]\le 2\cdot e^{-t^{2}/8} ,  
    \] and for any unit vector  $\mathbf{v}\in \mathcal{R}^{m}$ , \[\Pr_{\mathbf{x}}[|\langle \mathbf{A}\cdot  \mathbf{x}, \mathbf{v} \rangle|\ge t]\le 2\cdot e^{-t^{2}/8} . \]
    
    Then,  $\Pr_{\mathbf{x}}[\boldsymbol{\epsilon}=\mathbf{x}]\le \Pr_{\mathbf{x}}[|\langle   \mathbf{x}, \frac{\boldsymbol{\epsilon} }{\|\boldsymbol{\epsilon}\|_{2} }\rangle|\ge \sqrt{n}]\le 2\cdot e^{-n/8}.$  The second conclusion holds by a union bound over the canonical basis with $t=O(\sqrt{\log m})$.
\end{proof}

%\begin{lemma}[Part of Lemma 2.1 of \cite{bansal2022smoothed}]\label{modified_sub-gaussian}
%   Given $\mathbf{A}\in \mathbb{R}^{m\times n}$ whose columns have $\ell_2$ norm at most $1$, there exists a random sampling algorithm, whose output  $\mathbf{x}\in \{\pm 1\}^{n}$ and $\mathbf{A}\cdot \mathbf{x}\in \mathbb{R}^{m}$are $\sqrt{2}$-subgaussian, namely for any unit vector $\mathbf{u}\in \mathcal{R}^{n}$, \[    \Pr_{\mathbf{x}}[|\langle \mathbf{x}, \mathbf{u} \rangle|\ge t]\le e^{-t^{2}/8} ,  \] and for any unit vector  $\mathbf{v}\in \mathcal{R}^{m}$ , \[\Pr_{\mathbf{x}}[|\langle \mathbf{A}\cdot  \mathbf{x}, \mathbf{v} \rangle|\ge t]\le e^{-t^{2}/8} . \]\end{lemma}
%The following result about Gram-Schmidt walk is critical for the random sampling algorithm.

%\begin{proofof}{Theorem \ref{modified_sub-gaussian}}
  
%   Simply run Gram-Schmidt walk of \cite{BDGL19} on the stacked matrix $\tbinom{\mathbf{A}}{\mathbf{I}_{n}}$ and let  $\mathbf{x}\in \{\pm 1\}^{n}$ be the output. \end{proofof}

\section{Random Sampling in Banaszczyk's setting}\label{sec:sampling_Banaszczyk}

We prove Theorem~\ref{thm:sampling_Banaszczyk} in this section. It extends the fixed point Gaussian random walk algorithm in \cite{LiuSS22} 
to show (1) the existence of exponentially many vectors satisfying Banaszczyk's bound and (2) an efficient algorithm to sample a random good vector when $n=\Omega(m)$. For ease of reading, we restate the theorem and algorithm in Theorem~\ref{thm:random_Bana} and Algorithm~\ref{alg:Bana}.
%\subsection{Randomnthm:random_Banaess analysis of a modified coloring algorithm from \texorpdfstring{\cite{LiuSS22}}{}}

\begin{theorem}\label{thm:random_Bana}
    Given any $\gamma>0$, for any $n \ge \gamma \cdot m$ vectors $\bv_1,\ldots,\bv_n \in \R^m$ with $\ell_2$-norm at most $1$, there exist $e^{\Omega (n)}$ many colorings $\bx \in \{\pm 1\}^{n}$ such that $\|\sum_{i=1}^n \bx(i) \cdot \bv_i\|_{\infty} = O(\sqrt{\log m})$. Moreover, there exists an efficient algorithm that guarantees:
    \begin{enumerate}
        \item Its output $\bx$ always has discrepancy $\|\sum_{i=1}^n \bx(i) \cdot \bv_i\|_{\infty} = O(\sqrt{\log m})$. %\gynote{with probability of 1-$(nm)^{-C}$}.
        \item The probability of outputting any string in $\{\pm 1\}^n$ is at most $\delta^n$ for some $\delta<1$.
    \end{enumerate}
\end{theorem}

% \begin{remark}
%     We did not attempt to optmize those parameters. But for $n=m$, we could prove that there are $\ge 1.05^n$ colorings (by setting $\sigma=10$, $\alpha=0.8931$, $\beta=0.99974$, and $B=210$ for  Lemma~\ref{lem:induction_bana}). 

%     On the other hand, there are sets of vectors whose number of colorings $\bx$ with $\|\sum_{i=1}^n \bx(i) \bv_i\|_{\infty} \le \sqrt{\log m}$ is at most $2^{-\Omega(\frac{n}{\log n})} \cdot 2^n$. For example, consider $\frac{n}{\log n}$  groups of same vectors like $\bv_1=\cdots=\bv_{\log n}=\mathbf{e}_1,\ldots,\bv_{\frac{n}{\log n} - \log n + 1}=\cdots=\bv_{n}=\mathbf{e}_{\frac{n}{\log n}}$.
% \end{remark}
Algorithm~\ref{alg:Bana} uses $\boldsymbol{\omega_{i}}$ to denote the discrepancy vector of the first $i$ vectors, i.e., $\boldsymbol{\omega}_i=\sum_{j=1}^{i}\bx(j)\bv_{j}$. The key ingredient of Algorithm~\ref{alg:Bana} is the random walk $J^{\sigma}_{\ell}$ over $\mathbb{R}$ introduced in \cite{LiuSS22}, whose steps are $\{\pm 1\}$ mostly and stationary distribution is $N(0,\sigma^2)$. % (where $\sigma$ is defined on line 3 of Algorithm~\ref{alg:Bana}) 
We give a full description in Section~\ref{sec:stationary_RW}. Our main technical contribution here is to prove several monotonicity of this random walk on its parameters $\sigma$ and $\ell$ --- see Lemma~\ref{lem:mono_on_prob} in Section~\ref{sec:stationary_RW}.

%Let    $x_{i}$ denote the location of the $i$'th step in the $\left\{0,\pm 1\right\}$ walk of  \cite{LiuSS22} described in Section \ref{sec:stationary_RW}, which equals     $\frac{ \langle \boldsymbol{\omega}_{i},v_{i+1} \rangle }{\|v_{i+1}\|_{2}^{2}}$.

\begin{algorithm}[H]
\caption{Randomly Sample $\bx \in \{\pm 1\}^{m}$ for Banaszczyk's bound}\label{alg:Bana}
\hspace*{\algorithmicindent} \textbf{Input}: $\gamma$, $\left\{\bv_{i}\right\}_{i=1}^{n}\subset \R^{m}$ where $0<\|\bv_{i}\|_{2}\le 1~\forall i\in[n]$. \\
\hspace*{\algorithmicindent} \textbf{Output}: $\bx$ in $\left\{\pm 1 \right\}^{n}$ and its discrepancy vector $\boldsymbol{\omega}$.
\begin{algorithmic}[1]
\Procedure{SamplingBanaszczyk}{}
\State Resort $\bv_1,\ldots,\bv_n$ such that $\|\bv_1\|_2 \ge \|\bv_2\|_2 \ge \cdots \ge \|\bv_n\|_2$.
\State Set $\sigma \ge 1$ such that $r_{\sigma}(\ell)<\frac{1}{6}e^{-\frac{24}{\gamma}}$ for all $\ell \in [-0.5,0.5]$ --- see Section~\ref{sec:stationary_RW} for the definition of $r_{\sigma}(\ell)$ and its properties.
\State Sample $\boldsymbol{\omega}_0 \sim N(0,\sigma^{2} \cdot I_{m})$.
\For{$1\le i\le n$}
 \State $\sigma^{\prime}$ $\gets$ $\sigma/\|\bv_{i}\|_{2}$.
 \State $\ell_i \gets$ $\langle \boldsymbol{\omega}_{i-1},\bv_{i}\rangle/\|\bv_{i}\|_{2}$ .%\xcnote{$x'$ is a bad notation since our output is $\bx$, any suggestion?} \hynote{Maybe we could replace $x^{\prime}$ with $\ell_{i}$. And this notation should be consistent with the following argument.}
 \State  $\by(i) \gets$ $J^{\sigma^{\prime}}_{\ell_i}$ (see Section~\ref{sec:stationary_RW} for the definition). 
 \State $\boldsymbol{\omega}_{i}$ $\gets$ $\boldsymbol{\omega}_{i-1}+\by(i) \cdot \bv_{i}$.
\EndFor
%\State $\boldsymbol{\omega}^{\prime}$ $\gets$ $\boldsymbol{\omega}_{n}-\boldsymbol{\omega}_{0}$
\State Let $E=\{i|\by(i)=0\}$ denote those empty entries in $\by$.
\State Apply Theorem~\ref{subgaussian_Gsw_HSSZ} to obtain $\mathbf{z} \in \{\pm 1\}^E$ for vectors $\{\bv_i\}_{i \in E}$.
\State Output $\bx=\by\left(\overline{E}\right) \circ \bz(E)$ and $\boldsymbol{\omega}=\sum_{i=1}^{n} \bx({i})\cdot \bv_{i}$.
\EndProcedure
\end{algorithmic}
\end{algorithm}

In the rest of this section, we analyze Algorithm~\ref{alg:Bana} and finish the proof of Theorem~\ref{thm:random_Bana}. Let $\boldsymbol{\omega}^{\prime}:=\boldsymbol{\omega}_{n}-\boldsymbol{\omega}_{0}$ denote the discrepancy vector generated by $\by$ and $\boldsymbol{\omega''}:=\sum_{i\in E} \bz({i})\cdot \bv_{i}$ such that the discrepancy $\|\boldsymbol{\omega}\|_{\infty} \le \|\boldsymbol{\omega}'\|_{\infty}+\|\boldsymbol{\omega}''\|_{\infty}$. Then we only need to bound $\|\boldsymbol{\omega'}\|_{\infty}$ and $\|\boldsymbol{\omega''}\|_{\infty}$ separately. We state the main result in \cite{LiuSS22} to bound $\|\boldsymbol{\omega'}\|_{\infty}$.
\begin{theorem}[\cite{LiuSS22}]\label{thm:gaussian_stationary}
    In Algorithm~\ref{alg:Bana}, $\boldsymbol{\omega_i} \sim N(0,\sigma^2 \cdot I_m)$ for all $i \in [n]$.
\end{theorem}
Theorem~\ref{thm:gaussian_stationary} implies that $\boldsymbol{\omega}'=\sum_{i} \by(i) \cdot \bv_i=\boldsymbol{\omega}_{n}-\boldsymbol{\omega}_{0}$ fails to satisfy $\|\boldsymbol{\omega}'\|_{\infty}=O(\sigma \cdot \sqrt{\log m})$ with probability $m^{-1}$. To control the failing probability, the algorithm of Theorem~\ref{thm:random_Bana} will repeat Procedure~\textsc{SamplingBanaszczyk} for $n$ times and verify $\|\boldsymbol{\omega}'\|_{\infty}=O(\sigma \cdot \sqrt{\log m})$ each time until it succeeds. If all $n$ calls failed, the algorithm returns a fixed solution from Theorem~\ref{subgaussian_Gsw_HSSZ}. %For each call, the coloring $\|\sum_{i=1}^n \bx(i) \cdot \bv(i)\|= O(\sqrt{\log m})$  with probability $1-m^{-1}$ by taking $\delta=m^{-1}$ in Theorem 4.1 of \cite{LiuSS22}.
%is dominated by $\|\boldsymbol{\omega}^{\prime}\|_{\infty}$ and the discrepancy of the re-assign part. The $\|\boldsymbol{\omega}'\|_\infty$  part was studied in \cite{LiuSS22} and the re-assign term is bounded by Theorem~\ref{thm:Banaszczyk}. %\textcolor{blue}{Xue: ref to the restatement in Section~\ref{sec:Prel}} \textcolor{red}{Haoyu: Done}

To show that there are exponentially many good colorings, we use the following lemma to prove that the for-loop in Algorithm~\ref{alg:Bana} outputs any fixed string in $\{0, \pm 1\}^n$ with an exponentially small probability. Since the success probability is $1-m^{-1}$, after proving Algorithm~\ref{alg:Bana} outputting any coloring with exponentially small probability $c^n$ for some $c<1$, we have the number of good colorings is at least $(1-m^{-1})c^{-n}.$
%by saying that from start point $\boldsymbol{\omega}_0$ the algorithm  $\textsc{PARTIALCOLORING}_{\sigma}(\bv_{1},\dots,\bv_{t})$ of \cite{LiuSS22} randomly outputs any one partial coloring with small probability.

\begin{lemma}\label{pr_for_LSS}\label{lem:induction_bana}
% Let $\alpha=0.942$ and $\beta=0.9802$. 

 %\xcnote{How about $\alpha=\beta^3$ and $\beta=0.981$?}\hynote{setting $\alpha=0.942$,  $\beta=0.9802$, $B=2$, $\sigma=1$ is ok . So $\alpha\cdot (\frac{1}{1+2 \sigma^{2}ln(\beta)})^{0.5} < 0.9615$ and $0.9615\cdot 2^{0.055}\le 0.999$}
%\gynote{Let $0< \alpha,~\beta<1$ be any parameters satisfying Condition \ref{condition:alpha_beta}. }
For any $\gamma>0$, let $\sigma$ be defined in Line 3 of Algorithm~\ref{alg:Bana}. Then there exist $\alpha<1$ and $\beta<1$ such that for any $t\le n$ and any sequence $\boldsymbol{\epsilon} \in \{0, \pm 1\}^t$, $(\by_1,\ldots,\by_t)$ generated by Algorithm~\ref{alg:Bana} satisfies
 %be the partial coloring in Algorithm \ref{alg:Bana} after $t$ iterations. Suppose $|supp(\bx_{L})|=t-t_{0}$, we have
% \textcolor{blue}{Xue: It sounds wired to use "always have" since it depends on the above random event. How about use $t_0$ to denote the number of zeros and rewrite RHS as $\alpha^{(t-t_0)} \cdot 0.0361^{t_0}$?} \textcolor{red}{Haoyu: Done.}
 \begin{equation}\label{eq:induction_bana}
    \Pr_{\by}\left[ \by(1)=\boldsymbol{\epsilon}(1),\ldots,\by(t)=\boldsymbol{\epsilon}(t)|\boldsymbol{\omega}_{0} \right] \le \alpha^{t-t_{0}} \cdot r_\sigma^{t_{0}} \cdot \beta^{ \|\bv_t\|_2 \cdot \|\boldsymbol{\omega}_t\|_2^2- \|\boldsymbol{\omega}_0\|_2^2},
    \end{equation}  
    where $t_0$ denotes the number of zeros in $\big( \boldsymbol{\epsilon}(1),\dots,\boldsymbol{\epsilon}(t) \big)$ and $r_{\sigma}=\underset{\ell \in[-0.5,0.5]}{\max}r_\sigma(\ell)$; more importantly, $\alpha$ and $\beta$ satisfy $\gamma > 2 \frac{\ln(1+2\sigma^2\ln \beta)}{\ln(\alpha+r_\sigma)}$ where both $\alpha+r_\sigma$ and $1+2\sigma^2\ln \beta$ are in $(0.9,1)$. %\gynote{Why we need to in $(0.9,1)$?}
%     let $t_0$ denote the number of zeros in $\big( \boldsymbol{\epsilon}(1),\dots,\boldsymbol{\epsilon}(t) \big)$. Consider the for loop in Algorithm~\ref{alg:Bana} generating . Define $\mathbf{r}_{\sigma}=\underset{x\in[-0.5,0.5]}{\max}r_\sigma(x)$ and recall that $\boldsymbol{\omega}_t=\sum_i \by(i)\cdot \bv_i$
     %Specifically, $\alpha$  and $\beta $ can be $0.95$ and $ 0.958$ respectively.
\end{lemma}

Here we finish the proof of Theorem~\ref{thm:random_Bana} and defer the proof of Lemma~\ref{lem:induction_bana} to Section~\ref{sec:induction_bana}.  

%%%%%%%%%%%%%%%%%%%%%%%%%%%%%%%%%%

\begin{proofof}{Theorem~\ref{thm:random_Bana}}
First, we bound the discrepancy of Procedure~\textsc{SamplingBanaszkzyk}. We rewrite $\|\boldsymbol{\omega}\|_{\infty} \le \|\boldsymbol{\omega}'\|_{\infty}+\|\boldsymbol{\omega}''\|_{\infty}$. Then $\|\boldsymbol{\omega}'\|_{\infty}=O(\sigma \cdot \sqrt{\log m})$ with probability $1-m^{-1}$ from Theorem~\ref{thm:gaussian_stationary} (in fact $\sigma=O(\sqrt{1/\gamma})$ from the proof of Lemma~\ref{lem:mono_on_prob}). And $\|\boldsymbol{\omega}''\|_{\infty}=O(\sqrt{\log m})$ from Theorem~\ref{subgaussian_Gsw_HSSZ}.
%Any output of Algorithm \ref{alg:Bana} is a union of output from \cite{LiuSS22} and output  from \cite{BDGL19}, by replacing $0$s in \cite{LiuSS22} via algorithm in \cite{BDGL19}  with $\left\{\pm 1\right\}$. It's a direct corollary that the output coloring $\bx_t$ satisfies $\|\sum_{i=1}^t \bx(i) \cdot \bv_i\|_{\infty} = O(\sqrt{\log mt})$ with probability at least $1-(mt)^{-\Omega(1)}$.

Then we show the probability of outputting any coloring is exponentially small. By Lemma \ref{lem:induction_bana} and $\boldsymbol{\omega}_{0}\sim N(0,\sigma^{2}I_{m})$, Procedure~\textsc{SamplingBanaszkzyk} has
    \begin{align*}
    &\Pr\left[\by(1)=\boldsymbol{\epsilon}(1),\dots,\by(n) =\boldsymbol{\epsilon}(n) \right]\\
    &=
    \int_x \Pr[\by(1)=\boldsymbol{\epsilon}(1),\ldots,\by(n)=\boldsymbol{\epsilon}(n)|\boldsymbol{\omega}_0=\mathbf{g}] \cdot \Pr[\boldsymbol{\omega}_0=\mathbf{g}] \mathrm{d} \mathbf{g}
    \\
    & = \int_{\R^{m}} \Pr[\by(1)=\boldsymbol{\epsilon}(1),\ldots,\by(n)=\boldsymbol{\epsilon}(n)|\boldsymbol{\omega}_0=\mathbf{g}] \cdot (2\pi \sigma^{2})^{-\frac{m}{2}}\cdot e^{-\|\mathbf{g}\|_{2}^{2}/(2\sigma^{2})} \mathrm{d} \mathbf{g} \tag{since $\boldsymbol{\omega}_{0}\sim N(0,\sigma^{2}I_{m})$}\\
    &\le  (2\pi \sigma^{2})^{-\frac{m}{2}} \cdot \alpha^{n-t_{0}} \cdot \beta^{ \|\bv_n\|_2 \cdot \|\boldsymbol{\omega}_n\|_2^2} \cdot r_\sigma^{t_{0}} 
 \cdot \int_{\R^{m}}e^{-\|\mathbf{g}\|_{2}^{2} \cdot (1/(2\sigma^{2}) + \ln(\beta))}\cdot \mathrm{d} \mathbf{g} \tag{apply Lemma~\ref{lem:induction_bana} and rewrite $\beta^{-\|\mathbf{g}\|_2^2}$ as $e^{- \|\mathbf{g}\|_2^2 \cdot \ln \beta}$}
\\ %\tag{\text{use }\beta<1} \\
    &\le (2\pi \sigma^{2})^{-\frac{m}{2}} \cdot \alpha^{n-t_{0}} \cdot r_\sigma^{t_{0}} \cdot \left(\pi \cdot \frac{1}{1/(2\sigma^2) + \ln \beta}\right)^{\frac{m}{2}}\tag{note that $\beta<1$} \\
    & = \alpha^{n} \cdot (r_\sigma/\alpha)^{t_0} \cdot \left(\frac{1}{ 1 + 2\sigma^2 \ln \beta}\right)^{m/2}. \numberthis \label{eqn:prob_y}
%    &\le 0.9615^{n}\cdot0.0384^{t_0}.\tag{$\alpha = 0.942$, $\beta =0.9802$, $\sigma=1$ and $n\ge m$}
\end{align*}

%\textcolor{blue}{Xue: rewrite it using $t_0$? Then we can assume $t_0\le \gamma t$ because $\Pr[t_0>\gamma t] \le C^{-t}$ small.} \textcolor{red}{Haoyu: Done}

%\textcolor{red}{Xue: Add a concentration bound for $t_0$.}
%\eqref{eq:cal_over_omega0}

Now we show that the final output is sufficiently random. The loop of Algorithm~\ref{alg:Bana} generates a partial coloring $\by \in \{0,\pm 1\}^{n}$. Let $\bz \in \{\pm 1\}^{E}$ be the partial coloring generated in line $13$ on the empty set $E$, and the final output be $\bx=\by\left(\overline{E}\right) \circ \bz(E)$. %\hynote{We should use $\circ$.} \hynote{a samll bug}
After replacing all $0$s of $\by$ in Line $13$, there are at most $\binom{n}{t_0}$ possible $\by$s that could result in the same full coloring $\bx\in \left\{\pm 1\right\}^{n}$. 
By the upper bound of $\Pr[\by]$ in \eqref{eqn:prob_y}, we bound $\Pr[\bx=\boldsymbol{\epsilon}]$ as 
\begin{equation}
\begin{aligned}
\sum_{E \subseteq [n]} \Pr[\boldsymbol{\epsilon}=\by\left(\overline{E}\right) \circ \bz(E)]&\le  \sum_{t_0=0}^n \binom{n}{t_0}\alpha^{n} \cdot (r_\sigma/\alpha)^{t_0} \cdot \left(\frac{1}{ 1 + 2\sigma^2 \ln \beta}\right)^{m/2}\\
&=\alpha^n(1+\frac{r_\sigma}{\alpha})^n \cdot \left(\frac{1}{ 1 + 2\sigma^2 \ln \beta}\right)^{m/2}\\
&=(\alpha+r_\sigma)^n \cdot \left(\frac{1}{ 1 + 2\sigma^2 \ln \beta}\right)^{m/2}\\
& = \left(\frac{\alpha+r_\sigma}{ (1 + 2\sigma^2 \ln \beta)^{\frac{m}{2n}}}\right)^{n} \le \left(\frac{\alpha+r_\sigma}{ (1 + 2\sigma^2 \ln \beta)^{\frac{2}{\gamma}}}\right)^{n}.
%\sum_{t_0 \le n} \binom{n}{t_0}\cdot 0.0384^{t_0} \cdot 0.9615^{n}\\
%&\le \sum_{t_0 \le n} \binom{n}{t_0}\cdot 0.0384^{t_0} \cdot 0.9615^{n}\\
%&=(1+0.0384)^n\cdot (0.9615)^n\\
%&\le 0.9985^n.
\end{aligned}
\end{equation} 
In the last step, we uses the property $n \ge \gamma m$.
%\gynote{We just sum over $t_0\le n$ and obtain same bound. It seems that upper bound for $t_0$ is useless}
%\xcnote{Modify the rest.}
%\hynote{Maybe  $(1+0.0384)^{n}$ is more clear}
%By Stirling's approximation 
%\textcolor{red}{(Haoyu: Dose this approximation's accuracy suffice? 
% Stirling formula: .$1000!/(\sqrt{2\pi \cdot 1000}\cdot (1000/e)^{1000})\approx 1.00008$. Besides, for $0<\gamma<1$, $2^{(-\gamma \log_{2}\gamma - (1-\gamma)\log_{2}(1-\gamma))t}<e^{\frac{3\epsilon}{2}\log_{2}(\frac{1}{\epsilon)}}$)}
 
 %\begin{equation}
%    \binom{n}{t_{0}} = 
 %   \binom{n}{\gamma n} = 2^{(-\gamma \log_{2}\gamma - (1-\gamma)\log_{2}(1-\gamma)+o(1))n}.
%\end{equation}
%Since \begin{align}
%    \binom{n}{\gamma n}\cdot 0.0384^{\gamma n} &\le  2^{(-4.7\gamma-\gamma \log_{2}\gamma - (1-\gamma)\log_{2}(1-\gamma))n}\notag \\ &\le 2^{0.0545n} \tag{holds for any $ \gamma\in [0,1]$},
%\end{align}

%we have $\Pr[\bx=\by\vee \bz]= O\left( 2^{0.0545n}\cdot 0.9615^{n} \right) = O(0.9986^{n})$. 
After repeating Procedure~\textsc{SamplingBanaszkzyk} for $n$ times, this probability becomes
\[
\left(\frac{\alpha+r_\sigma}{ (1 + 2\sigma^2 \ln \beta)^{\frac{2}{\gamma}}}\right)^{n}/(1-m^{-1})+m^{-n} \] where $m^{-n}$ is the probability that each call fails to generate $\|\boldsymbol{\omega}'\|_{\infty}=O(\sqrt{\log m})$.

%The number of all possible colorings is exponentially large, as any coloring is output with exponentially small probability.

By the last property of Lemma \ref{lem:induction_bana}, $\alpha$ and $\beta$ gives a constant $\frac{\alpha+r_\sigma}{ (1 + 2\sigma^2 \ln \beta)^{\frac{2}{\gamma}}}$ that is strictly less than 1. So each coloring is output with an exponentially small probability. And there are exponentially many good colorings.

\end{proofof}

%%%%%%%%%%%%%%%%%%%%%%%

\subsection{\texorpdfstring{$0, \pm 1$}{} Random Walk in \texorpdfstring{\cite{LiuSS22}}{}}\label{sec:stationary_RW}
%The section $2$ of  provides a $0,\pm 1$ random walk limited on the  real line. 
We review the Gaussian fix-point random walk $J^{\sigma}_x$ introduced in \cite{LiuSS22} and discuss its properties here.

\begin{definition}[$0, \pm 1$ walk]\label{LSSrw}
% Given a $\sigma >0$ and let $p(i,j)$ be the transition probability from location $i$ to location $j$ where $i,j\in \R$. For $x> 0.5$, let  $p(x,x+1)= p_{\sigma}(x)$ and  $p(x,x-1)=1-p_{\sigma}(x)$ and symmetrically for $x< -0.5$, let $p(x,x-1)= p_{\sigma}(-x)$ and  $p(x,x+1)=1-p_{\sigma}(x)$. Besides for $x\in [-0.5,0.5]$, let $p(x,x+1)= p_{\sigma}(x)$,  $p(x,x-1)=p_{\sigma}(-x)$ and $p(x,x)=r_{\sigma}(x)$. Here 
Given $\sigma \ge 1$ and $x \in \mathbb{R}$, consider two transition probability functions,
 \begin{equation}\label{def:p_sigma}
     p_{\sigma}(x)=\sum_{j\ge 1}(-1)^{j-1}\exp\left(\frac{-j\cdot (2x+j)}{2\sigma^{ 2}}\right) \textit{ and }
 \end{equation}
  \begin{equation}
     r_{\sigma}(x)=\sum_{j=-\infty }^{\infty}(-1)^{j}\exp\left(\frac{-j\cdot(2x+j)}{2\sigma^{2}}\right).
 \end{equation}
Define the random variable $J_x^{\sigma}$ as
%And naturally a random variable $J_{x}^{\sigma}$, denoting the random walk, could be defined as:
\begin{equation}
\label{eq6}
J_{x}^{\sigma}=\left\{
\begin{aligned}
0 &  \textit{ with prob. } r_{\sigma}(x), \textit{ for } x \in [-0.5,0.5]\\
1 &  \textit{ with prob. } p_{\sigma}(x), \textit{ for } x \in [-0.5,0.5]\\
-1 &  \textit{ with prob. } p_{\sigma}(-x), \textit{ for } x \in [-0.5,0.5]\\
1 &  \textit{ with prob. } p_{\sigma}(x), \textit{ for } x>0.5\\
-1 &  \textit{ with prob. } 1-p_{\sigma}(x), \textit{ for } x>0.5\\
-1 &  \textit{ with prob. } p_{\sigma}(-x), \textit{ for } x<-0.5\\
1 &  \textit{ with prob. } 1-p_{\sigma}(-x), \textit{ for } x<-0.5.
\end{aligned}
\right.
\end{equation}
\end{definition}

Liu, Sah, and Sawhney \cite{LiuSS22} showed that $p_{\sigma}(x)+p_{\sigma}(-x)+r_{\sigma}(x)=1$ and $J^{\sigma}_x$ defines a random walk on $\mathbb{R}$ whose stationary distribution is Gaussian.
%And this random walk keep Gaussian:   

\begin{lemma}[Lemma 2.3 of \cite{LiuSS22}]
If $Z\sim N(0,\sigma^{2})$, then $Z+J_{Z}^{\sigma}$ is distributed as $N(0,\sigma^{2})$.  
\end{lemma}

To prove that this random walk is sufficiently random, we prove the following properties for $p_\sigma$ and $r_\sigma$, which may be of independent interest.

\begin{lemma}\label{lem:mono_on_prob}
The transition probabilities of $J^{\sigma}_x$ satisfy the following monotonicity properties:
\begin{enumerate}
\item For any fixed $\sigma\ge 1$, $p_{\sigma}(x)$ is strictly monotonically decreasing on $x\in[-0.5,+\infty)$.

\item For any $x\in [0.5,+\infty)$, $p_{\sigma}(x)$ is strictly monotonically increasing on $\sigma\in[1,+\infty)$. 
%\textcolor{red}{(Haoyu: As we have set $\sigma=1$ in Algorithm \ref{alg:Bana}, this property seems not necessary. If we want to make the proof satisfying for all $\sigma\ge 1$, we may need to redesign potential function with $\sigma$ like $f(i)^{1/\sigma^{2}}$. The restriction lies in case (iv). Because $f(i)^{1/\sigma^{2}}\ge f(i)$ for $\sigma\ge 1$ and $f(i)\le 1$, it seems alright to replace f(i) with $f(i)^{1/\sigma^{2}}$ directly . This change could let us get rid of  this monotone property and $\sigma=1$.)}\textcolor{green}{Guangyi: I have shown this property for $x\ge 1/2$}

\item For $x = -0.5$, $p_{\sigma}(-0.5)$ is strictly monotonically decreasing on $\sigma\in[1,+\infty)$. %\textcolor{blue}{Xue: Cannot understand this property, seems contradicted with Property 2.}\textcolor{green}{Guangyi: I think $p_{\sigma}(x)$ is supposed to be decreasing for $x\le -0.5$ and increasing for $x\ge 0$.}

\item For any $x\in [-0.5,0.5]$, $r_{\sigma}(x)$ is strictly monotonically decreasing on $\sigma\in[1,+\infty)$. And $\underset{x\in [-0.5,0.5]}{\max} r_{\sigma}(x)=r_{\sigma}(0)$ approaches to $0$ as $\sigma\rightarrow\infty.$
\end{enumerate}
\end{lemma}

The proof of Lemma~\ref{lem:mono_on_prob} is deferred to Section~\ref{sec:pf_mono}. We will use these properties in the following way:
\begin{corollary}\label{cor:upper_bounds_pr}
    \begin{enumerate}
        \item For any $x\ge -0.5$ and any $\sigma' \ge \sigma \ge 1$, $p_{\sigma'}(x) \le p_{\sigma'}(-0.5)\le p_{\sigma}(-0.5).$ %\drnote{There are two same subscript $p_{\sigma}$?}

        \item For a fixed $\sigma$, let $r_{\sigma}=r_{\sigma}(0)$. For any $x\in [-0.5,0.5]$ and any $\sigma' \ge \sigma$, $r_{\sigma'}(x) \le r_{\sigma}$.
%$\underset{x\in [-0.5,0.5]}{\max} r_{\sigma}(x)\le \underset{x\in[-0.5,0.5]}{\max} r_{\sigma_0}(x).$
    \end{enumerate}
\end{corollary}

% \begin{corollary}\label{trans_pr} 
%     For the $(i+1)$'th step, $x_{i+1}=\frac{\langle \boldsymbol{\omega}_{i},v_{i+1} \rangle }{\|v_{i+1}\|_{2}}$ and $\sigma^{\prime}=\frac{\sigma}{\|v_{i+1}\|_{2}}$, the transition probability 

% \begin{align}
%        p_{\sigma^{\prime}}(x_{i+1})&=\sum_{j\ge 1}(-1)^{j-1}exp\left(\frac{-j(2x_{i+1}+j)}{2\sigma^{\prime 2}}\right),\label{eq_tp1}\\
%        &=\sum_{j\ge 1}(-1)^{j-1} exp\left(\frac{ - j \cdot \|v_{i+1}\|_{2} \cdot \langle \boldsymbol{\omega}_{i},v_{i+1}\rangle}{\sigma^{2}}-\frac{j^{2}\|v_{i+1}\|_{2}^{2}}{2\sigma^{2}}\right) \label{eq:tp2}.
% \end{align}  
% As equation \ref{eq_tp1}, $p_{\sigma^{\prime}}(x_{i+1})$ is a monotone decreasing function of $x_{i+1}$. On the other hand as equation \ref{eq:tp2},  $p_{\sigma^{\prime}}(x_{i+1})$  is a function of two independent variables:  $\|v_{i+1}\|_{2}$ and $\langle \boldsymbol{\omega}_{i},v_{i+1}\rangle$, and monotonically decreases by them respectively.
% Moreover for any $f\in[0,0.5]$, \cite{LiuSS22} concludes that \begin{equation}
%         p_{\sigma^{\prime}}(f)+p_{\sigma^{\prime}}(-f)+r_{\sigma^{\prime}}(f)=1,
%     \end{equation} so we get \begin{equation}
%         p_{\sigma^{\prime}}(x)\le p_{\sigma^{\prime}}(0)\le 0.5.
%     \end{equation}
% \end{corollary}
% \textcolor{red}{Haoyu: Need to be completed}
%%%%%%%%%%%%%%%%%%%%%%%

\subsection{Proof of Lemma~\ref{lem:induction_bana}}\label{sec:induction_bana}
In this section, we assume $\sigma\ge 1$ taken in Line 3 of Algorithm~\ref{alg:Bana} satisfying $r_{\sigma}<\frac{1}{6}e^{-\frac{24}{\gamma}}$ and $\gamma<0.5$ without loss of generality. We fix $B=\frac{24\sigma^2}{\gamma}$ as a quite large constant and choose $\beta= 1-\frac{1}{4B}e^{-\frac{B}{\sigma^2}}=1-\frac{\gamma}{96\sigma^2}e^{-\frac{24}{\gamma}}$ and $\alpha=1-\frac{1}{3}e^{-\frac{B}{\sigma^2}}=1-\frac{1}{3}e^{-\frac{24}{\gamma}}$ such that the following condition holds in this section.

\begin{condition}\label{condition:alpha_beta}
We assume that $\alpha,\beta,B$, and $\sigma$ satisfy the following properties (which will be verified in the end of this section):
\begin{enumerate}
    \item \label{sec5:condition1}$\alpha\beta^{-2B+1}\ge 1.$%(From case \ref{sec5:Case(v)}.)
    \item\label{sec5:condition2} $\alpha\beta^2\ge p_{\sigma}(-0.5).$ %(From case \ref{sec5:Case(iii)})
    \item\label{sec5:condition3} $\alpha\beta^{1+2B}\ge p_{\sigma}(-0.5).$ %(From case \ref{sec5:Case(i)})
    \item\label{sec5:condition4} $\alpha \beta^{\|\bv_{i+1}\|_2^3} \notag >1-\exp \left( \frac{-B}{\sigma^2} - \frac{\|\bv_{i+1}\|_2^2}{2\sigma^2} \right) + \exp \left( \frac{-2B}{\sigma^2} - \frac{2\|\bv_{i+1}\|_2^2}{\sigma^2} \right)$. %(From case \ref{sec5:Case(ii)})
    \item\label{sec5:condition5} $\exp\left(\frac{ - \langle \boldsymbol{\omega}_i, \bv_{i+1}\rangle \cdot \|\bv_{i+1}\|_2}{\sigma^{2}}-\frac{\|\bv_{i+1}\|_{2}^{2}}{2\sigma^{2}}\right)\le \alpha \cdot \beta^{\|\bv_{i+1}\|_2 \cdot (2\langle \boldsymbol{\omega}_i, \bv_{i+1}\rangle + \|\bv_{i+1}\|_2^2)}$ for  any $\langle\boldsymbol{\omega}_i,\bv_{i+1}\rangle\ge B/\|\bv_{i+1}\|_2$. %(From case \ref{sec5:Case(iv)})
\end{enumerate}
\end{condition}

We verify that $\alpha=1-\frac{1}{3}e^{-B/\sigma^2},\beta=1-\frac{1}{4B}e^{-B/\sigma^2},B=\frac{24\sigma^2}{\gamma}$ and $\sigma$ with $r_{\sigma}<\frac{1}{6} e^{-24/\gamma}$ satisfy Condition~\ref{condition:alpha_beta} in the end of this section. Now we finish the proof of Lemma~\ref{lem:induction_bana} by rewriting
\begin{equation}
\Pr[\by(1)=\boldsymbol{\epsilon}(1),\ldots,\by(t)=\boldsymbol{\epsilon}|\boldsymbol{\omega}_0]=
\Pr[\by(1)|\boldsymbol{\omega}_0] \cdot \Pr[\by(2)|\by(1), \boldsymbol{\omega}_0] \cdots \Pr[\by(t)|\by(t-1), \ldots,\boldsymbol{\omega}_0].
\end{equation}
Notice that $\Pr[\by(i)|\by(i-1), \ldots,\by(1), \boldsymbol{\omega}_0]=\Pr[J_{\ell_i}^{\sigma'}=\by(i)]$ for $\ell_i=\langle \boldsymbol{\omega}_{i-1}, \bv_i \rangle/\|\bv_i\|_2$ and $\sigma'=\sigma/\|\bv_i\|_2$. Since $\bv_1,\ldots,\bv_n$ are fixed in Algorithm~\ref{alg:Bana}, $\Pr[\by(i)|\by(i-1), \ldots,\by(1), \boldsymbol{\omega}_0]=\Pr[J_{\ell_i}^{\sigma'}=\by(i)]$ depends only on $\langle \boldsymbol{\omega}_{i-1},\bv_i \rangle$ and we simplify the notation as $\Pr[\by(i)|\langle \boldsymbol{\omega}_{i-1},\bv_i \rangle]$. We use a potential function $f(\by,i)$ to dominate $\Pr \left[\by(1),\dots, \by(i+1)|\boldsymbol{\omega}_{0}\right]$.

\begin{definition}[Potential function]
%Given $\alpha=0.942$ and $\beta=0.9802$
Assuming $0<\alpha,~\beta<1$ satisfying Condition \ref{condition:alpha_beta}, for any $\bv_1,\ldots,\bv_i$ and $\by(1),\ldots,\by(i)$, let $\boldsymbol{\omega}_i=\by(1) \cdot \bv_1 + \cdots + \by(i) \cdot \bv_i$. For each $i$ we define $f(\by, i)$ by induction. If $\by(i)=\pm 1$, let
\begin{equation}
    f(\by,i)=f(\by,i-1)\cdot \alpha\cdot \beta^{\|\bv_i\|_2\cdot (\|\boldsymbol{\omega}_i\|_2^2-\|\boldsymbol{\omega}_{i-1}\|_2^2)},
\end{equation}
and if $\by(i)=0$, let
\begin{equation}
    f(\by,i)=f(\by,i-1)\cdot r_{\sigma/\|\bv_i\|_2}(\ell_i).
\end{equation}
%For $\by(i)=\pm 1$, define
% \begin{equation}\label{def:potential_func}
%     f(i)= \alpha^{i} \cdot \beta^{\|\bv_{i}\|_{2} \cdot \|\boldsymbol{\omega}_{i}\|_{2}^{2}-\|\boldsymbol{\omega}_{0}\|_{2}^{2}},
% \end{equation}
% and for $\by(i) =0$, define
% \begin{equation} 
%     f(i)=f(i-1) \cdot r_{\sigma/\|\bv_{i}\|_{2}}(x_{i}).
% \end{equation}
\end{definition}

First of all, We show $f(\by,t) \le \beta^{\|\bv_t\|_2\|\boldsymbol{\omega}_t\|_2^2-\|\boldsymbol{\omega}_0\|_2^2}\cdot r_{\sigma}^{t_0}$. By the definition,
\begin{equation}
\begin{aligned}
    f(\by,t)&=\alpha^{t-t_0} \beta^{\sum_{i\in[t],\by(i)\ne 0}\|\bv_i\|_2\cdot (\|\boldsymbol{\omega}_i\|_2^2-\|\boldsymbol{\omega_{i-1}}\|_2^2)}\prod_{\by(i)=0}r_{\sigma/\|\bv_i\|_2}(\ell_i).\\
\end{aligned}
\end{equation}
Let $i_1< i_2< \ldots <i_k$ be coordinates such that $\by(i_{j})\ne 0$ for any $1\le j\le k$. We have 

\begin{align*}
    &\sum_{i=1,\by(i)\ne 0}^t \|\bv_i\|_2\cdot(\|\boldsymbol{\omega}_i\|_2^2-\|\boldsymbol{\omega}_{i-1}\|_2^2)\\ 
    &=\sum_{s=1}^k \|\bv_{i_s}\|_2\cdot(\|\boldsymbol{\omega}_{i_s}\|_2^2-\|\boldsymbol{\omega}_{i_{s}-1}\|_2^2)\\
    &=\sum_{s=1}^k \|\bv_{i_s}\|_2\cdot(\|\boldsymbol{\omega}_{i_s}\|_2^2-\|\boldsymbol{\omega}_{i_{s-1}}\|_2^2)\tag{use $\boldsymbol{\omega}_{i_s-1}=\boldsymbol{\omega}_{i_{s-1}}$}\\ %\tag{noting that $\omega_{i_{s-1}}=\omega_{i_s-1}$} \\
&=\sum_{s=2}^{k-1} \|\boldsymbol{\omega}_{i_s}\|_2^2\cdot(\|\bv_{i_s}\|_2-\|\bv_{i_{s+1}}\|_2)+\|\bv_{i_k}\|_2\cdot\|\boldsymbol{\omega}_{i_{k}}\|_2^2-\|\bv_{i_1}\|_2\cdot\|\boldsymbol{\omega}_{0}\|_2^2\\
&\ge \|\bv_{i_k}\|_2\cdot \|\boldsymbol{\omega}_{i_{k}}\|_2^2-\|\bv_{i_1}\|_2\cdot \|\boldsymbol{\omega}_{0}\|_2^2\tag{use $\|\bv_{i}\|_2\ge \|\bv_{i+1}\|_2$}\\
&=\|\bv_{i_k}\|_2\cdot \|\boldsymbol{\omega}_{t}\|_2^2-\|\bv_{i_1}\|_2\cdot \|\boldsymbol{\omega}_{0}\|_2^2\tag{use $\|\bv_{i_1}\|_2\le 1$ and $\|\bv_{i_k}\|_2\ge \|\bv_t\|_2$}\\
&\ge \|\bv_{t}\|_2\cdot\|\boldsymbol{\omega}_{t}\|_2^2-\|\boldsymbol{\omega}_{0}\|_2^2.
\end{align*}
By Corollary \ref{cor:upper_bounds_pr}, $r_{\sigma/\|\bv_i\|_2}(\ell_i)\le r_{\sigma}$. We have 
\begin{equation}\label{equ:potential_bound}
    f(\by,t)\le \alpha^{t-t_0}\cdot \beta^{\|\bv_t\|_2\cdot\|\boldsymbol{\omega}_t\|_2^2-\|\boldsymbol{\omega}_0\|_2^2}\cdot r_{\sigma}^{t_0}.
\end{equation}

For $\sigma,~B$ and $\alpha,~\beta$ satisfying Condition \ref{condition:alpha_beta}  we apply induction on $t$ to prove 
\begin{equation}\label{eq:hypothesis_Bana}
    \text{hypothesis: for any $t$,} \qquad f(\by,t) \ge \Pr[\by(1),\ldots,\by(t) | \boldsymbol{\omega}_0].
\end{equation} 
From \eqref{equ:potential_bound}, this shows the probability bound in Lemma~\ref{lem:induction_bana}. In the end of this section, we verify the last property of Lemma~\ref{lem:induction_bana} and Condition~\ref{condition:alpha_beta}.

\paragraph{Base case:} For convenience, assume $\|\bv_0\|_2=1$ such that \eqref{eq:hypothesis_Bana} is true for the base case $t=0$. 
\paragraph{Induction step: } From $i$ to $i+1$, it suffices to show 
\begin{equation}
    f(\by,i+1)/f(\by,i)\ge \Pr \left[\by(i+1)|\by(1),\dots, \by(i),\boldsymbol{\omega}_{0}\right]=\Pr[\by(i+1)|\boldsymbol{\omega}_i],
\end{equation}
We consider 5 cases depending on $\by(i+1)=1$ or $-1$ and whether $|\langle \boldsymbol{\omega}_{i},\bv(i+1)\rangle| \ge B/\|\bv_{i+1}\|_{2}$ %(for $B=2$)
or not. Recalling that $\ell_{i+1}=\frac{\langle \boldsymbol{\omega}_{i},\bv_{i+1}\rangle}{\|\bv_{i+1}\|_{2}}$ is the location of the $i-$th step  in $\R$. %In the upcoming proof, we take $\alpha=0.942$, $\beta=0.9802$, $B=2$, and $\sigma=1$. %\xcnote{How about $B=2$, $\alpha=\beta^3$, $\beta=0.981$?}

\begin{enumerate}[label=(\roman*)]
    \item \label{sec5:Case(i)}$\by(i+1)=1$ and $-0.5\cdot\|\bv_{i+1}\|_{2} \le \langle \boldsymbol{\omega}_{i},\bv_{i+1}\rangle \le B/\|\bv_{i+1}\|_{2}$: In this case, $\ell_{i+1}\ge -0.5$ and Corollary \ref{cor:upper_bounds_pr} shows
    \begin{equation}
    \begin{aligned}
    \Pr\bigg[ \by(i+1)=1 \big| \langle \boldsymbol{\omega}_{i},\bv_{i+1}\rangle \bigg]& 
     = p_{\sigma/\|\bv_{i+1}\|_{2}}(\ell_{i+1})\le p_{\sigma}(-0.5).
    \end{aligned}
    \end{equation} 
    And by Property \ref{sec5:condition3} of Condition \ref{condition:alpha_beta}, 
    \begin{equation}
    f(\by,i+1)/f(\by,i)=\alpha \cdot \beta^{\|\bv_{i+1}\|_2 \cdot (2\langle \boldsymbol{\omega}_i, \bv_{i+1}\rangle + \|\bv_{i+1}\|_2^2)} \ge \alpha \cdot \beta^{\|\bv_{i+1}\|_2^3 + 2B}\ge p_{\sigma}(-0.5).
    \end{equation}
            
    \item \label{sec5:Case(ii)} $\by(i+1)=-1$ and $0.5\cdot \|\bv_{i+1}\|_{2}< \langle \boldsymbol{\omega}_{i},\bv_{i+1}\rangle \le B/ \|\bv_{i+1}\|_{2}$: In this case $  \ell_{i+1}=\frac{\langle \boldsymbol{\omega}_{i},\bv_{i+1}\rangle}{\|\bv_{i+1}\|_{2}}>0.5,$ by Property 4 of Lemma~\ref{lem:mono_on_prob},  \begin{align*}
    \Pr\bigg[ \by(i+1)=-1 \big| \langle\boldsymbol{\omega}_i,\bv_{i+1} \rangle \bigg]& = 1 
     - p_{\sigma/\|\bv_{i+1}\|_2}(\ell_{i+1})\\
      & \le 1 - p_{\sigma/\|\bv_{i+1}\|_2}(B/ \|\bv_{i+1}\|_2^{2}).
    \end{align*} We simplify $p_{\sigma/\|\bv_{i+1}\|_2}(B/ \|\bv_{i+1}\|_2^{2})$ as 
    
    \begin{align}
        p_{\sigma/\|\bv_{i+1}\|_2}(B/ \|\bv_{i+1}\|_2^{2})&=\sum_{j\ge 1}(-1)^{j-1} \exp\left(\frac{ - j \cdot B}{\sigma^{2}}-\frac{j^{2} \cdot \|\bv_{i+1}\|_{2}^{2}}{2\sigma^{2}}\right)\notag \\
        &\ge exp \left( \frac{-B}{\sigma^2} - \frac{\|\bv_{i+1}\|_2^2}{2\sigma^2} \right) - \exp \left( \frac{-2B}{\sigma^2} - \frac{2\|\bv_{i+1}\|_2^2}{\sigma^2} \right). \label{eq:1_bana_mono_proof}
    \end{align}
    
        And we can verify the following inequality by Property \ref{sec5:condition4} of Condition \ref{condition:alpha_beta}.
    \begin{align}
        f(\by,i+1)/f(\by,i) & =\alpha \cdot \beta^{\|\bv_{i+1}\|_2 \cdot ( - 2\langle \boldsymbol{\omega}_i, \bv_{i+1}\rangle + \|\bv_{i+1}\|_2^2)} \notag \\
        &\ge \alpha \beta^{\|\bv_{i+1}\|_2^3} \notag \\
        &>1-\exp \left( \frac{-B}{\sigma^2} - \frac{\|\bv_{i+1}\|_2^2}{2\sigma^2} \right) + \exp \left( \frac{-2B}{\sigma^2} - \frac{2\|\bv_{i+1}\|_2^2}{\sigma^2} \right)\label{eq:30_lem_5.3}\\
        &\ge 1 - p_{\sigma/\|\bv_{i+1}\|_2}(B/ \|\bv_{i+1}\|_2^{2})\label{eq:31_lem_5.3}\\
        &\ge \Pr\left[ \by(i+1)=-1 \big| \langle\boldsymbol{\omega}_i,\bv_{i+1} \rangle \right]. \notag
    \end{align}

%The inequality \eqref{eq:31_lem_5.3} holds by inequality \eqref{eq:1_bana_mono_proof}. And the inequality \eqref{eq:30_lem_5.3} holds by three facts:  
%(1) $1-\exp \left( \frac{-B}{\sigma^2} - \frac{\|\bv_{i+1}\|_2^2}{2\sigma^2} \right) + \exp \left( \frac{-2B}{\sigma^2} - \frac{2\|\bv_{i+1}\|_2^2}{\sigma^2} \right)$ is monotonically increasing on $\|\bv_{i+1}\|_{2}^{2}\in [0,1]$ when $B=2$ and $\sigma =1$; (2) $0.921\ge 1-\exp(-2-1/2)+\exp(-4-2)$; (3) $\alpha\cdot\beta^{\|\bv_{i+1}\|_{2}^{3}}\ge \alpha\cdot\beta\ge 0.923.$

%\hynote{to prove $0.942\cdot 0.9802\ge 1-exp(-2-\|v\|_{2}^{2}/2)+exp(-4-2\|v\|_{2}^{2})$. By calculator, $exp(-2-z/2)-exp(-4-2z)$ is decreasing in $[0,1]$ and $exp(-2-1/2)-exp(-4-2)\ge 0.079\ge 0.077\ge 1- 0.942\cdot 0.9802$  } 

    \item \label{sec5:Case(iii)}$\by(i+1)=-1$ and $-0.5\cdot \|\bv_{i+1}\|_{2}\le \langle \boldsymbol{\omega}_{i},\bv_{i+1}\rangle \le 0.5\cdot\|\bv_{i+1}\|_{2}$. In this case \begin{equation}
        \ell_{i+1}=\frac{\langle \boldsymbol{\omega}_{i},\bv_{i+1}\rangle}{\|\bv_{i+1}\|_{2}}\in [-0.5,0.5].
    \end{equation}
    By monotonicity of $p_{\sigma/\|\bv_{i+1}\|_{2}}(\ell_{i+1})$ in Corollary \ref{cor:upper_bounds_pr}, we have \begin{equation}\begin{aligned}
    \Pr[\by(i+1)=-1|\langle \boldsymbol{\omega}_i,\bv_{i+1} \rangle] &= p_{\sigma/\|\bv_{i+1}\|_{2}}(-\ell_{i+1})\le p_{\sigma}(-0.5).
    \end{aligned}\end{equation}
%    \textcolor{blue}{Xue: Can not understand the last step, what is the property used there? Why is the last term $B/\|v_{i+1}\|_2^2$?}    \textcolor{red}{Haoyu\_Reply: Suppose $B\ge 0.5$, in fact $B$ could be $log(3)\approx 1.099, $ which means $\langle \boldsymbol{\omega}_{i},v_{i+1}\rangle \le 0.5\cdot\|v_{i+1}\|_{2}\le B/ \|v_{i+1}\|_2$ and $x_{i+1}=\langle \boldsymbol{\omega}_{i},v_{i+1}\rangle/\|v_{i+1}\|_2\le B/\|v_{i+1}\|_2^2$. By the monotone property of $p_{\sigma}(x)$ to get that step. It's better to bound $p_{\sigma/\|v_{i+1}\|_{2}}(-x_{i+1})\le p_{\sigma/\|v_{i+1}\|_{2}}(-0.5)$ by 0.68 directly and I have fixed up it.}
    
    And by Property \ref{sec5:condition2} of Condition \ref{condition:alpha_beta}
\begin{equation}
    f(\by,i+1)/f(\by,i)=\alpha \cdot \beta^{\|\bv_{i+1}\|_2 \cdot ( - 2\langle \boldsymbol{\omega}_i, \bv_{i+1}\rangle + \|\bv_{i+1}\|_2^2)} \ge \alpha \beta^{1+\|\bv_{i+1}\|_2^3}\ge p_{\sigma}(-0.5).
\end{equation}

    \item \label{sec5:Case(iv)}$\by(i+1)=1$ and $\langle \boldsymbol{\omega}_{i},\bv_{i+1}\rangle \ge B/\|\bv_{i+1}\|_{2}$. 
    \begin{equation}\begin{aligned}
        \Pr\bigg[ \by(i+1)=1 \big| \langle \boldsymbol{\omega}_{i},\bv_{i+1}\rangle \bigg] & =\sum_{j\ge 1}(-1)^{j-1} \exp\left(\frac{ - j \cdot \langle \boldsymbol{\omega}_i, \bv_{i+1}\rangle \cdot \|\bv_{i+1}\|_2}{\sigma^{2}}-\frac{j^{2} \cdot \| \bv_{i+1} \|_{2}^{2}}{2\sigma^{2}}\right) \\
        & \le \exp\left(\frac{ - \langle \boldsymbol{\omega}_i, \bv_{i+1}\rangle \cdot \|\bv_{i+1}\|_2}{\sigma^{2}}-\frac{\|\bv_{i+1}\|_{2}^{2}}{2\sigma^{2}}\right).
    \end{aligned}\end{equation}
    And by Property \ref{sec5:condition2} of Condition \ref{condition:alpha_beta} we have $f(\by,i+1)/f(\by,i) = \alpha \cdot \beta^{\|\bv_{i+1}\|_2 \cdot (2\langle \boldsymbol{\omega}_i, \bv_{i+1}\rangle + \|\bv_{i+1}\|_2^2)}\ge \Pr[\by(i+1)|\langle \boldsymbol{\omega}_{i},\bv_{i+1}\rangle]$.
    
    \item \label{sec5:Case(v)}$\by(i+1)=-1$ and $\langle \boldsymbol{\omega}_{i},\bv_{i+1}\rangle \ge B/\|\bv_{i+1}\|_{2}$ We bound $\Pr[\by(i+1)=-1|\boldsymbol{\omega}_i] \le 1$. And by Property \ref{sec5:condition1} of Condition \ref{condition:alpha_beta}, % we verify that for $\alpha=0.942$  and $\beta=0.9802$, 
    \begin{equation}
    f(\by,i+1)/f(\by,i)=\alpha \cdot \beta^{\|\bv_{i+1}\|_2 \cdot ( - 2\langle \boldsymbol{\omega}_i, \bv_{i+1}\rangle + \|\bv_{i+1}\|_2^2)} \ge \alpha\cdot \beta^{-2B + \|\bv_{i+1}\|_2^3}\ge 1.
    \end{equation}
\end{enumerate}
Hence by equation (\ref{equ:potential_bound}) we have
\begin{equation}
    \Pr_{\by}\left[ \by(1)=\boldsymbol{\epsilon}(1),\ldots,\by(t)=\boldsymbol{\epsilon}(t) \big| \boldsymbol{\omega}_{0} \right]\le f(\by,t)\le \alpha^{t-t_0}\cdot \beta^{\|\bv_t\|_2\|\boldsymbol{\omega}_t\|_2^2-\|\boldsymbol{\omega}_0\|_2^2}\cdot r_{\sigma}^{t_0}
\end{equation}
as desired.

Finally we show that our choices of $\alpha$ and $\beta$ satisfy both Condition \ref{condition:alpha_beta} and the last property of Lemma \ref{lem:induction_bana}. %arbitrarily small $\gamma>0$ we can choose proper $\sigma,~B$ and $\alpha,~\beta$ satisfying Condition \ref{condition:alpha_beta} and $\gamma > \frac{1}{2}\frac{\ln(1+2\sigma^2\ln \beta)}{\ln(\alpha+r_\sigma)}.$
 For $\sigma$ satisfying $r_{\sigma}<\frac{1}{6}e^{-\frac{24}{\gamma}}$ defined in Line 3 of Algorithm~\ref{alg:Bana}, 
 recall that $B=\frac{24\sigma^2}{\gamma}$, $\beta= 1-\frac{1}{4B}e^{-\frac{B}{\sigma^2}}$ and $\alpha=1-\frac{1}{3}e^{-\frac{B}{\sigma^2}}$. We show Condition~\ref{condition:alpha_beta} holds before verifing $\gamma>2\frac{\ln(1+2\sigma^2\ln \beta)}{\ln(\alpha+r_\sigma)}$. 

To show this we add some auxiliary restrictions to simplify those properties in Condition \ref{condition:alpha_beta}.
\begin{enumerate}
    \item $\alpha\ge \beta^{2B-1}$ from Property \ref{sec5:condition1}.
    
    \item If $\beta\ge e^{-\frac{1}{2\sigma^2}}$, then $\alpha \beta^{2B+1} \ge e^{-B/\sigma^2}$ implies that Property \ref{sec5:condition5} holds for any $|\langle \boldsymbol{\omega}_i, \bv_{i+1} \rangle| \ge B$. To satisfy $\alpha \beta^{2B+1} \ge e^{-B/\sigma^2}$, given $\alpha \ge \beta^{2B-1}$, we only need $\beta \ge e^{-\frac{1}{4\sigma^2}}$.
%    \[ \alpha\cdot \beta^{2\|\bv\|_2 \langle \boldsymbol{\omega_i}, \bv \rangle +\|\bv\|_2^3} \ge e^{\frac{ - \|\bv\|_2 \langle \boldsymbol{\omega_i}, \bv \rangle}{\sigma^2}  - \frac{\|bv\|_2^2}{2 \sigma^2}} \cdot \alpha \cdot \beta^{B}. \]
%    Indeed we just need $\alpha\cdot \beta^{2B-1}\ge e^{-B/\sigma^2}$, i.e., $\beta\ge e^{-\frac{B}{4B-1}\frac{1}{\sigma^2}}.$
    \item Property \ref{sec5:condition3} needs $\beta\ge p_{\sigma}(-0.5)^{1/4B}$ and Property \ref{sec5:condition2} is satisfied if Property \ref{sec5:condition1} and Property \ref{sec5:condition3} are satisfied.
\end{enumerate}

So Condition \ref{condition:alpha_beta} are fulfilled if the following restrictions are satisfied: (1) $\alpha \ge \beta^{2B-1}$; (2) $\beta\ge 0.68^{1/4B}\ge p_{\sigma}(-0.5)^{1/4B}$; (3) $\beta\ge e^{-\frac{1}{4}\frac{1}{\sigma^2}}$; (4)$\beta^{2B}>1-\exp \left( \frac{-B}{\sigma^2} - \frac{\|\bv_{i+1}\|_2^2}{2\sigma^2} \right) + \exp \left( \frac{-2B}{\sigma^2} - \frac{2\|\bv_{i+1}\|_2^2}{\sigma^2} \right)$. 

For $\sigma\ge 1$ and $\frac{B}{\sigma^2}\ge 10$, these restrictions are satisfied by $\beta = 1-\frac{1}{4B}e^{-\frac{B}{\sigma^2}}$ and $\alpha=1-\frac{1}{3}e^{-\frac{B}{\sigma^2}}$. After plugging $B=\frac{24\sigma^2}{\gamma}$, they become $\beta= 1-\frac{\gamma}{96\sigma^2}e^{-\frac{24}{\gamma}}$ and $\alpha=1-\frac{1}{3}e^{-\frac{24}{\gamma}}$. Given $r_{\sigma}<\frac{1}{6} e^{-24/\gamma}$, we have $\alpha + r_{\sigma}<1-\frac{1}{6} e^{-24/\gamma}$ and
 \begin{equation}
 \begin{aligned}
     2\cdot\frac{\ln(1+2\sigma^2\ln \beta)}{\ln(\alpha+r_\sigma)}& \le 2 \cdot 1.5\cdot \frac{-2\sigma^2\ln \beta}{1-\alpha-r_{\sigma}}\\
     &\le \frac{6\cdot \frac{\sigma^2}{2B}e^{-\frac{B}{\sigma^2}}}{\frac{1}{6}e^{-\frac{B}{\sigma^2}}}\\
     &\le \frac{18\sigma^2}{B}= \frac{3}{4} \gamma.
\end{aligned}
 \end{equation}
Hence we have proven the last property of Lemma \ref{lem:induction_bana}. 
    \subsection{Proof of Lemma~\ref{lem:mono_on_prob}}\label{sec:pf_mono}
    To study the monotonicity of $p_{\sigma}(x)$, we define two functions
    \begin{equation}
        A_{\sigma}(x)=\sum_{j\ge 1}(-1)^{j-1} \cdot \exp(-\frac{(x+j)^2}{2\sigma^2}),
    \end{equation}
    and 
    \begin{equation}\label{eq:B_sigma}
       \begin{aligned}
       B_{\sigma}(x)&=\sum_{i\ge 0}(-1)^{i} \cdot A_{\sigma}(x+i)\\
        &=\sum_{j\ge 1}(-1)^{j-1}\cdot j \cdot \exp(-\frac{(x+j)^2}{2\sigma^2}).
        \end{aligned}
    \end{equation} Then we have 
    \begin{equation}\label{p_equal_A}
        p_{\sigma}(x)=\exp(\frac{x^2}{2\sigma^2})\cdot A_{\sigma}(x).
    \end{equation}
 Moreover,    \begin{align}
    \frac{\partial p_{\sigma}(x)}{\partial x}&= \sum_{j\ge 1}(-1)^{j-1}\cdot\frac{-j}{\sigma^2}\cdot \exp(\frac{-j\cdot(2x+j)}{2\sigma^2})\\
        &=-\frac{1}{\sigma^2}\cdot \exp(\frac{x^2}{2\sigma^2})\cdot B_{\sigma}(x)   \label{p_differential} \\
        &=-\frac{1}{\sigma^2}\cdot \exp(\frac{x^2}{2\sigma^2})\cdot \sum_{i\ge 0}(-1)^i\cdot A_{\sigma}(x+i) \qquad \textit{by \eqref{eq:B_sigma}}.  
\end{align}
    We show $p_\sigma(x)$ is monotonically decreasing on $x$ in $[-1/2,\infty]$ for all $\sigma>0$ first.
    
    \begin{proofof}{Lemma~\ref{lem:mono_on_prob} (1)}
    
    To show $p_{\sigma}(x)$ is decreasing when $x\ge -1/2$, by equation (\ref{p_differential}) we just need to show $B_{\sigma}(x)>0$ for all $x\ge -1/2$.

    Fixing $\sigma\ne 0$, $B_{\sigma}(x)=\sum_{j\ge 1}(-1)^{j-1}\cdot j \cdot \exp(-\frac{(x+j)^2}{2\sigma^2})$, there exists an $x_\sigma$ large enough such that for all $x> x_\sigma$, $B_{\sigma}(x)>0$ due to the exponentially decay of $\exp(-\frac{(x+j)^2}{2\sigma^2})$ and only the first term becoming significant.
    %So if there are some $x$ making $B_{\sigma}(x)\le 0$ we can choose the largest $x$ and denote it by $x_0.$ We are going to show that if such $x_0$ exsits and $x_0>-1/2$ then $B_{\sigma}(x_0)>0.$
    Now $B_{\sigma}(x)$ is an alternative sum of $A_{\sigma}(x+i)$. Hence by alternating series test the sign of $B_\sigma(x)$ is determined by monotonicity of $A_\sigma(x)$. %\hynote{Maybe we need more explanation here.} We have
    \begin{equation}\label{eq:derivative_A}
        \frac{\partial A_{\sigma}(x)}{\partial x}=-\frac{1}{\sigma^2}\cdot(x\cdot A_\sigma(x)+B_{\sigma}(x)).
    \end{equation}
    For $x>\max \{x_\sigma, 0\}$ we have $A_{\sigma}(x)>0,~B_\sigma(x)>0$, hence $\frac{\partial A_{\sigma}(x)}{\partial x}<0$. So $A_{\sigma}(x)$ is monotonically decreasing on $x\ge x_\sigma$. Moreover, if there exists $x$ such that  $A_\sigma(x)-A_\sigma(x+1)\le 0$, saying the largest one is $x_0$, i.e., $x_0\triangleq \max\{x:A_{\sigma}(x+1)\ge A_{\sigma}(x)\}$. %\hynote{the statement maybe need to be re-organised. How about:``Moreover, let  $x_0\triangleq max\{x:A_{\sigma}(x+1)\ge A_{\sigma}(x)\}$."}%\gynote{We can do this because for $x$ very large $A(x)$ is decreasing so there is an upper bound for $x_0$}
    
    First we show that $x_0<\max(x_0,0)$ and hence $x_0<0$. %\hynote{We need more clear statement here.} For $x\ge x_0$ we rewrite \eqref{eq:derivative_A} as
%    Since $x_0$ is the largest one such that $\frac{\partial A_{\sigma}(x)}{\partial x}\le 0$ so for $x>x_0$ we have $A_\sigma(x)>A_{\sigma}(x+1)$. This implies 
    \begin{equation}\label{equ:partial_A}
        -\sigma^2\cdot\frac{\partial A_{\sigma}(x)}{\partial x}=x\cdot A_\sigma(x)+\sum_{i\ge 0}(-1)^{i}\cdot A_{\sigma}(x+i)> x\cdot A_{\sigma}(x)+A_{\sigma}(x)-A_{\sigma}(x+1).
    \end{equation}
    For $s\ge \max(x_0,0)$, since $A_{\sigma}(s)>0$, we simplify \eqref{equ:partial_A} to
    \begin{equation}
    -\sigma^2\cdot\frac{\partial A_{\sigma}(s)}{\partial s}>A_{\sigma}(s)-A_{\sigma}(s+1)\ge 0.
    \end{equation} 
    Hence $\frac{\partial A_{\sigma}(s)}{\partial s}<0$ for $s \ge \max(x_0,0)$, which implies $A_\sigma(s)-A_{\sigma}(s+1)>0$ for $s\ge \max(x_0,0)$. %\hynote{I can't understand why $x_{0}\ge 0$ is not allowed here. In particular inequality above.} 
    Therefore $x_0 < \max(x_0,0)$ which means $x_0< 0$ and $A_{\sigma}(s)$ decreases on $s\ge 0.$
     
   Now we show $x_0< \max(x_0,-1/2)$. For $0\ge s\ge \max(x_0, -1/2)$ we have 
        \begin{align*}
        &A_{\sigma}(s)-A_\sigma(s+1)\\
        &=-\int_{s}^{s+1} \frac{\partial A_{\sigma}(x)}{\partial x}\cdot dx\\
        &>\frac{1}{\sigma^2}\cdot\int_{s}^{s+1}\big( (1+x)\cdot A_\sigma(x)-A_{\sigma}(x+1)\big)\cdot \mathrm{d} x\tag{by Equation \ref{equ:partial_A}}\\
        &=\frac{1}{\sigma^2}\cdot \int_{s}^{0} \big((1+x)\cdot A_\sigma(x)-A_{\sigma}(x+1)\big)\cdot \mathrm{d} x+\frac{1}{\sigma^2}\cdot \int_{0}^{s+1} \big((1+x)\cdot A_\sigma(x)-A_{\sigma}(x+1)\big)\cdot \mathrm{d} x\\
        &\ge \frac{1}{\sigma^2}\cdot\int_{s}^{0}\big( (1+x)\cdot A_\sigma(x+1)-A_{\sigma}(x+1)\big)\cdot \mathrm{d} x+\frac{1}{\sigma^2}\cdot \int_{0}^{s+1}\big( (1+x)\cdot A_\sigma(s+1)-A_{\sigma}(s+1) \big)\cdot\mathrm{d} x \tag{first term uses $A_\sigma(x)\ge A_{\sigma}(x+1)$ for $x\ge s$ and second term uses $A(x)$ decreasing for $x\ge 0$}\\
        &=\frac{1}{\sigma^2}\cdot \int_{s}^{0} x\cdot A_\sigma(x+1)\cdot \mathrm{d} x+\frac{1}{\sigma^2}\cdot \int_{0}^{s+1} x\cdot A_\sigma(s+1)\cdot \mathrm{d} x\\
        &\ge \frac{1}{\sigma^2}\cdot \int_{s}^{0} x\cdot A_\sigma(s+1) \cdot \mathrm{d} x+\frac{1}{\sigma^2}\cdot \int_{0}^{s+1} x\cdot A_\sigma(s+1)\cdot \mathrm{d} x\tag{use $A_{\sigma}(x+1)$ decreasing for $x\ge -1$}\\
%        &>\frac{1}{\sigma^2}\int_{s}^{s+1} x \cdot A_\sigma(x) \mathrm{d} x \tag{since $A_{\sigma}(x) \ge A_{\sigma}(x+1)$ for $x>s$}\\
%        &> \frac{1}{\sigma^2}\int_{s}^{s+1} (1+x)A_\sigma(s+1)-A_{\sigma}(s+1) \mathrm{d} x\\
        &=\frac{1}{\sigma^2} \cdot (s+1/2)\cdot A_{\sigma}(s+1)\ge 0. %\tag{\hynote{lack for $\sigma$.}} 
    \end{align*}
    So if $s\ge \max (-1/2,x_0)$ we have $A_\sigma(s)>A_\sigma(s+1)$. This means $x_0< \max(-1/2,x_0)$ by the definition of $x_0$.
    %$ If $x_0 \ge -1/2$, we take $s=x_0$ and obtain $A_{\sigma}(x_0)-A_{\sigma}(x_0+1)>0$, contradicting to that $x_0$ is the largest $x$ making $A_{\sigma}(x_0)\le A_{\sigma}(x_0+1)$. 
So we must have $x_0 <-1/2.$ Hence $A_{\sigma}(x)$ is monotonically decreasing when $x\ge -1/2$ which implies $B_{\sigma}(x)>0$ and  $p_{\sigma}(x)$ decreasing on $x\ge -1/2$.
%$p_{\sigma}(x)$ is decreasing \textcolor{blue}{strictly} in $[x_0,\infty)$ by equation (\ref{p_differential}). By equation (\ref{p_equal_A}) 
%    \begin{equation}
%        A_{\sigma}(x)=exp(-\frac{x^2}{2\sigma^2})p_{\sigma}(x)
%    \end{equation}
%    and pσ(x)p_{\sigma}(x) is positive and strictly decreasing in [x0,∞][x_0,\infty], and exp(−x22σ2)exp(-\frac{x^2}{2\sigma^2}) decreases strictly for x>0x>0 we have Aσ(x)A_{\sigma}(x) is decreasing strictly in [x0,∞)[x_0,\infty) if x0≥0x_0\ge 0. So by alternating series test we have 
%    \begin{equation}
%        B_{\sigma}(x_0)=\sum_{i\ge 0}(-1)^i A(x_0+i)>0. 
%    \end{equation}
%    So we have x0>0x_0>0 is impossible and 
    \end{proofof}
    
    Then we show $p_{\sigma}(-1/2)$ is strictly increasing on $\sigma>0$.
    
    \begin{proofof}{Lemma \ref{lem:mono_on_prob} (2)}
   First we express $\frac{\partial p_{\sigma}(x)}{\partial \sigma}$ by $B_{\sigma}(x)$.
    \begin{align*}
    \frac{\partial p_{\sigma}(x)}{\partial \sigma}&=\sum_{j\ge 1}(-1)^{j-1}\cdot\frac{2 j\cdot (2x+j)}{2\sigma^3}\cdot \exp(\frac{-j\cdot(2x+j)}{2\sigma^2})\\
        &=\frac{1}{\sigma^3}\cdot \exp(\frac{x^2}{2\sigma^2})\cdot \sum_{j\ge 1}(-1)^{j-1}\cdot(2x\cdot j+j^2)\cdot \exp(-\frac{(x+j)^2}{2\sigma^2})\\
        &=\frac{1}{\sigma^3}\cdot \exp(\frac{x^2}{2\sigma^2})\cdot \big(2x\cdot B_{\sigma}(x)+\sum_{j\ge 1}(-1)^{j-1}\cdot j^2\cdot \exp(-\frac{(x+j)^2}{2\sigma^2})\big).
        %&=\frac{1}{\sigma^3}exp(\frac{x^2}{2\sigma^2})(2x B_{\sigma}(x)+B_{\sigma}(x)+2\sum_{i\ge 1}(-1)^i B_{\sigma}(x+i)).\tag{use $j^2=\sum_{i=1}^{j}2i-1$}\\
    \end{align*}
Also, we can express the last term as alternative summation of $B_\sigma(x+i).$
\begin{align*}
    &\sum_{j\ge 1}(-1)^{j-1}\cdot j^2\cdot \exp(-\frac{(x+j)^2}{2\sigma^2})\\
    &=\sum_{j\ge 1}(-1)^{j-1}\cdot \sum_{1\le i\le j}(2i-1)\cdot \exp(-\frac{(x+j)^2}{2\sigma^2})\tag{use $j^2=\sum_{i=1}^{j}2i-1$}\\
    &=\sum_{i\ge 1}\sum_{k=j-i\ge 0}(-1)^{k+i-1}\cdot 2i\cdot \exp(-\frac{(x+k+i)^2}{2\sigma^2})+\sum_{j\ge 1}(-1)^{j-1}\cdot \sum_{1\le i\le j}(-1)\cdot \exp(-\frac{(x+j)^2}{2\sigma^2})\\
    &=2\sum_{k\ge 0}\sum_{i\ge 1}(-1)^{k+i-1}\cdot i\cdot \exp(-\frac{(x+k+i)^2}{2\sigma^2})-B_{\sigma}(x)\\
    &=2\sum_{k\ge 0}(-1)^{k}\cdot B_{\sigma}(x+k) - B_{\sigma}(x).
\end{align*}
Hence we have
\begin{equation}\label{equ:partial_p}
    \frac{\partial p_{\sigma}(x)}{\partial \sigma}=\frac{1}{\sigma^3}\cdot \exp(\frac{x^2}{2\sigma^2})\cdot (2x\cdot B_{\sigma}(x)+B_{\sigma}(x)+2\cdot \sum_{i\ge 1}(-1)^i\cdot B_{\sigma}(x+i)).
\end{equation}
%    By some simple computations.
%    \begin{equation}
%        \sum_{j\ge 1}(-1)^{j-1}(2xj+j^2)exp(\frac{(x+j)^2}{2\sigma^2})=2x B_{\sigma}(x)+B_{\sigma}(x)+2\sum_{i\ge 1}(-1)^i B_{\sigma}(x+i).
%    \end{equation}
    Now we consider the monotonicity of $B_\sigma(x)$ and
    \begin{equation}
    \begin{aligned}
        \frac{\partial B_{\sigma}(x)}{\partial x}&=-\frac{1}{\sigma^2}\sum_{j\ge 1}(-1)^{j-1}\cdot (x+j)\cdot j\cdot \exp(-\frac{(x+j)^2}{2\sigma^2})\\
        &=-\frac{1}{\sigma^2} \left( x\cdot B_{\sigma}(x)+B_{\sigma}(x)+2\sum_{i\ge 1}(-1)^i\cdot B_{\sigma}(x+i) \right).
    \end{aligned}
    \end{equation}
    Similar to the proof of Lemma \ref{lem:mono_on_prob} (1) that $A_{\sigma}(x)$ is monotonically decreasing on $[0,+\infty)$, we have $\frac{\partial B_{\sigma}(x)}{\partial x}<0$ for $x\ge 1$. %\xcnote{Do not know why --- we haven't proved that $B$ is monotonically decreasing.}\gynote{proof added now}
    
    Let $x_0$ be the largest $x$ so that  $B_\sigma(x)-B_\sigma(x+1)\le 0$ as well. We will show that $x_0<\max(x_0,1)$. For $s\ge \max(x_0,1)$ we have
    \begin{align*}
        \frac{\partial B_{\sigma}(s)}{\partial s}
        &=-\frac{1}{\sigma^2}\cdot \left( x\cdot B_{\sigma}(x)+B_{\sigma}(x)+2\sum_{i\ge 1}(-1)^i\cdot B_{\sigma}(x+i) \right)\\
        &\le -\frac{1}{\sigma^2}\cdot \left( 2B_{\sigma}(s)+2\sum_{i\ge 1}(-1)^i\cdot B_{\sigma}(s+i) \right)\tag{use $B_{\sigma}(s)>0$ for $s>-1/2$}\\
        &=-\frac{2}{\sigma^2}\cdot\sum_{i=0}(-1)^i\cdot B_{\sigma}(s+i)
        < 0.     \tag{by the definition of $x_0$ and the alternative criterion}   \end{align*}
    Hence we have $\frac{\partial B_{\sigma}(s)}{\partial s}<0$ for $s\ge \max(x_0,1)$ which implies $B_{\sigma}(s)-B_{\sigma}(s+1)>0$ and $x_0< \max(x_0,1).$ And therefore we have $x_0<1$. 
    
    And then we show $x_0<\max(x_0,1/2)$. For $1\ge s\ge \max (x_0,1/2)$ %\xcnote{Is this a typo?}\gynote{fixed}
    \begin{align*}
    &B_\sigma(s)-B_{\sigma}(s+1)\\
    &=-\int_{s}^{s+1} \frac{\partial B_{\sigma}(x)}{\partial x}\cdot \mathrm{d} x\\
    &=\frac{1}{\sigma^2}\cdot \int_{s}^{s+1}\big( x\cdot B_{\sigma}(x)+B_{\sigma}(x)+2\sum_{i\ge 1}(-1)^i\cdot B_{\sigma}(x+i)\big)\cdot \mathrm{d} x\\
    &>\frac{1}{\sigma^2}\cdot\big(\int_{s}^{1} (1+x)\cdot B_{\sigma}(x)-2B_{\sigma}(x+1)\cdot\mathrm{d} x+\int_{1}^{s+1} (1+x)\cdot B_{\sigma}(x)-2B_{\sigma}(x+1)\cdot\mathrm{d} x\big)\\
    &\ge \frac{1}{\sigma^2}\cdot\big(\int_{s}^{1} (x-1)\cdot B_{\sigma}(x+1)\cdot\mathrm{d} x+\int_{1}^{s+1} (1+x)\cdot B_{\sigma}(s+1)-2B_{\sigma}(s+1)\cdot \mathrm{d} x\big)\tag{use $B_\sigma(x)\ge B_\sigma(x+1)$ for $x\ge s$ and  $B_\sigma(x)$  decreasing for $x\ge 1$.}\\
    &\ge \frac{1}{\sigma^2}\cdot\int_{s}^{s+1} (x-1)\cdot B_{\sigma}(s+1)\cdot \mathrm{d} x\tag{$B_\sigma(x+1)$ is decreasing for $x\ge 0$.}\\
%    &\gynote{\text{I filled the gap in the proof}}\\
    &=\frac{1}{\sigma^2}\cdot (s-1/2)\cdot B_\sigma(s+1)\ge 0.
    \end{align*}
 %   For the last inequality, first term we use $B(x)\ge B(x+1)$ for $x\ge s$ and $B(x+1)$ decreasing for $x\ge 0$ and second term uses $B(s)$  decreasing for $x\ge 1$. 
    By a similar argument we have $B_\sigma(x)-B_\sigma(x+1)>0$ for $x\ge 1/2.$ %Back to equation (\ref{p_sigma_-1/2}), by alternative series test, we have $\frac{\partial p_{\sigma}(-1/2)}{\partial \sigma}\le 0.$

    And if $x\ge 1/2$ we have 
    \begin{equation}
         \frac{\partial p_{\sigma}(x)}{\partial \sigma}\ge \frac{3}{2\sigma^3}\cdot \exp(\frac{x^2}{2\sigma^2})\cdot (2\cdot \frac{1}{2} \cdot B_{\sigma}(x)+B_{\sigma}(x)+2\sum_{i\ge 1}(-1)^i\cdot B_{\sigma}(x+i))>0,
    \end{equation}
    which means $p_{\sigma}(x)$ is monotonically increasing strictly on $\sigma>0$ for any fixed $x\ge 1/2.$
%    By similar argument we know that for $x\ge 1$, $B_\sigma(x)$ is monotonically decreasing. So we know that $\frac{\partial p_{\sigma}(x)}{\partial \sigma}>0$ if $x\ge 1$ no matter how much $\sigma$ is.
    \end{proofof}
%}
\begin{proofof}{Lemma \ref{lem:mono_on_prob}(3)}
        Plugging $x=-1/2$ into Equation \eqref{equ:partial_p} we have
    \begin{equation}\label{p_sigma_-1/2}
        \frac{\partial p_{\sigma}(-1/2)}{\partial \sigma}=\frac{3}{2\sigma^3}\cdot \exp(\frac{x^2}{2\sigma^2})\cdot 2\sum_{i\ge 1}(-1)^i\cdot B_{\sigma}(x+i).
    \end{equation}
    And by the monotonicity of $B_{\sigma}(x+1)$ for $x\ge -1/2$ proved above, we have $\frac{\partial p_{\sigma}(-1/2)}{\partial \sigma}<0.$
\end{proofof}

\begin{proofof}{Lemma \ref{lem:mono_on_prob}(4)}
As mentioned in \cite{LiuSS22}, let $u=\exp{(\frac{-1}{2\sigma^{2}})}$ and $v=\sqrt{-1}\cdot \exp{(\frac{-x}{2\sigma^{2}})}$ which means $|u|<1$ and $v\neq 0$. By Jacobi triple product identity (reference e.g. \cite{adr1965_Jacobi}) we have
\begin{equation}
\begin{aligned}
        r_{\sigma}(x)&=\sum_{j=-\infty}^{\infty}u^{j^{2}}\cdot v^{2j}\\
        &=\prod\limits_{j=1}^{\infty}(1-u^{2j})\cdot (1+u^{2j-1}v^{2})\cdot (1+u^{2j-1}v^{-2})\\
        &=\prod\limits_{j=1}^{\infty}(1-e^{-j/\sigma^{2}})\cdot (1-e^{-(2j+2x-1)/(2\sigma^{2})})\cdot(1-e^{-(2j-2x-1)/(2\sigma^{2})}) \label{eq_r_sigma(x)}.
\end{aligned}
\end{equation}
Besides, \begin{align*}
       (1-e^{-(2j+2x-1)/(2\sigma^{2})})\cdot(1-e^{-(2j-2x-1)/(2\sigma^{2})})
    &=1+e^{-2\cdot(2j-1)/(2\sigma^2)}-e^{-(2j-1)/(2\sigma^2)}(e^{-\frac{x}{\sigma^2}}+e^{\frac{x}{\sigma^2}})\\
    &\le 1+e^{-2\cdot(2j-1)/(2\sigma^2)}-2e^{-(2j-1)/(2\sigma^2)}.
\end{align*}
Therefore $r_{\sigma}(x)\le r_\sigma(0)$ for $x\in [-0.5,0.5].$ And it is easy to verify
\begin{equation}
\begin{aligned}
    \ln r_\sigma(0)&\le \sum_{j=1}^\infty \ln(1-e^{-j/\sigma^2})\le -\sum_{j=1}^\infty e^{-j/\sigma^2}=-\Theta(\sigma^2).
\end{aligned}
\end{equation}
Hence for $\sigma\rightarrow\infty$ we have $r_\sigma(0)\le e^{-\Omega(\sigma^2)}\rightarrow 0.$
\end{proofof}

%\input{supplement}
%%%%%%%%%%%%%%%%%%%%%%%%%%%%%%%%%%%%%%%%%%

%%%%%%%%%%%%%%%%%%%%%%%%%%%%%%%%%%%%%%%%%%%%

\end{document}